\documentclass[12pt]{article}

\usepackage[makeindex,split,idxcommands]{splitidx}
\usepackage{bbm}
\usepackage{amssymb}
\usepackage{amsthm}
\usepackage{amsfonts}
\usepackage{graphicx}
\usepackage{verbatim}
\usepackage{enumerate}
\usepackage[intlimits]{amsmath}
\usepackage{fullpage}
\usepackage[numbers]{natbib}

\makeindex

\newindex[Notation Index]{symbols}
\newindex[Subject Index]{notions}
\newif\ifhyper\IfFileExists{hyperref.sty}{\hypertrue}{\hyperfalse}
\ifhyper\usepackage{hyperref}\fi

\newif\ifdraft
\drafttrue
\numberwithin{equation}{section}
\numberwithin{figure}{section}

\newtheorem{theorem}{Theorem}
\numberwithin{theorem}{section}
\newtheorem{corollary}[theorem]{Corollary}
\newtheorem{lemma}[theorem]{Lemma}

\newtheorem{definition}[theorem]{Definition}
\theoremstyle{definition}
\newtheorem{example}[theorem]{Example}
\newtheorem{remark}[theorem]{Remark}
\newtheorem{claim}{Claim}

\usepackage[english]{cleveref}

\newcommand{\E}{\mathbb{E}}

\newcommand{\G}{\mathbb{G}}

\newcommand{\R}{\mathbb{R}}

\newcommand{\Z}{\mathbb{Z}}
\newcommand{\N}{\mathbb{N}}

\newcommand{\EEE}{\mathcal{E}}

\newcommand{\KK}{\mathcal{K}}
\newcommand{\I}{\mathbb{I}}
\newcommand{\PPP} {{\mathbb P}}

\newcommand{\FF}{{\bf F}}

\newcommand{\hh}{\mathbb{H}}
\def\FF{\mathcal F}
\def\KK{\mathcal K}
\def\CC{\mathcal C}
\def\LL{\mathcal L}

\def\eval{\mathfrak{eval}}
\def\hhh{\mathfrak{h}}

\def\tw{\mathrm{tv}}

\def \P {{\mathcal P}}

\def \_reg {\rightarrow_{\bf reg}}

\def\maxdeg/{\Delta}

\def\ee{\mathfrak{e}}
\def\eE{\mathfrak{E}}

\allowdisplaybreaks

\makeatletter
\newcommand{\subjclass}[2][1991]{%
	\let\@oldtitle\@title%
	\gdef\@title{\@oldtitle\footnotetext{#1 \emph{Mathematics subject classification.} #2}}%
}
\newcommand{\keywords}[1]{%
	\let\@@oldtitle\@title%
	\gdef\@title{\@@oldtitle\footnotetext{\emph{Key words and phrases.} #1.}}%
}
\makeatother

\begin{document}
\newcommand{\na}{\mathbf{a}}
\newcommand{\nb}{\mathbf{b}}
\newcommand{\nc}{\mathbf{c}}
\newcommand{\inj}{\mathrm{inj}}
\newcommand{\ind}{\mathrm{ind}}
\newcommand{\Sym}{\mathrm{Sym}}
\newcommand{\Pd}{\mathrm{Pd}}
\newcommand{\du}{\mathrm{d}}

\title{\textbf{Limits of CSP Problems and Efficient Parameter Testing}}
\author{Marek Karpinski\thanks{Dept. of Computer Science and the Hausdroff Center for Mathematics, University of Bonn. Supported in part by DFG grants, the Hausdorff grant EXC59-1. Research partly supported by Microsoft Research New England. E-mail: \textrm{marek@cs.uni-bonn.de}}
\and 
Roland Mark\'o\thanks{Hausdorff Center for Mathematics, University of Bonn. Supported in part by a Hausdorff scholarship. E-mail: \textrm{roland.marko@hcm.uni-bonn.de}}}


\date{}
\maketitle
\begin{abstract}
{\normalsize We present a unified framework on the limits of \emph{constraint satisfaction problems} (CSPs) and efficient parameter testing which depends only on array exchangeability and the method of cut decomposition without recourse to the weakly regular partitions. In particular, we formulate and prove a representation theorem for compact colored $r$-uniform directed hypergraph ($r$-graph) limits, and apply this to $r$CSP limits. We investigate the sample complexity of testable $r$-graph parameters, we discuss the generalized ground state energies and demonstrate that they are efficiently testable.  }
\end{abstract}

\section{Introduction} \label{sec:intro}

We study the limits and efficient parameter testing properties of \emph{Maximum Constraint Satisfaction Problems} of arity~$r$ (MAX-$r$CSP or $r$CSP for short), c.f. e.g., \cite{AVKK2}. These two topics, limiting behavior and parameter estimation, are treated in the paper to a degree separately, as they require a different set of ideas and could be analyzed on their own right. The establishment of the underlying connection between convergence and testability is one of the main applications of the limit theory of dense discrete structures, see \cite{BCL}, \cite{BCL2}.

In the first part of the paper we develop a general framework 
for the above CSP problems which depends only on 
the principles of the array exchangeability without a recourse to the weakly regular partitions used hitherto in the general graph and hypergraph settings. Those fundamental techniques and results were worked out in a series of papers by  Borgs, Chayes, Lov\'asz, S\'os, Vesztergombi and Szegedy \cite{BCL},\cite{BCL2},\cite{LSzlim}, and \cite{LSztest} for graphs including connections to statistical physics and complexity theory, and were subsequently extended to hypergraphs by Elek and Szegedy \cite{ESz} via the ultralimit method. The central concept of $r$-graph convergence is defined through convergence of sub-$r$-graph densities, or equivalently through weak convergence of probability measures on the induced sub-$r$-graph yielded by uniform node sampling. Our line of work particularly relies on ideas presented in \cite{DJ} by Diaconis and Janson, where the authors shed some light on the correspondence between combinatorial aspects (that is, graph limits via weak regularity) and the probabilistic viewpoint of sampling: Graph limits provide an infinite random graph model that has the property of exchangeability. The precise definitions, references and results will be given in Section \ref{sec:lim}, here we only formulate our main contribution informally: We prove a representation theorem for compact colored $r$-uniform directed hypergraph limits. This says that every limit object in this setup can be transformed into a measurable function on the $(2^r-2)$-dimensional unit cube that takes values from the probability distributions on the compact color palette, see Theorem \ref{ch2:reprmain} below. This extends the result of Diaconis and Janson  \cite{DJ}, and of Lov\'asz and Szegedy \cite{LSzcom}. As an application, the description of the limit space of $r$CSPs is presented subsequent to the aforementioned theorem.
  
The second part of the paper, Sections \ref{ch3:sec.lim} to \ref{sec:appl}, is dedicated to the introduction of a notion of efficient parameter testability of $r$-graphs and $r$CSP
problems. We use the limit framework from the first part of the paper to formulate several results on it, which are proved with the aid of the cut decomposition method. We set our focus especially on parameters called \emph{ground state energies} and study variants of them. These are in close relationship with MAX-$r$CSP problems, our results can be regarded as the continuous generalization of the former. We rely on the notion of parameter testing and sample complexity, that was introduced by Goldreich, Goldwasser, and Ron \cite{GGR} and was employed in the graph limit theory in \cite{BCL2}. A graph parameter is testable in the sense of \cite{BCL2}, when its value is estimable through a uniform sampling process, where the sample size only depends on the desired error gap, see Definition \ref{ch3:deftest} below for the precise formulation. The characterization of the real functions on the graph space was carried out in \cite{BCL2}, the original motivation of the current paper was to provide an analogous characterization for efficiently testable parameters. These latter are parameters, whose required sample size for the estimation is at most polynomial in the multiplicative inverse of the error. 

The investigation of such parameters has been an active area of research for the finite setting in complexity theory. The method of exhaustive sampling in order to approximately solve NP-hard problems was proposed by Arora, Karger and Karpinski  \cite{AKK2}, their upper bound on the required sample size was still logarithmically increasing in the size of the problem. The approach in \cite{AKK2} enabled the employment of linear programming techniques. Subsequently, the testability of MAX-CUT was shown in \cite{GGR}, explicit upper bounds for the sample complexity in the general boolean MAX-$r$CSP were given by Alon, F. de la Vega, Kannan and Karpinski \cite{AVKK2} using cut decomposition of $r$-arrays and sampling, that was inspired by the introduction of weak regularity by Frieze and Kannan \cite{FK}. In \cite{AVKK2} and \cite{FK}, the design of polynomial time approximation schemes (PTAS) in order to find not only an approximate value for MAX-$r$CSP, but also an assignment to the base variables that certify this value was an important subject, we did not pursue the generalization regarding this aspect in the current work. The achievements of these two aforementioned contributions turned out to be highly influential, and took also a key role in the first elementary treatment of graph limits and in the definition of the $\delta_\square$-metric in \cite{BCL} that defines an equivalent topology on the limit space to the subgraph density convergence.

The best currently known upper bound on the sample complexity of MAX-$r$CSP is $\mathcal{O}(\varepsilon^{-4})$, and has been shown by Mathieu and Schudy \cite{MS}, see also Alon, F. de la Vega, Kannan and Karpinski \cite{AVKK2}. Unfortunately, the approach of \cite{MS} does not seem to have a natural counterpart in the continuous setting, although one can use their result on the \emph{sample} to achieve an improved upper bound on the sample complexity.  We mention that for the original problem we do not aim to produce an assignment for MAX-$r$CSP, or a partition for the ground state energy whose evaluation is nearly optimal as opposed to the above works, although we believe this could be done without serious difficulties.  

Our contribution in the second part of the paper is the following. By employing a refined version of the proof of the main result of  \cite{AVKK2} adapted to the continuous setting we are able to prove the analogous efficient testability result  for a general finite state space for ground state energies, see Theorem \ref{ch4:main} in Section \ref{ch4:sec:gse} for a precise formulation. Among the applications of this development we analyze the testability of the microcanonical version of ground state energies providing the first explicit upper bounds on efficiency. For the finite version a similar question was investigated by F. de la Vega, Kannan and Karpinski \cite{VKK} by imposing additional global constraints (meaning a finite number of them with unbounded arity).  Furthermore, the continuous version of the quadratic assignment problem is treated the first time in a sample complexity context, this subject is related to the recent contributions to the topic of approximate graph isomorphism and homomorphism, see \cite{LRS} and \cite{AKKV}.    

\subsection{Outline of the paper}
The organization of this paper is as follows. In Section \ref{sec:lim} we develop the limit theory for $\mathcal K$-decorated $r$-uniform directed hypergraphs with reference to previously known special (and in some way generic) cases, and use the representation of the limit to describe the limit space of $r$CSP problems. In Section \ref{ch3:sec.lim} the basic notion of efficiency in context of parameter testing is given with some additional examples. The subsequent Section \ref{ch4:sec:gse} contains the proof of Theorem \ref{ch4:main} regarding ground state energies of $r$-graphons, and in the following Section \ref{sec:appl} some variants are examined, in particular microcanonical energies and the quadratic assignment problem. We summarize possible directions of further research in Section \ref{sec:fr}.

\section{Limit theory and related notions}\label{sec:lim}

We will consider the objects called \emph{$r$CSP formulas} that are used to define instances of the decision and optimization problems called $r$CSP and MAX-$r$CSP, respectively. In the current framework a formula consists of a variable set and a set of boolean or integer valued functions. Each of these functions is defined on a subset of the variables, and the sets of possible assignments of values to the variables are uniform.
Additionally, it will be required that each of the functions, which we will call \emph{constraints} \sindex[notions]{constraint} in what follows, depend exactly on~$r$ of the variables. 


For the treatment of an $r$CSP (of a MAX-$r$CSP) corresponding to a certain formula we are required to simultaneously evaluate all the constraints of the formula by assigning values to each of the variables in the variable set. If we deal with an $r$CSP optimization problem on some combinatorial structure, say on graphs, then the formula corresponding to a certain graph has to be constructed according to the optimization problem in question. The precise definitions will be provided next.

Let $r\geq 1$, $K$ be a finite set, and $f$ be a boolean-valued function $f\colon  K^r \to \{0,1\}$ on $r$ variables (or equivalently $f \subseteq K^r$).
We call $f$ a \emph{constraint-type on $K$ in  $r$ variables}\sindex[notions]{constraint-type}, $\mathcal C=\mathcal C(K,r)$\sindex[symbols]{c@$\mathcal C(K,r)$} denotes the set of all such objects.

\begin{definition}[$r$CSP formula]  \label{ch2:defcsp}
	Let $V=\{x_1,x_2, \dots, x_n\}$ be the set of variables, $x_e=(x_{e_1}, \dots , x_{e_r}) \in V^r$ and $f$ a constraint-type on $K$ in $r$ variables.  We call an $n$-variable function $\omega=(f;x_e)\colon K^V \to \{0,1\}$ with $\omega(l_1,\dots ,l_n)=f(l_{e_1}, \dots,l_{e_r})$ a constraint on $V$ in $r$ variables determined by an $r$-vector of constrained variables and a constraint type. 
	
	We call a collection $F$ of constraints on $V(F)=\{x_1,x_2, \dots, x_n\}$ in $r$ variables of type $\mathcal C(K,r)$ for some finite $K$  an $r$CSP formula.
\end{definition}
\sindex[notions]{r@$r$CSP formula}

Two constraints $(f_1;x_{e_1})$ and $(f_2;x_{e_2})$ are said to be equivalent\sindex[notions]{constraint!equivalence} if they constrain the same $r$ variables, and their evaluations coincide, that is, whenever there exists a $\pi \in S_r$ such that $e_1=\pi(e_2)$ (here $\pi$ permutes the entries of $e_2$) and $f_1=\hat \pi (f_2)$, where $[\hat \pi(f)](l)=f(\pi(l))$. Two formulas $F_1$ and $F_2$ are equivalent\sindex[notions]{r@$r$CSP formula!equivalence} if there is a bijection $\phi$ between their variable sets such that there is a one-to-one correspondence between the constraints of $F_1$ and $F_2$ such that the corresponding pairs $(f_1;x_{e_1}) \in F_1$ and $(f_2;x_{e_2}) \in F_2$ satisfy $(f_1; \phi(x_{e_1})) \equiv (f_2;x_{e_2})$.

In the above definition the set of states of the variables in $V(F)$ denoted by $K$ is not specified for each formula, it will be considered as fixed similar to the dimension $r$ whenever we study a family of $r$CSPs. We say that $F$ is symmetric\sindex[notions]{r@$r$CSP formula!symmetric}, if it contains only constraints with constraint-types which are invariant under the permutations of the constrained variables. When we relax the notion of the types to be real or $\KK$-valued functions on $K^r$ with $\KK$ being a compact space, then we speak of weighted $r$CSP formulas\sindex[notions]{r@$r$CSP formula!weighted}.

The motivation for the name CSP formula is immediately clear from the notation used in \Cref{ch2:defcsp} if we consider constraints to be satisfied at some point in $K^n$, whenever they evaluate to $1$ there.
Most problems defined on these objects ask for parameters that are, in the language of real analysis, global or conditioned extreme values of the objective function given by an optimization problem and a formula. A common assumption is that equivalent formulas should get the same parameter value.

\begin{definition}[MAX-$r$CSP]  \label{ch2:defcsp2}
	
	Let $F$ be an $r$CSP formula over a finite domain $K$. Then the MAX-$r$CSP value of $F$ is given by 
	
	\begin{align}
	\mathrm{MAX-}r\mathrm{CSP}(F)=\max_{l \in K^{V(F)}} \sum_{\omega=(f;x_e) \in F}\omega(l),
	\end{align}
	and $F$ is satisfiable, if $\mathrm{MAX-}r\mathrm{CSP}(F)=|\{ \,\omega \mid \omega=(f;x_e) \in F\, \}|$.
	
\end{definition}
\sindex[notions]{optimization problem!MAX-$r$CSP}\sindex[notions]{r@$r$CSP formula!satisfiable}
Such problems are for example
MAX-CUT, fragile MAX-$r$CSP, MAX-$3$-SAT, and Not-All-Equal-$3$-SAT, where only certain constraint types are allowed for instances, or MAX-BISECTION, where additionally only specific value assignments are permitted in the above maximization.
In general, formulas can also be viewed as directed $r$-graphs, whose edges are colored with constraint types (perhaps with multiple types), and we will exploit this representation in our analysis.

Typically, we will not store and recourse to an $r$CSP formula $F$ as it is given by its definition above, but we will only consider the $r$-array tuple $(F^z)_{z \in K^r}$, where 
\begin{align}\label{ch2:evaldef}
F^z(e)=\sum_{\phi \in S_r}\sum_{(f;x_{\phi(e)}) \in F} f(z_{\phi(1)}, \dots, z_{\phi(r)})
\end{align}
for each $e \in [n]^r$. The data set $(F^z)_{z \in K^r}$ is called the \emph{evaluation representation of $F$}\sindex[notions]{r@$r$CSP formula!evaluation representation}, or short $\eval(F)$\sindex[symbols]{e@$\eval(F)$}, we regard $\eval(F)$ as a parallel colored (with colors from $[q]^r$) multi-$r$-graph, see below. We impose a boundedness criteria on CSPs that will apply throughout the paper, that means we fix $d \geq 1$ for good, and require that $\|F^z\|_\infty \leq d$ for every $z \in K^r$ and CSP formula $F$ with $\eval(F)=(F^z)_{z \in K^r}$ in consideration. We note that for each $z \in K^r$, $e \in [n]^r$ and $\phi \in S_r$ we have the symmetry $F^z(e)=F^{z_{\phi(1)}, \dots, z_{\phi(r)}}(e_{\phi(1)}, \dots, e_{\phi(r)})$, also, on the diagonal $F^z$ is $0$.

The main motivation for what follows in the current section originates from the aim to
understand the long-range behavior of a randomly evolving $r$CSP formula together with the value of the corresponding MAX-$r$CSP by making sense of a limiting distribution. This task is equivalent to presenting a structural description of $r$CSP limits analogous to the graph limits of \cite{LSzlim}.

The convergence notion should agree with parameter estimation via sampling. In this setting we pick a set of variables of fixed size at random from the constrained set $V(F)$ of an $r$CSP formula $F$ defined on a large number of variables, and ask for all the constraints in which the sampled variables are involved 
and no other, this is referred to as the induced subformula on the sample. Then we attempt to produce some quantitative statement about the parameter value of the original formula by relying only on
the estimation of the corresponding value of the parameter on a subformula, see \Cref{ch3:deftest}. 

Having formally introduced the notion of $r$CSP formulas and MAX-$r$CSP, we proceed to the outline of the necessary notation and to the analysis of the limit behavior regarding the colored hypergraph models that are used to encode these formulas.

\subsection{Limits of \texorpdfstring{$\mathcal K$}{K}-colored \texorpdfstring{$r$}{r}-uniform directed hypergraphs}
\label{ch2:sec.limdirect}

Let $\KK$ be a compact Polish space\sindex[notions]{Polish space} and $r\geq 1$ an integer. Recall that a space $\KK$ is called Polish if it is a separable completely metrizable topological space. In what follows we will consider  the limit space of $\KK$-colored $r$-uniform directed hypergraphs, or with different words $r$-arrays with non-diagonal entries from $\KK$, and the diagonal entries are occupied by a special element which also can be in $\KK$, but in general this does not have to be the case.

The basic content of the current subsection starts with the general setting given above, CSPs will be considered as a special case in this topic whose limit characterization will be derived at the end. Some of the basic cases are already settled regarding the representation of the limits, we refer to Lov\'asz and Szegedy  \cite{LSzlim}, \cite{LSzcom}, \cite{L} for the $r=2$, general $\KK$, undirected case, to Elek and Szegedy \cite{ESz} for the general $r$, $\KK=\{0,1\}$, undirected case; and Diaconis and Janson \cite{DJ} for $r=2$, $\KK=\{0,1\}$, directed and undirected case. These three approaches are fundamentally different in their proof methodology (they rely on weak regularity, ultralimits, and exchangeability principles respectively) and were further generalized or applied by Zhao \cite{Z} to general $r$; respectively by Aroskar \cite{AA} to the directed case; respectively by Austin \cite{Au} general $r$ and by Janson \cite{J} to the directed case where the graph induces a partial order on the vertex set.

\paragraph{Definition of convergence}

Let $\mathcal C$ denote space $C(\KK)$ of continuous functionals on $\mathcal K$, and let $\mathcal F \subset \mathcal C$ be a countable generating set with $\|f\|_\infty \leq 1$ for each  $f \in \FF$ , that is, the linear subspace generated by $\mathcal F$ is dense in $\mathcal C$ in the $L^\infty$-norm. 

Denote by $\Pi(S)=\Pi^r(S)$ the set of all unlabeled $S$-decorated directed $r$-uniform hypergraphs for some arbitrary set $S$, where we will suppress $r$ in the notation, when it is clear which $r$ is meant (alternatively, $\Pi(S)$ denotes the isomorphism classes of the node labeled respective objects). The set $\Pi_k(S)$ denotes the elements of $\Pi(S)$ of vertex cardinality $k$.
Let $\G(k,F)$ denote the random induced subformula of $F$ on the set $S \subset V(F)$ that is chosen uniformly among the subsets of $V(F)$ of cardinality $k$. We define the homomorphism densities\sindex[notions]{homomorphism density} next.

\begin{definition} \label{ch2:defdens}
	Let $\KK$ be an arbitrary set or space, and $C(\KK)$ be the set of continuous functionals on $\KK$. If for some $r\geq 1$ $F \in \Pi^r(C(\KK))$ is a uniform directed graph with $V(F)=[k]$ and $G \in \Pi^r(\KK)$, then the homomorphism density of $F$ in $G$ is defined as
	\begin{equation}\label{ch2:defdens4}
	t(F,G)=\frac{1}{|V(G)|^k}\sum_{\phi\colon [k] \to V(G)} \prod_{i_1, \dots, i_r=1}^{k} F(i_1,\dots,i_r) (G(\phi(i_1), \dots, \phi(i_r))).
	\end{equation}
	\sindex[symbols]{t@$t(F,G)$}\sindex[notions]{homomorphism density!graph}
	
	The injective homomorphism density\sindex[notions]{homomorphism density!injective} $t_\inj(F,G)$\sindex[symbols]{t@$t_\inj(F,G)$} is defined similarly, with the difference that the average of the products is taken over all injective $\phi$ maps (normalization changes accordingly).

\end{definition}

In the special case when $\KK$ is finite we can associate to the elements of $\Pi(\KK)$ functions in $\Pi(C(\KK))$ through replacing the edge colors in $\KK$ by the corresponding indicator functions. 
%
Note that this way if $F,G \in \Pi^{r}$, then $t_\inj(F,G)=\PPP(\G(k,G)=F)$. 

Let the map $\tau$ be defined as $\tau(G)=(t(F,G))_{F \in \Pi(\mathcal F)} \in [0,1]^{\Pi(\mathcal F)}$ for each $G \in \Pi(\mathcal K)$. 
We set $\Pi(\KK)^*=\tau(\Pi(\KK))\subset [0,1]^{\Pi(\FF)}$, and $\overline{\Pi(\KK)^*}$ to the closure of $\Pi(\KK)^*$. Also, let  $\Pi(\KK)^+=\{ \,(\tau(G),1/|V(G)|) \mid G \in \Pi(\KK)\, \} \subset [0,1]^{\Pi(\FF)} \times [0,1]$\sindex[symbols]{p@$\Pi(\KK)^*$,$\overline{\Pi(\KK)^*}$,$\Pi(\KK)^+$,$\overline{\Pi(\KK)^+}$}, and let $\overline{\Pi(\KK)^+}$ be the closure of $\Pi(\KK)^+$.
The function $\tau^+(G)=(\tau(G),1/|V(G)|)$\sindex[symbols]{t@$\tau(G)$,$\tau^+(G)$} will be useful for our purposes, because,  opposed to $\tau$, it is injective, which can be verified easily. For any $F \in \Pi(\FF)$ the function $t(F,.)$ on $\Pi(\KK)$ can be uniquely continuously extended to a function $t(F,.)$ on $\overline{\Pi(\KK)^+}$, this is due to the compactness of $[0,1]^{\Pi(\FF)} \times [0,1]$. For an element $\Gamma \in \overline{\Pi(\KK)^+} \setminus \Pi(\KK)^+$, let $t(F,\Gamma)$\sindex[symbols]{t@$t(F,\Gamma)$} for $F \in \Pi(\FF)$ denote the real number in $[0,1]$ that is the coordinate of $\Gamma$ corresponding to $F$.

The functions $\tau_\inj(G)$\sindex[symbols]{t@$\tau_\inj(G)$,$\tau^+_\inj(G)$} and $\tau^+_\inj(G)$, and the sets $\Pi_{\inj}(\KK)=\tau_{\inj}(\Pi(\KK))$ and $\Pi_{\inj}(\KK)^+$\sindex[symbols]{p@$\Pi_{\inj}(\KK)$,$\Pi_{\inj}(\KK)^+$} are defined analogously.
It was shown in \cite{LSzlim} that 
\begin{equation} \label{ch2:eqinj}
|t_{_\inj}(F,G)-t(F,G)|\leq \frac{|V(F)|^2\|F\|_\infty}{2|V(G)|}
\end{equation}
for any pair $F \in \Pi(\CC)$ and $G\in\Pi(\KK)$.


The precise definition of convergence will be given right after the next theorem which is analogous to a result of \cite{LSzcom}. 

\begin{theorem} \label{ch2:thm1} Let $(G_n)_{n=1}^\infty$ be a random sequence in $\Pi(\KK)$ with $|V(G_n)|$ tending to infinity in probability. Then the following are equivalent.
	\begin{itemize}
		\item[(1)] The sequence $(\tau^+(G_n))_{n=1}^\infty$ converges in distribution in $\Pi(\KK)^+$.
		\item[(2)] For every $F \in \Pi(\FF)$, the sequence $(t(F,G_n))_{n=1}^\infty$ converges in distribution.
		\item[(3)] For every $F \in \Pi(\CC)$, the sequence $(t(F,G_n))_{n=1}^\infty$ converges in distribution.
		\item[(4)] For every $k \geq 1$, the sequence $(\G(k,G_n))_{n=1}^\infty$ of random elements of $\Pi(\KK)$ converges in distribution.
	\end{itemize}
	
	If any of the above apply, then the respective limits in $(2)$ and $(3)$ are $t(F,\Gamma)$ with $\Gamma$ being a random element of $\overline{\Pi(\KK)^+}$ given by
	$(1)$, and also $\Gamma \in \overline{\Pi(\KK)^+} \setminus \Pi(\KK)^+$, almost surely.
	
	If $t(F,G_n)$ in $(2)$ and $(3)$ is replaced by $t_\inj(F,G_n)$, then the equivalence of the four statements still persists and the limits in $(2)$ and $(3)$ are $t(F,\Gamma)$.
	
	If every $G_n$ is concentrated on some single element of $\Pi(\KK)$ (non-random case), then the equivalence holds with the sequences in $(1)$, $(2)$, and $(3)$ being numerical instead of distributional, while $(4)$ remains unchanged.
	
\end{theorem}
\begin{proof}
	The equivalence of $(1)$ and $(2)$ is immediate. The implication from $(3)$ to $(2)$ is also clear by definition. 
	
	For showing that $(2)$ implies $(3)$, we consider first an arbitrary $F \in \Pi(\langle \FF \rangle)$, where $\langle \FF \rangle$ is the linear space generated by $\FF$. Then there exist $F^1, \dots, F^l \in \Pi(\FF)$ on the same vertex set as $F$, say $[k]$, and $\lambda_1, \dots, \lambda_l \in \R$ such that for any non-random $G \in \Pi(\KK)$ and $\phi\colon [k] \to V(G)$ it holds that 
	\begin{align*}
	\prod_{i_1, \dots, i_r=1}^{k} F(i_1,\dots,i_r)& (G(\phi(i_1), \dots, \phi(i_r))) \\ &=\sum_{j=1}^l \lambda_j\prod_{i_1, \dots, i_r=1}^{k} F^j(i_1,\dots,i_r) (G(\phi(i_1), \dots, \phi(i_r))).
	\end{align*}
	 So therefore we can express $t(F,G)= \sum_{j=1}^l \lambda_j t(F^j,G)$. We return to the case when $G_n$ is random. The weak convergence of $t(F,G_n)$ is equivalent to the convergence of each of its moments, its $t$th moment can be written by the linearity of the expectation as a linear combination of a finite number of mixed moments of the densities corresponding to $F^1, \dots, F^l \in \Pi(\FF)$. For an arbitrary vector of non-negative integers $\alpha=(\alpha_1, \dots, \alpha_l)$, let $F^\alpha$ be the element of $\Pi(\FF)$ that is the disjoint union $\alpha_1$ copies of $F^1$, $\alpha_2$  copies of $F^2$, and so on. It holds that $t(F^1,G_n)^{\alpha_1}\dots t(F^l,G_n)^{\alpha_l}=t(F^\alpha,G_n)$, and in particular the two random variables on the two sides are equal in expectation. Condition $(2)$ implies that $\E[t(F^\alpha,G_n)]$ converges for each $\alpha$, therefore the mixed moments of the $t(F^i,G_n)$ densities and the moments of $t(F,G_n)$ also do. This implies that  $t(F,G_n)$ also converges in distribution for any $F \in \Pi(\langle F \rangle)$. Now let $F' \in \Pi(\CC)$ and $\varepsilon>0$ be arbitrary, and $F \in \Pi(\langle F \rangle)$ on the same vertex set $[k]$ as $F'$ be such that its entries are at most $\varepsilon$-far in $L^\infty$ from the corresponding entries of $F'$. Then \begin{align*}
	&&&|t(F',G)-t(F,G)|  \\ &&&= \Bigg|\frac{1}{|V(G)|^k}\sum_{\phi\colon [k] \to V(G)} \prod_{i_1, \dots, i_r=1}^{k} F(i_1,\dots,i_r) (G(\phi(i_1), \dots, \phi(i_r))) \\ &&& \qquad \qquad \qquad \qquad \qquad \qquad- \prod_{i_1, \dots, i_r=1}^{k} F'(i_1,\dots,i_r) (G(\phi(i_1), \dots, \phi(i_r)))\Bigg| \\ 
	&&&= \frac{1}{|V(G)|^k}\sum_{\phi\colon [k] \to V(G)} \sum_{i_1, \dots, i_r=1}^{k} \left|\prod_{(j_1, \dots, j_r)<(i_1, \dots, i_r)}F(j_1,\dots,j_r) (G(\phi(j_1), \dots, \phi(j_r)))\right|\\&&& \qquad \left|\prod_{(j_1, \dots, j_r)>(i_1, \dots, i_r)}F'(j_1,\dots,j_r) (G(\phi(j_1), \dots, \phi(j_r)))\right| \\&&& \qquad \left|F(i_1,\dots,i_r) (G(\phi(i_1), \dots, \phi(i_r))) -  F'(i_1,\dots,i_r) (G(\phi(i_1), \dots, \phi(i_r))) \right| \\ &&&\leq k^r \varepsilon \max\{(\|F'\|_\infty+\varepsilon)^{k^r-1},1\}\end{align*} for any $G \in \Pi(\KK)$ (random or non-random), which implies $(3)$, as $\varepsilon>0$ was chosen arbitrarily.
	
	We turn to show the equivalence of $(3)$ and $(4)$. Let $\Pi_k(\KK) \subset \Pi(\KK)$ the set of elements of $\Pi(\KK)$ with vertex cardinality $k$. 
	The sequence $(\G(k,G_n))_{n=1}^\infty$ converges in distribution exactly when for each continuous function $f \in C(\Pi_k(\KK))$ on $\Pi_k(\KK)$ the expectation $\E[f(\G(k,G_n))]$ converges as $n \to \infty$. For each $F \in \Pi(\CC)$ and $\alpha \geq 1$, the function $t_\inj^\alpha(F,G)$ is continuous on $\Pi^{|V(F)|}(\KK)$ and $t_\inj(F,G)=t_\inj(F,\G(|V(F)|,G))$, so $(3)$ follows from $(4)$.

	For showing the other direction, that $(3)$ implies $(4)$, let us fix $k \geq 1$.
	We claim that the linear function space $M=\langle t(F, .) | F \in \Pi(\CC) \rangle \subset  C(\Pi_k(\KK))$ is an algebra containing the constant function, and that it separates any two elements of $\Pi_k(\KK)$. It follows that $\langle t(F, .) | F \in \Pi(\CC) \rangle$ is $L^\infty$-dense in $ C(\Pi_k(\KK))$ by the Stone-Weierstrass theorem, which implies by our assumptions that $\E[f(\G(k,G_n))] $ converges for any $f \in C(\Pi_k(\KK))$, since we know that $\E[t_\inj(F,\G(k,G_n))]= \E[t_\inj(F,G_n)]$ whenever $|V(F)|\leq k$. We will see in a moment that $t_{\inj}(F,. )\in M$, convergence of $\E[t_\inj(F,G_n)]$ follows from (\ref{ch2:eqinj}) and the requirement that $|V(G_n)|$ tends to infinity in probability.  
	
	Now we turn to show that our claim is indeed true. For two graphs $F_1, F_2 \in \Pi(\CC)$ we have $t(F_1,G)t(F_2,G)=t(F_1F_2,G)$ for any $G \in \Pi_k(\KK)$, where the product $F_1F_2$ denotes the disjoint union of the two $\CC$-colored graphs. Also, $t(F,G)=1$ for the graph $F$ on one node with a loop colored with the constant $1$ function. Furthermore we have that $\hom(F, G)= k^{|V(F)|} t(F, G) \in M$ for $|V(G)|=k$, so therefore $$
	\inj(F,G)=\sum_{\P \textrm{ partition of } V(F)} (-1)^{|V(F)|-|\P|} \prod_{S \in \P} (|S|-1)! \, \hom(F/\P,G) \quad  \in M,
	$$ where  $\inj(F,G)= t_\inj(F,G) k (k-1) \dots (k-|V(F)|+1)$ and $F/\P \in \Pi^{|\P|}(\CC)$ whose edges are colored by the product of the colors of $F$ on the edges between the respective classes of $\P$. This equality is the consequence of the Mobius inversion formula, and that $\inj(F,G)=\sum_{\P \textrm{ partition of } V(F)} \hom(F/\P, G)$.  For $G$ and $F$ defined on the node set $[k]$ recall that
	\begin{equation}
	\label{ch2:eqinj2}
	\inj(F,G)=\sum_{\phi \in S_k} \prod_{i_1, \dots, i_r=1}^{k} F(i_1,\dots,i_r) (G(\phi(i_1), \dots, \phi(i_r))).
	\end{equation}
	
	Now fix $G_1, G_2 \in \Pi_k(\KK)$ and let $F \in \Pi_k(\CC)$ such that  $\{F(i_1,\dots,i_r) (G_j(l_1, \dots, l_r))\}$ are algebraically independent elements of $\R$ (such an $F$ exists, we require a finite number of algebraically independent reals, and can construct each entry of $F$ by polynomial interpolation). If $G_1$ and $G_2$ are not isomorphic, than for any possible node-relabeling for $G_2$ there is at least one term in the difference $\inj(F,G_1)-\inj(F,G_2)$ written out in the form of (\ref{ch2:eqinj2}) that does not get canceled out, so therefore $\inj(F,G_1) \neq \inj(F,G_2)$.    
	
	We examine the remaining statements of the theorem. Clearly, $\Gamma \notin \Pi(\KK)^+$, because $|V(G_n)| \to \infty$ in probability. The results for the case where the map in $(1)$ and the densities in $(2)$ and $(3)$ are replaced by the injective version are yielded by (\ref{ch2:eqinj}), the proof of the non-random case carries through in a completely identical fashion.

\end{proof}

We are now ready to formulate the definition of convergence in $\Pi(\KK)$\sindex[notions]{graph!convergence definition}\sindex[notions]{hypergraph!convergence definition}.

\begin{definition}\label{ch2:defconv}
	If  $(G_n)_{n=1}^\infty$ is a sequence in $\Pi(\KK)$ with $|V(G_n)| \to \infty$ and any of the conditions above of \Cref{ch2:thm1} hold, then we say that $(G_n)_{n=1}^\infty$ converges. 
\end{definition}

We would like to add that, in the light of \Cref{ch2:thm1}, the convergence notion is independent from the choice of the family $\FF$.

The next lemma gives information about the limit behavior of the sequences where the vertex set cardinality is constant.

\begin{lemma}
	\label{ch2:lem2} Let $(G_n)_{n=1}^\infty$ be a random sequence in $\Pi_k(\KK)$, and additionally be such that for every $F \in \Pi(\FF)$ the sequences $(t_{\inj}(F,G_n))_{n=1}^\infty$ converge in distribution. Then there exists a random $H \in \Pi_k(\KK)$, such that for every $F \in \Pi(\FF)$ we have $t(F,G_n)\to t(F,H)$ and $t_\inj(F,G_n) \to t_\inj(F,H)$ in distribution.
\end{lemma} 
\begin{proof}
	We only sketch the proof. The distributional convergence of $(G_n)_{n=1}^\infty$ follows the same way as in the proof of \Cref{ch2:thm1}, the part about condition $(2)$ implying $(3)$ together with  the part stating that $(3)$ implies $(4)$. The existence of a random $H$ satisfying the statement of the lemma is obtained by invoking the Riesz representation theorem for positive functionals.
\end{proof}

\paragraph{Exchangeable arrays}

The correspondence analogous to the approach of Diaconis and Janson in \cite{DJ} will be established next between the elements of the limit space $\overline{\Pi(\KK)^+}$ that is compact, and the extreme points of the space of random exchangeable infinite $r$-arrays with entries in $\KK$. These are arrays, whose distribution is invariant under finite permutations of the underlying index set. 

\begin{definition}[Exchangeable $r$-array] \label{ch2:defexc}
	Let $(H(e_1,\dots,e_r))_{1 \leq e_1, \dots, e_r < \infty}$ be an infinite $r$-array of random entries from a Polish space $\KK$. We call the random array separately exchangeable\sindex[notions]{exchangeable array!separately} if 
	$$ 
	(H(e_1,\dots,e_r))_{1 \leq e_1, \dots, e_r < \infty} $$
	has the same probability distribution as
	$$
	(H(\rho_1(e_1),\dots,\rho_r(e_r)))_{1 \leq e_1, \dots, e_r < \infty}
	$$ 
	for any $\rho_1, \dots, \rho_r \in S_\N$ collection of finite permutations, and jointly exchangeable\sindex[notions]{exchangeable array!jointly} (or simply exchangeable), if the former holds only for all  $\rho_1= \dots= \rho_r \in S_\N$.
\end{definition}
\sindex[notions]{exchangeable array}

For  a finite set $S$, let $\hhh_0(S)$\sindex[symbols]{h@$\hhh_0(S)$,$\hhh(S)$,$\hhh(S,m)$,$\hhh_0(S,m)$} and $\hhh(S)$ denote the power set and the set of nonempty subsets of $S$, respectively, 
and $\hhh(S,m)$ the set of nonempty subsets of $S$ of cardinality at most $m$, also $\hhh_0(S,m)=\hhh(S,m) \cup \{\emptyset\}$. A $2^r-1$-dimensional real vector $x_{\hhh(S)}$\sindex[symbols]{x@$x_{\hhh(S)}$} denotes $(x_{T_1}, \dots, x_{{T_{2^r-1}}})$, where $T_1, \dots, T_{2^r-1}$ is a fixed ordering of the nonempty subsets of $S$ with $T_{2^r-1}=S$, for a permutation $\pi$ of the elements of $S$ the vector  $x_{\pi(\hhh(S))}$ means $(x_{\pi^*(T_1)}, \dots, x_{\pi^*({T_{2^r-1}})})$, where $\pi^*$ is the action of $\pi$ permuting the subsets of $S$. Similar conventions apply when $x$ is indexed by other set families.

It is clear that if we consider a measurable function $f\colon  [0,1]^{\hhh_0([r])} \to \KK$, and independent random variables uniformly distributed on $[0,1]$ that are associated 
with each of the subsets of $\N$ of cardinality at most $r$, then by plugging in these random variables into $f$ for every $e \in \N^r$ in the right way suggested by a fixed natural bijection $l_e\colon e \to [r]$, the result will be an exchangeable random $r$-array. The shorthand $\mathrm{Samp}(f)$\sindex[symbols]{s@$\mathrm{Samp}(f)$} denotes this law of the infinite directed $r$-hypergraph model generated by $f$.

The next theorem, states that all exchangeable arrays with values in $\KK$ arise from some $f$ in the former way.
\begin{theorem}\cite{K}\label{ch2:exch}
	Let $\KK$ be a Polish space. Every $\KK$-valued exchangeable $r$-array $(H(e))_{e \in {\N^r}}$ has law equal to $\mathrm{Samp}(f)$ for some measurable $f\colon  [0,1]^{\hhh_0([r])} \to \KK$, that is, there exists a function $f$, so that if $(U_s)_{s \in \hhh_0(\N,r)}$ are independent uniform $[0,1]$ random variables, then \begin{equation} \label{ch2:repr}
	H(e)=f(U_\emptyset, U_{\{e_1\}}, U_{\{e_2\}},  \dots, U_{\overline{e} \setminus \{e_r\}},U_{\overline{e}})
	\end{equation}
	for every $e=(e_1, \dots , e_r)\in \N^r$, where $H(e)$ are the entries of the infinite $r$-array.
\end{theorem}\sindex[notions]{exchangeable array!representation}

If $H$ in the above theorem is invariant under permuting its coordinates, then the corresponding function $f$ is invariant under the coordinate permutations that are induced by the set permuting $S_r$-actions.

\Cref{ch2:exch} was first proved by de Finetti \cite{dF}  (in the case $\KK=\{0,1\}$) and by Hewitt and Savage \cite{HS} (in the case of general $\KK$) for $r=1$, independently by Aldous \cite{Al} and Hoover \cite{H} for $r=2$,  and by Kallenberg \cite{K} for arbitrary $r\geq 3$. For equivalent formulations, proofs and further connections to related areas see the recent survey of Austin \cite{Au}.

In general, there are no symmetry assumptions on $f$, in the directed case $H(e)$ might differ from $H(e')$, even if $e$ and $e'$ share a common base set. In this case these two entries do not have the property of conditional independence over a $\sigma$-algebra given by some lower dimensional structures, that means for instance  the independence over $\{ \,U_{\alpha} \mid \alpha \subsetneq e\, \}$ for an exchangeable $r$-array with law $\mathrm{Samp}(f)$ given by a function $f$ as above. 

With the aid of \Cref{ch2:exch} we will provide a form of representation of the limit space $\overline{\Pi(\KK)^+}$ through the points of the space of random infinite exchangeable $r$-arrays.
The correspondence will be established through a sequence of theorems analogous to the ones stated and proved in \cite[Section 2 to 5]{DJ}, combined with the compactification argument regarding the limit space from \cite{LSzcom}, see also \cite[Chapter 17.1]{L} for a more accurate picture. The proofs in our case are mostly ported in a straightforward way, if not noted otherwise we direct the reader for the details to \cite{DJ}.

Let $\LL_\infty=\LL_\infty(\KK)$\sindex[symbols]{l@$\LL_\infty(\KK)$,$\LL_n(\KK)$} denote the set of all node labeled countably infinite $\KK$-colored $r$-uniform
directed hypergraphs. Set the common vertex set of the elements of $\LL_\infty$ to $\N$, and define the set of $[n]$-labeled $\KK$-colored $r$-uniform directed hypergraphs as $\LL_n=\LL_n(\KK)$. Every $G \in \LL_n$ can be viewed as an element of $\LL_\infty$ simply by adding isolated
vertices to $G$ carrying the labels $\N \setminus [n]$ in the uncolored case, and the arbitrary but fixed color $c \in \KK$ to edges incident to these vertices in the colored case, therefore we think about $\LL_n$ as a subset of $\LL_\infty$ (and also of $\LL_m$ for every $m\geq n$). Conversely, if $G$ is a (random) element of $\LL_\infty$, then by restricting $G$ to the vertices labeled by $[n]$, we get $G|_{[n]} \in \LL_n$.\sindex[symbols]{g@$G\vert_{S}$} 
If $G$ is a labeled or unlabeled $\KK$-colored $r$-uniform directed hypergraph (random or not) with
vertex set of cardinality $n$, then let $\hat G$ stand for the random element of $\LL_n$ (and also $\LL_\infty$) which we obtain by first throwing away the labels of $G$ (if there where any), and then apply a random labeling chosen uniformly from all possible ones with the label set $[n]$. 

A random element of $\LL_\infty$ is \emph{exchangeable} analogously to \Cref{ch2:defexc} if its distribution is invariant under any permutation of the vertex set $\N$ that only moves finitely many vertices, for example infinite hypergraphs whose edge-colors are independently identically distributed are exchangeable. An element of $\LL_\infty$ can also be regarded as an infinite $r$-array whose diagonal elements are colored with a special element $\iota$ that is not contained in $\KK$, therefore the corresponding $r$-arrays will be $\KK \cup \{\iota\}$-colored.

%

The next theorem relates the elements of $\Pi(\KK)^+$ to exchangeable random elements of $\LL_\infty$.
\begin{theorem}\label{ch2:thm3}
	Let $(G_n)_{n=1}^\infty$ be a random sequence in $\Pi(\KK)$ with $|V(G_n)|$ tending to infinity
	in probability. Then the following are equivalent.
	\begin{itemize}
		\item[(1)] $\tau^+(G_n) \to \Gamma$ in distribution for a random $\Gamma \in \overline{\Pi(\KK)^+} \setminus \Pi(\KK)^+$.
		\item[(2)] $\hat G_n \to H$ in distribution in $\LL_\infty(\KK)$, where $H$ is a random element of $\LL_\infty(\KK)$.
	\end{itemize}
	If any of these hold true, then $\E t(F,\Gamma) = \E t_\inj(F,H|_{[k]})$ for every $F \in \Pi_k(\CC)$, and also, $H$ is exchangeable.  
\end{theorem}
\begin{proof}
	If $G \in \Pi(\KK)$ is deterministic and $F \in \Pi_k(\FF)$ with $|V(G)| \geq k$ then $E t_\inj(F, \hat G|_{[k]}) = t_\inj(F,G)$, where the expectation $E$ is taken with respect to the
	random (re-)labeling $\hat G$ of $G$. For completeness we mention that for a labeled, finite $G$ the quantity $t(F,G)$ is understood as $t(F,G')$ with $G'$ being the unlabeled version of $G$, also, $F$ in $t(F,G)$ is always regarded a priori as labeled, however the densities of isomorphic labeled graphs in any graph coincide. 
	If we consider $G$ to be random, then by the fact that
	$0 \leq t(F,G) \leq 1$ (as $\|F\|_\infty \leq 1$) we have that $|\E E t_\inj(F, \hat G|_{[k]}) -\E t_\inj(F,G)| \leq \PPP(|V(G)| < k )$ for $F \in \Pi_k(\FF)$. 
	
	Assume $(1)$, then the above implies, together with
	$\PPP(|V(G_n)| < k ) \to 0$ and $(1)$, that $\E E t_\inj(F, \hat G_n|_{[k]}) \to \E t(F, \Gamma)$ (see \Cref{ch2:thm1}). This implies that $\hat G_n|_{[k]} \to H_k$ in distribution for some random $H_k \in \LL_k$ with $\E t_\inj(F,H_k) = \E t(F, \Gamma)$, see \Cref{ch2:lem2}, furthermore, with appealing to the consistency of the $H_k$ graphs in $k$, there exists a random $H \in \LL_\infty$ such that $H|_{[k]} = H_k$ for each $k \geq 1$, so $(1)$ yields $(2)$. 
	
	Another consequence is that $H$ is exchangeable: the exchangeability property is equivalent to the vertex permutation invariance of the distributions of $H|_{[k]}$ for each $k$. This is ensured by the fact that $H|_{[k]} = H_k$, and $H_k$ is the weak limit of a vertex permutation invariant random sequence, for each $k$. 
	
	For the converse direction we perform the above steps in the reversed order using $$|\E E t(F, \hat G_n|_{[k]}) -\E t(F,G_n) \leq \PPP(|V(G_n)| < k )$$ 
	again in order to establish the convergence of $(\E t(F,G_n))_{n =1}^\infty$. \Cref{ch2:thm1} certifies now the existence of the suitable random $\Gamma \in \overline{\Pi(\KK)^+} \setminus \Pi(\KK)^+$, this shows that
	$(2)$ implies $(1)$.
	
\end{proof}

We built up the framework in the preceding statements \Cref{ch2:thm1} and \Cref{ch2:thm3} in order to formulate the following theorem, which is the crucial ingredient to the desired representation of limits. 
\begin{theorem} \label{ch2:thm4}
	There is a one-to-one correspondence between random elements of $\overline{\Pi(\KK)^+} \setminus \Pi(\KK)^+$ and random exchangeable elements of $\LL_\infty$. Furthermore, there is a one-to-one correspondence between elements of $\overline{\Pi(\KK)^+} \setminus \Pi(\KK)^+$ and extreme points of the set of random exchangeable elements of $\LL_\infty$.
	The relation is established via the equalities $\E t(F,\Gamma) = \E  t_\inj(F,H|_{[k]})$ for every $F \in \Pi_k(\CC)$ for every $k\geq 1$.
\end{theorem}
\begin{proof}
	Let $\Gamma$  a random element of $\overline{\Pi(\KK)^+} \setminus \Pi(\KK)^+$. Then by definition of $\Pi(\KK)^+$  there is a sequence $(G_n)_{n=1}^\infty$ in $\Pi(\KK)$ with $|V(G_n)| \to \infty$ in probability such that $\tau^+(G_n) \to \Gamma$ in distribution in $\Pi(\KK)^+$. By virtue of  \Cref{ch2:thm3} there exists a random $H \in \LL_\infty$ so that $\hat G_n \to H$ in distribution in $\LL_\infty$, and $H$ is exchangeable. The distribution of $H|_{[k]}$ is determined by the numbers $\E t_\inj(F,H|_{[k]})$, see \Cref{ch2:thm1}, \Cref{ch2:lem2}, and the arguments therein, and these numbers are provided by the correspondence.
	
	For the converse direction, let  $H$ be random exchangeable element of $\LL_\infty$. Then let $G_n=H|_{[n]}$, we have $G_n \to H$ in distribution, and also $\hat G_n \to H$ in distribution by the vertex permutation invariance of $G_n$ as a node labeled object. Again, we appeal to \Cref{ch2:thm3}, so $\tau^+(G_n) \to \Gamma$ for a $\Gamma$ random element of $\overline{\Pi(\KK)^+} \setminus \Pi(\KK)^+$, which is determined completely by the numbers $\E t(F,\Gamma)$ that are provided by the correspondence, see \Cref{ch2:thm3}.
	
	The second version of the relation between non-random $\Gamma$'s and extreme points of exchangeable elements is proven similarly, the connection is given via  $t(F,\Gamma) = \E t_\inj(F,H|_{[k]})$ between the equivalent objects.
	
\end{proof}

The characterization of the aforementioned extreme points in \Cref{ch2:thm4} was given \cite{DJ} in the uncolored graph case, we state it  next for our general setting, but refrain from giving the proof here, as it is completely identical to  \cite[Theorem 5.5.]{DJ}.

\begin{theorem}\cite{DJ}
	The distribution of $H$ that is an exchangeable random element of $\LL_\infty$ is exactly in that case an extreme point of the set of exchangeable measures if the random objects $H|_{[k]}$ and $H|_{\{k+1, \dots\}}$ are probabilistically independent for any $k\geq 1$. In this case the representing function $f$ from \Cref{ch2:exch} does not depend on the variable corresponding to the empty set. 
\end{theorem}

\paragraph{Graphons as limit objects}\label{ch2:sec:graphonintro}



Let the \emph{$r$-kernel}\sindex[notions]{r@$r$-kernel} space $\hat \Xi_0^r$ denote the space of the bounded measurable functions of the form $W\colon [0,1]^{\hhh([r],r-1)} \to \R$, and the subspace $\Xi^r_0$\sindex[symbols]{x@$\hat \Xi_0^r$,$\Xi^r_0$} of $\hat \Xi_0^r$ the symmetric $r$-kernels\sindex[notions]{r@$r$-kernel!symmetric} that are invariant under coordinate permutations $\pi^*$ induced by some $\pi \in S_r$, that is $W(x_{\hhh([r],r-1)})=W(x_{\pi^*(\hhh([r],r-1))})$ for each $\pi \in S_r$. We will refer to this invariance in the paper both for $r$-kernels and for measurable subsets of $[0,1]^{\hhh([r])}$ as \emph{$r$-symmetry}\sindex[notions]{r@$r$-symmetric}.
The kernels $W\in \Xi^r_{I}$ take their values in some interval $I$, for $I=[0,1]$ we call these special symmetric $r$-kernels  \emph{$r$-graphons}\sindex[notions]{r@$r$-graphon}, and their set $\Xi^r$\sindex[symbols]{x@$\Xi^r_{I}$,$\Xi^r$,$\Xi^{r,q}$}. In what follows, $\lambda$\sindex[symbols]{l@$\lambda$} as a measure always denotes the usual Lebesgue measure in $\R^d$, where the dimension $d$ is everywhere clear from the context.

If $W \in \Xi^r$ and $F \in \Pi^r$, then the $F$-density of $W$ is defined as

\begin{align}\label{ch2:defdens3}
t(F,W)= \int_{ [0,1]^{\hhh([k], r-1)} } \prod_{e \in E(F)} W(x_{\hhh(e,r-1)})  \prod_{e \notin E(F)} (1-W(x_{\hhh(e,r-1)})) \du \lambda(x_{\hhh([k], r-1)}).
\end{align}

Let $\KK$ be a compact Polish space, and $W\colon  [0,1]^{\hhh([r])} \to \KK$ be a measurable function, we will refer to such an object as a \emph{$(\KK,r)$-digraphon}\sindex[notions]{K@$(\KK,r)$-digraphon}, their set is denoted by $\tilde \Xi^{r}(\KK)$. Note that there are no symmetry assumptions in this general case, if additionally $W$ is $r$-symmetric, then we speak about  \emph{$(\KK,r)$-graphons}\sindex[notions]{K@$(\KK,r)$-graphon}, their space is $\Xi^{r}(\KK)$\sindex[symbols]{x@$\tilde \Xi^{r}(\KK)$,$\Xi^{r}(\KK)$}. For $\KK=\{0,1\}$ the set $\P(\KK)$ can be identified with the $[0,1]$ interval encoding the success probabilities of Bernoulli trials to get the common $r$-graphon form as a function $W\colon [0,1]^{2^r-2} \to [0,1]$ employed in \cite{ESz}.

The density of a $\KK$-colored graph $F\in \tilde \Pi_k(C(\KK))$ in the $(\KK,r)$-digraphon $W$ is defined analogously to  (\ref{ch2:defdens4}) and (\ref{ch2:defdens3}) as
\sindex[notions]{homomorphism density!colored graphon}
\begin{align}\label{ch2:defdens5}
t(F,W)= \int_{ [0,1]^{\hhh([k], r)} } \prod_{e \in [k]^r} F(e)( W(x_{\hhh(e,r)})) \du \lambda(x_{\hhh([k], r)}).
\end{align}

For $k \geq 1$ and an undirected $W \in \Xi^{r}(\KK)$ the random $(\KK,r)$-graph $\G(k, W)$ is defined on the vertex set $[k]$ by selecting a  uniform random point $(X_S)_{S \in \hhh([k], r)} \in [0,1]^{\hhh([k], r)}$ that enables the assignment of the color $W(X_{\hhh(e)})$ to each edge $e \in {[k] \choose r}$. For a directed $W$ the sample point is as above, the color of the directed edge $e \in [k]^r$ is $W(X_{\hhh(e)})$, but in this case the ordering of the power set of the base set $\overline e$\sindex[symbols]{e@$\overline e$} of $e$ matters in contrast to the undirected situation and is given by $e$, as $W$ is not necessarily $r$-symmetric. 

Additionally we define the averaged sampled $r$-graph\sindex[notions]{r@$r$-graph!averaged sampled} for $\KK \subset \R$ denoted by $\hh(k,W)$\sindex[symbols]{h@$\hh(k,W)$}, it has vertex set $[k]$, and the weight of the edge $e \in {[k] \choose r}$ is the conditional expectation $\E[W(X_{\hhh(e)}) \mid  X_{\hhh(e,1)} ]$, and therefore the random $r$-graph is measurable with respect to $X_{\hhh([k],1)}$. We will use the compact notation $X_i$ for $X_{\{ i\}}$ for the elements of the sample indexed by singleton sets.


We define the random exchangeable $r$-array $H_W$ in $\LL_\infty$ as the element that has law $\mathrm{Samp}(W)$ for the $(\KK,r)$-digraphon $W$, as in \Cref{ch2:exch}. 
Furthermore, we define \sindex[symbols]{g@$\Gamma_W$}$\Gamma_W \in \overline{\Pi(\KK)^+} \setminus \Pi(\KK)^+$ to be the element associated to $H_W$ through \Cref{ch2:thm3}.

Now we are able to formulate the representation theorem for $\mathcal K$-colored $r$-uniform directed hypergraph limits using the representation of exchangeable arrays, see (\Cref{ch2:exch}). It is an immediate consequence of \Cref{ch2:thm3} and \Cref{ch2:thm4} above. 

\begin{theorem}\label{ch2:reprmain}
	Let $(G_n)_{n=1}^\infty$ be a sequence in $\Pi(\KK)$ with $|V(G_n)| \to \infty$  such that for  every $F \in \Pi(\FF)$ the sequence $t(F,G_n)$ converges. Then there exists a function  $W\colon  [0,1]^{\hhh([r])} \to \KK$ (that is $W \in \Xi^{r}(\KK)$) such that $t(F,G_n) \to t(F,\Gamma_W)$ for every $F \in \Pi(\FF)$. In the directed case when the sequence is in $\tilde \Pi(\KK)$, then the corresponding limit object $W$ is in  $\tilde \Xi^{r}(\KK)$.
\end{theorem}

We mention that  $t(F,\Gamma_W)=t(F,W)$ for every  $F\in \Pi(C(\KK))$ and $W \in \Xi^r (\KK)$. Alternatively we can also use the form $W\colon  [0,1]^{\hhh([r],r-1)} \to \P(\KK)$ for $(\KK,r)$-graphons and digraphons in $\Xi^r(\KK)$ whose values are probability measures, this representation was applied in \cite{LSzcom}. 

In previous works, for example in \cite{DJ}, the limit object of a sequence of simple directed graphs without loops was represented by a $4$-tuple of $2$-graphons $(W^{(0,0)},W^{(1,0)},W^{(0,1)},W^{(1,1)})$ that satisfies $\sum_{i,j} W^{(i,j)}(x,y)=1$ and $W^{(1,0)}(x,y)=W^{(0,1)}(y,x)$ for each $(x,y) \in [0,1]^2$. 
A generalization of this representation can be given in our case of the $\Pi(\KK)$ limits the following way. We only present here the case when $\KK$ is a continuous space, the easier finite case can be dealt with analogously. 

We have to fix a Borel probability measure $\mu$ on $\KK$, we set this to be the uniform distribution if $\KK \subset \R^d$ is a domain or $\KK$ is finite. The limit space consists of collections of $(\R,r)$-kernels $W=(W^{u})_{u \in \mathcal U}$, where $\mathcal U$ is the set of all functions $u \colon  S_r \to \KK$. Additionally, $W$ has to satisfy $\int_{\mathcal U} W^u(x) \du \mu^{\otimes S_r}(u)=1$ and $0 \leq W^{\pi^*(u)}(x)=W^u(x_{\pi^*(\hhh([r],r-1))})$ for each $\pi \in S_r$ and $x \in [0,1]^{\hhh([r],r-1)}$. As before, the action $\pi^*$ of $\pi$ on $[0,1]^{\hhh([r],r-1)}$ is the induced coordinate permutation by $\pi$, with the unit cubes coordinates indexed by non-trivial subsets of $[r]$. Without going into further details we state the connection between the limit form spelled out above and that in \Cref{ch2:reprmain}. It holds
\begin{equation*}
\int_{U} W^u((x_{\hhh([r],r-1)}) \du \mu^{\otimes S_r}(u) = \PPP[(W(x_{\pi^*(\hhh([r],r-1))},Y)_{\pi \in S_r}) \in U]
\end{equation*}
for every measurable $U \subset \mathcal U$ and $x \in [0,1]^{\hhh([r],r-1)}$, where $Y$ is uniform on $[0,1]$, and the $W$ on the right-hand side is a $(\KK,r)$-digraphon, whereas on the left we have the corresponding representation as a (possibly infinite) collection of $(\R,r)$-kernels.


\label{ch2:naivegraphon} In several applications it is more convenient to use a naive form for the limit representation, from which the limit element in question is not decisively retrievable. The naive limit space consists of \emph{naive $(\KK,r)$-graphons}\sindex[notions]{K@$(\KK,r)$-graphon!averaged naive} $\overline W\colon  [0,1]^{r} \to \P(\KK)$, where now the arguments of $W$ are indexed with elements of $[r]$. From a proper $r$-graphon  $W\colon  [0,1]^{\hhh([r])} \to \KK$ we get its naive counterpart by averaging, the $\KK$-valued random variable $\E[W(x_{\hhh([r],1)}, U_{\hhh([r],r-1) \setminus \hhh([r],1)} ,Y)| Y]$ has distribution $\overline W(x_1, \dots, x_r)$, where $(U_{S})_{S \in \hhh([r],r-1) \setminus \hhh([r],1)}$ and $Y$ are i.i.d. uniform on $[0,1]$. 

On a further note we introduce \emph{averaged naive $(\KK,r)$-graphons}\sindex[notions]{K@$(\KK,r)$-graphon!averaged naive} for the case, when $\KK \subset \R$, these are of the form $\tilde{W}\colon  [0,1]^{r} \to \R$ and are given by complete averaging, that is $\E[W(x_1, \dots, x_r, U_{\hhh([r]) \setminus \hhh([r],1)} ,Y)]=\tilde W(x_1, \dots, x_r),$ where $(U_{S})_{S \in \hhh([r]) \setminus \hhh([r],1)}$ are i.i.d. uniform on $[0,1]$. A \emph{naive $r$-kernel}\sindex[notions]{r@$r$-kernel!naive} is a real-valued, bounded function on $[0,1]^r$, or equivalently on $[0,1]^{\hhh([r],1)}$.

We can associate to each $G \in \Pi_n^{r}(\KK)$ an element $W_G \in \Xi^{r}(\KK\cup\{\iota\})$\sindex[symbols]{w@$W_G$} by subdividing the unit $r$-cube $[0,1]^{\hhh([r],1)}$ into $n^r$ small cubes the natural way and defining the function $W' \colon  [0,1]^{\hhh([r],1)} \to \KK$ that takes the value $G(i_1, \dots, i_r)$ on $[\frac{i_1-1}{n}, \frac{i_1}{n}] \times \dots \times [\frac{i_r-1}{n}, \frac{i_r}{n}]$ for distinct $i_1, \dots, i_r$, and the value $\iota$ on the remaining diagonal cubes, note that these functions are naive $(\KK,r)$-graphons. Then we set $W_G(x_{\hhh([r],r-1)}) = W'(p_{\hhh([r],1)}(x_{\hhh([r],r-1)}))$, where $p_{\hhh([r],1)}$ is the projection to the suitable coordinates.
The special color $\iota$ here stands for the absence of colors has to be employed in this setting as rectangles on the diagonal correspond to loop edges. The corresponding $r$-graphon $W^\iota$ is $\{0,1\}$-valued. 
The sampled random $r$-graphs $\G(k,W_G)$ and $\hh(k,W_G)$ from the naive $r$-graphons are defined analogously to the general case. If $\KK \subset \R$, then note that  $\hh(k,W_G)=\G(k,W_G)$ for every $G$, because the colors of $G$ are all point measures.

Note that $t(F,G)=t(F,W_G)$, and
\begin{align}\label{ch6:eq22}
|t_\inj(F,G)-t(F,W_G)|\leq \frac{{k \choose 2}}{n - {k \choose 2}}
\end{align}
for each $F \in \Pi_k^{r}$, hence the representation as naive graphons is compatible in the sense that $\lim_{n \to \infty} t_\inj(F,G_n)=\lim_{n \to \infty} t(F,W_{G_n})$ for any sequence $(G_n)_{n=1}^\infty$ with $|V(G_n)|$ tending to infinity. This implies that $d_{\tw}(\G(k,G_n), \G(k,W_{G_n}))\to 0$ as $n$ tends to infinity.

We remark that naive and averaged naive versions in the directed case are defined analogously.

\subsection{Representation of \texorpdfstring{$r$}{r}CSP formulas as hypergraphs, and their convergence} \label{ch2:sec.limcsp}

In this subsection we elaborate on how homomorphism and sampling is meant in the CSP context, and formulate a representation the limit space in that context. Recall \Cref{ch2:defcsp} for the way how we perceive $r$CSP formulas.

Let $F$ be an $r$CSP formula on the variable set $\{x_1, \dots, x_n\}$ over an arbitrary domain $K$, and let $F[x_{i_1}, \dots, x_{i_k}]$ be the induced subformula of $F$ on the variable set $\{x_{i_1}, \dots,$ $x_{i_k}\}$. Let $\G(k,F)$ denote the random induced subformula on $k$ uniformly chosen variables from the elements of $V(F)$. 

It is clear using the terminology of \Cref{ch2:defexc} that the relation $\omega=(f; x_e) \in F[x_{i_1}, \dots, x_{i_k}]$ is equivalent to the relation \begin{align}\label{ch2:eq5}\phi(\omega)=(f_\phi, x_{\phi(e)}) \in F[x_{i_1}, \dots, x_{i_k}]\end{align} for permutations
$\phi \in S_r$, where $f_\phi(l_1, \dots, l_r) = f(l_{\phi(1)}, \dots, l_{\phi(r)})$ and $\phi(e)=(e_{\phi(1)}, \dots, e_{\phi(r)})$. This emergence of symmetry will inherently be reflected in the limit space, we will demonstrate this shortly. 

\paragraph{Special limits}

Let $K=[q]$. As we mentioned above, in general it is likely not to be fruitful to consider formula sequences as sequences of $C(K,r)$-colored $r$-graphs obeying certain symmetries due to constraint splitting. However, in the special case when each $r$-set of variables carries exactly one constraint we can derive a meaningful representation, MAX-CUT is an example. A direct consequence of \Cref{ch2:reprmain} is the following.

\begin{corollary}\label{ch2:reprcsp2}
	Let $r\geq 1$, and $K$ be a finite set, further, let $\KK \subset C(K,r)$, so that $\KK$ is permutation invariant. Let $(F_n)_{n=1}^\infty$ be a sequence of $r$CSP formulas with $|V(F_n)|$ tending to infinity, and each $r$-set of variables in each of the formulas carries exactly one constraint of type $\KK$. If for every formula $H$ obeying the same conditions the sequences $(t(H,F_n))_{n=1}^\infty$ converge as $(\KK,r)$-graphs, then there exists a $(\KK,r)$-digraphon $W \colon [0,1]^{\hhh([r])} \to \KK$ such that  $t(H,F_n) \to t(H,W)$ as $n$ tends to infinity for every $H$ as above. 
	
	Additionally, $W$ satisfies for each $x \in [0,1]^{\hhh([r])}$ and $\pi \in S_r$ that $W(x_{\pi(\hhh([r]))})=\hat \pi (W(x_{\hhh([r])}))$, where $\hat \pi$ is the action of $\pi$ on constraint types in $C(K,r)$ that permutes the rows and columns of the evaluation table according to $\pi$, that is  $(\hat \pi (f))(l_1, \dots, l_r) = f(l_{\pi(1)}, \dots, l_{\pi(r)})$.
	
\end{corollary}

\paragraph{General limits via evaluation}
In the general case of $r$CSP formulas we regard them as their evaluation representation $\eval$.

For $|K|=q$ we  identify the set of $r$CSP formulas with the set of arrays whose entries
are the sums of the evaluation tables of the constraints on $r$-tuples, that is $F$ with $V(F)=[n]$ corresponds to
a map $\eval(F)\colon[n]\times \dots \times [n] \to \{0, 1, \dots, d\}^{([q]^r)}$ that obeys the symmetry condition given after (\ref{ch2:evaldef}). This will be the way throughout the paper we look at these objects from here on. It seems that storing the whole structure of an $r$CSP formula does not provide any further insight, in fact splitting up constraints would produce non-identical formulas in a complete structure representation, which does not seem sensible.  


We denote the set $\{0, 1, \dots, d\}^{([q]^r)}$ by $L$ for simplicity, which one could also interpret as the set of multisets whose base set is $[q]^r$ and whose elements have multiplicity at most $d$. This perspective allows us to treat $r$CSPs as directed $r$-uniform hypergraphs whose edges are colored by the aforementioned elements of $L$, and leads to a representation of $r$CSP limits that is derived from the general representation of the limit set of $\Pi(L)$. We will show in a moment that the definition of convergence in the previous subsection given by densities of functional-colored graphs is basically identical to the convergence via densities of sub-multi-hypergraphs in the current case.

The definition of convergence for a general sequence of $r$CSP formulas, or equivalently of elements of $\Pi(L)$, was given in \Cref{ch2:defconv}. We describe here the special case for parallel multicolored graphs, see also \cite{LSzcom}.

Consider the evaluation representation of the $r$CSP formulas now as $r$-graphs whose oriented edges are parallel multicolored by $[q]^r.$
The map $\psi \colon \eval(H) \to \eval(F)$ is a homomorphism between two $r$CSP formulas $H$ and $F$ if it maps edges to edges of the same color from the color set $[q]^r$ and
is consistent when restricted to be a mapping between vertex sets, $\psi' \colon V(H) \to V(F)$, for simple graphs instead of CSP formulas this is the multigraph homomorphism notion.

Let $H$ be an $r$CSP formula, and let $\tilde H$ be the corresponding element in $C(L)$ on the same vertex set such that if the color on the fixed edge $e$ of $H$ is the $q$-sized $r$-array $(H^z(e))_{z \in [q]^r}$ with the entries being non-negative integers, then the color of $\tilde H$ at $e$ is $\prod\limits_{z \in [q]^r} x_{z}^{H^z(e)}$. More precisely, for an element $A \in L$ the value is given by \begin{align*}[\tilde H(e)](A)=\prod\limits_{z \in [q]^r} A(z)^{H^z(e)}.\end{align*} The linear space generated by the set \begin{align*}
\tilde L=\left\{ \,\prod\limits_{z \in [q]^r} x_{z}^{d_{z}} \mid  0 \leq d_{z_1,\dots, z_r}\leq d \,\right\}
\end{align*} forms an $L^\infty$-dense subset in $C(L)$, therefore \Cref{ch2:thm1} applies, and for a sequence $(G_n)_{n=1}^\infty$ requiring the convergence of $t(F,\eval(G_n))$ for all $F=\tilde H$ with $H \in \Pi(L)$ provides one of  the equivalent formulations of the convergence of $r$CSP formulas in the subformula density sense with respect to the evaluations.


The limit object will be given by \Cref{ch2:reprmain} as the space of measurable functions $W\colon[0,1]^{\hhh([r])} \to L$, where, as in the general case, the coordinates of the domain of $W$ are indexed by the non-empty subsets of $[r]$. In our case, not every possible $W$ having this form will serve as a limit of some sequence, the above mentioned symmetry in (\ref{ch2:eq5}) of the finite objects is inherited in the limit.

We state now the general evaluation $r$CSP version of \Cref{ch2:reprmain}.

\begin{corollary}\label{ch2:reprcsp}
	Let $(F_n)_{n=1}^\infty$ be a sequence of $r$CSP formulas that evaluate to at most $d$ on all $r$-tuples  with $|V(F_n)| \to \infty$  such that for  every finite $r$CSP formula $H$ obeying the same upper bound condition the sequence $(t(\tilde H,\eval(F_n)))_{n=1}^\infty$ converges. Then there exists an $(L,r)$-graphon  $W\colon [0,1]^{\hhh([r])} \to L$ such that $t(\tilde H,\eval(F_n)) \to t(\tilde H,W)$ for every $H$.
	Additionally,  $W$ satisfies for each $x \in [0,1]^{\hhh([r])}$ and $\pi \in S_r$ that $W(x_{\pi(\hhh([r]))})=\hat \pi (W(x_{\hhh([r])}))$, where $\hat \pi$ is as in \Cref{ch2:reprcsp2} when elements of $L$ are considered as maps from $[q]^r$ to non-negative integers.
\end{corollary}


\paragraph{Exchangeable partition-indexed processes}

We conclude the subsection with a remark that is motivated by the array representation of $r$CSPs. The next form presented seems to be the least redundant in some aspect, since no additional symmetry conditions have to be fulfilled by the limit objects.

The most natural exchangeable infinite random object fitting the one-to-one correspondence of \Cref{ch2:thm3}  with $r$CSP limits is the following process, that preserves every piece of information contained in the evaluation representation.
\begin{definition}
	Let  $N_q^r=\{ \,\P=(P_1, \dots, P_q) \mid \textrm{the sets } P_i\subset \N \textrm{ are pairwise disjoint and}$ $\sum_{i=1}^q |P_i|=r\, \}$ be the set of directed $q$-partitions of $r$-subsets of $\N$. We call the random process $(X_\P)_{\P \in N_q^r }$ that takes values in some compact Polish space $\KK$ a partition indexed process. The process $(X_\P)_{\P \in N_q^r }$ has the exchangeability property if its distribution is invariant under the action induced by finite permutations of $\N$, i.e., $(X_\P)_{\P \in N_q^r } \,{\buildrel d \over =}\, (X_{\rho^*(\P)})_{\P \in N_q^r }$ for any $\rho \in \Sym_0(\N)$.
\end{definition}\sindex[notions]{exchangeable array!partition-indexed processes}
Unfortunately, the existence of a representation theorem for partition-indexed exchangeable processes analogous to \Cref{ch2:exch} that offers additional insight over the directed colored $r$-array version is not established, and there is little hope in this direction. The reason for this is again the fact that there is no standard way of separating the generating process of the elements $X_\P$ and $X_{\P'}$ non-trivially in the case when $\P$ and $\P'$ have the same underlying base set of cardinality $r$ but are different as partitions into two non-trivial random stages with the first being identical for the two variables and the second stage being conditionally independent over the outcome of the first stage.

\section{Graph and graphon parameter testability}\label{ch3:sec.lim}

First we will invoke the method of sampling from $\KK$-colored $r$-graphs and $r$-graphons, as well as inspect the metrics that will occur later. 

Let $(U_{S})_{S \in \hhh([k], r)}$ be an independent uniform sample from $[0,1]$. Then for  an $r$-graph $G$, respectively an $r$-graphon $W$, the random $r$-graphs $\G(k,G)$ and $\G(k,W)$ have vertex set $[k]$, and edge weights $W_G((U_{p_e(S)})_{S \in \hhh(\hat e , r)})$, respectively $W((U_{p_e(S)})_{S \in \hhh(\hat e , r)}).$ Keep in mind, that $\G(k,G)\neq\G(k,W_G)$, the first term corresponds to sampling without, the second with replacement, but it is true that $\PPP(\G(k,G)\neq\G(k,W_G))\leq \frac{r^2}{|V(G)|}$. 


\paragraph{Norms and distances} We also mention the definitions of the norms and distances that will play a important role in what follows. In the next definition each object is real-valued. 

\begin{definition}
	
	The cut norm of an $n \times \dots \times n$ $r$-array $A$ is 
	$$
	\|A\|_\square=\frac{1}{n^r} \max_{S_1, \dots, S_r \subset [n]} \left| A(S_{1}, \dots, S_{r}) \right|,
	$$
	and the $1$-norm is
	$$
	\|A\|_1=\frac{1}{n^r} \sum_{i_1, \dots, i_r=1}^n |A(i_1, \dots, i_r)|.
	$$
	The cut distance of two labeled $r$-graphs or $r$-arrays $F$ and $G$ on the same vertex set $[n]$ is 
	$$
	d_\square(F,G)= \|F-G\|_\square,
	$$
	where $F(S_1, \dots, S_r)= \sum_{i_j \in S_j} F(i_1, \dots, i_r)$. The edit distance of the same pair is
	$$
	d_1(F,G)=  \|F-G\|_1.
	$$
	
\end{definition}

The continuous counterparts are described as follows.

\begin{definition}
	
	The cut norm of a naive $r$-graphon $W$ is 
	$$
	\|W\|_\square= \max_{S_1, \dots , S_r \subset [0,1]} \left|\int_{S_1 \times \dots \times S_r} W(x) \du x \right|,
	$$
	the cut distance of two naive $r$-graphons $W$ and $U$ is
	$$
	\delta_\square(W, U) = \inf_{\phi, \psi} \|W^\phi-U^\psi\|_\square,
	$$
	where the infimum runs over all measure-preserving permutations of $[0,1]$, and the graphon $W^\phi$ is defined as $W^\phi(x_1, \dots, x_r)= W(\phi(x_1),\dots,\phi(x_r))$. The cut distance for arbitrary unlabeled $r$-graphs or $r$-arrays $F$ and $G$ is
	$$
	\delta_\square(F,G)=\delta_\square(W_F,W_G).
	$$
\end{definition}

We remark that the above definition of the cut norm and distance is not satisfactory from one important aspect for $r \geq 3$: Not all sub-$r$-graph densities are continuous functions in the topology induced by this norm even in the most simple case, when $\KK=\{0,1\}$. Examples of subgraphs whose densities behave well with respect to the above norms are linear hypegraphs, that have the property that any two distinct edges intersect at most in one node. 


Originally, in \cite{BCL}, testability of $(\KK,r)$-graph parameters (which are real functions invariant under $r$-graph-isomorphisms) was defined as follows.
\begin{definition} \label{ch3:deftest}
	A $(\KK,r)$-graph parameter $f$ is testable, if for every $\varepsilon>0$ there exists a $k(\varepsilon) \in \mathbb{N}$ such that for every $k\geq k(\varepsilon)$ and simple $(\KK,r)$-graph $G$ on at least $k$ vertices 
	\[ \PPP(|f(G)-f(\mathbb{G}(k,G))|>\varepsilon)<\varepsilon.\]
\end{definition}

A $(\KK,r)$-graphon parameter $f$ is a functional on the space of $r$-graphons that is invariant under the action induced by measure preserving maps from $[0,1]$ to $[0,1]$, that is, $f(W)=f(W^\phi)$. Their testability is defined analogously to Definition \ref{ch3:deftest}.

\paragraph{Testing parameters}
A characterization of the testability of a graph parameter in terms of graph limits was developed in \cite{BCL} for $\KK=\{0,1\}$ in the undirected case, we will focus in the next paragraphs on this most simple setting and give an overview on previous work. Recall \Cref{ch3:deftest}.

\begin{theorem} \label{ch3:test}
	\cite{BCL} 
	Let $f$ be a simple graph parameter, then the following statements are equivalent.
	\begin{enumerate}[(i)]
		\item The parameter $f$ is testable.
		\item For every $\varepsilon>0$ there exists a $k(\varepsilon) \in \mathbb{N}$ such that for every $k\geq k(\varepsilon)$ and simple graph $G$ on at least $k$ vertices 
		\[ |f(G)-\E f(\mathbb{G}(k,G))| <\varepsilon.\]
		\item For every convergent sequence $(G_n)_{n=1}^\infty$  of simple graphs with $|V(G_n)| \to \infty$ the numerical sequence $(f(G_n))_{n=1}^\infty$ also converges. 
		\item For every $\varepsilon > 0$ there exist a $\varepsilon' > 0$ and a $n_0 \in \N$ such that for every pair $G_1$ and $G_2$ of simple graphs $|V(G_1)|,|V(G_2)| \geq n_0$ and $\delta_\square(G_1,G_2) < \varepsilon'$ together imply $|f(G_1)-f(G_2)| < \varepsilon$.
		\item There exists a $\delta_\square$-continuous functional $f'$ on the space of graphons, so that $f(G_n) \to f'(W)$ whenever $G_n \to W$.
		
	\end{enumerate}
\end{theorem}


A closely related  notion to parameter testing is property testing. 
A simple graph property $\P$ is characterized by the subset of the set of simple graphs containing the graphs which have the property, in what follows $\P$ will be identified with this subset.
\begin{definition}\cite{LSztest}
$\P$ is testable, if there exists another graph property $\P'$, such that
\begin{enumerate}[(a)]
\item $\PPP(\G(k,G) \in \P') \geq \frac{2}{3}$ for every $k \geq 1$ and $G \in \P$, and
\item for every $\varepsilon > 0$ there is a  $k(\varepsilon)$ such that for every $k \geq k(\varepsilon)$ and $G$ with $d_1(G,\P) \geq \varepsilon$ we have  that $\PPP(\G(k,G) \in \P') \leq \frac{1}{3}$.
\end{enumerate}
\end{definition}
Note that $\frac{1}{3}$ and $\frac{2}{3}$ in the definition can be replaced by arbitrary constants $0<a<b<1$, this change may alter the corresponding certificate $\P'$, but not the characteristic of testability.
The link below between the two notions is a simple consequence of the definitions. These concepts may be extended to the infinitary space of graphons, where a similar notion of sampling is available.
\begin{lemma} \cite{LSztest}
$\P$ is a testable graph property if and only if $d_1(.,\P)$ is a testable graph parameter.
\end{lemma}

We provide some remarks yielded by Theorem \ref{ch3:test}.

\begin{remark}
In the case $r=2$, the testability of a graphon parameter is equivalent to continuity in the $\delta_\square$ distance.
\end{remark}

\begin{remark}
	The intuitive reason for the absence of an analogous, easily applicable characterization of testability for higher rank uniform hypergraphs as in \Cref{ch3:test} is that no natural notion of a suitable distance is available at the moment. The construction of such a metric would require to establish a standard method to compare a large hypergraph $H_n$ to its random induced subgraph on a uniform sample. 
	
	The $\delta_\square$ metric for graphs is convenient because of its concise formulation and it induces a compact limit space, the main characteristic that is exploited that the total variation distance of probability measures of induced subgraphs of fixed size is continuous in this distance, any other $\delta_{\mathrm{var}}$ with this property would fit into the above framework. 
\end{remark}

\subsection{Examples of testable properties and parameters}
We introduce now a notion of efficient parameter testability\sindex[notions]{graph parameter!efficiently testable}\sindex[notions]{hypergraph parameter!efficiently testable}.
\Cref{ch3:deftest} of testability  does not ask for a specific upper bound on $k(\varepsilon)$ in terms of $\varepsilon$, but in applications the order of magnitude of this function may be an important issue once its existence has been verified.
Therefore we introduce 
a more restrictive class of graph parameters, we  refer to them as being efficiently testable. 

\begin{definition}\label{ch3:defeff}
	An $r$-graph parameter $f$ is called $\beta$-testable\sindex[notions]{graph parameter!$\beta$-testable} for a family of measurable functions $\beta=\{ \,\beta_i \mid \beta_i\colon\R^+ \to \R^+ ,i \in I\, \} $, if there exists an $i \in I$ such that for every $\varepsilon > 0$ and $r$-graph $G$ we have
	\[ \PPP(|f(G)-f(\mathbb{G}(\beta_i(\varepsilon),G))|>\varepsilon)<\varepsilon.\]
	
\end{definition}
With slight abuse of notation we will also use the notion of $\beta$-testability for a family containing only a single function $\beta$.
The term \emph{efficient testability} will serve as shorthand for $\beta$-testability for some (family)
of functions $\beta(\varepsilon)$ that are polynomial in $\frac{1}{\varepsilon}$.
One could rephrase this in the light of \Cref{ch3:defeff} by saying that a testable parameter $f$ is efficiently testable if its sample complexity is polynomial in $1/\varepsilon$.

We will often deal with statistics that are required to be highly concentrated around their mean, this might be important for us even if their mean is not known to us in advance. A quite universal tool for this purpose is a Chernoff-type large deviation result, the Azuma-Hoeffding-inequality for martingales with bounded jumps.  
Mostly, we require the formulation given below, see e.g. \cite{AlSp} for a standard proof and a wide range of applications. We will also apply a more elaborate version of this concentration inequality below.
\begin{lemma}[Azuma-Hoeffding-inequality]\label{ch3:azuma}
	Let $(M_k)_{k=0}^n$ be a super-martingale with the natural filtration such that with probability $1$ for every $k\in[n]$ we have $|M_k-M_{k-1}|\leq c_k$. Then for every $\varepsilon>0$ we have
	\begin{align*}
	\PPP ( |M_n -M_0| \geq \varepsilon) \leq 2 \exp\left(-\frac{\varepsilon^2 }{ 2\sum_{k=1}^n c^2_k} \right).
	\end{align*} 
\end{lemma}\sindex[notions]{Azuma-Hoeffding-inequality}

We will list some examples of graph parameters, for which there is information available about their sample complexity implicitly or explicitly in the literature.

\begin{example}
	One of the most basic testable simple graph parameters are subgraph densities $f_F(G)=t(F,G)$, where $F$ is a simple graph. The next result was formulated as Theorem 2.5 in \cite{LSzlim}, see also for hypergraphs Theorem 11 in \cite{ESz}.\end{example}
\begin{lemma}\cite{LSzlim,ESz}\label{ch3:subgraphcount}
	Let  $\varepsilon>0$ $q,r \geq 1$ be arbitrary. For any $q$-colored $r$-graphs $F$ and $G$, and integer $k\geq|V(F)|$ we have
	\begin{equation*} 
	\PPP(|t_\inj(F,G)-t_\inj(F, \G(k,G))|> \varepsilon) < 2\exp\left(-\frac{\varepsilon^2k}{2|V(F)|^2}\right),
	\end{equation*}
	and
	\begin{equation} 
	\label{ch3:sg_conc} \PPP(|t(F,G)-t(F, \G(k,G))|> \varepsilon) < 2\exp\left(-\frac{\varepsilon^2k}{18|V(F)|^2}\right).
	\end{equation}
	For any $q$-colored $r$-graphon $W$ we have
	\begin{equation*} 
	\PPP(|t(F,W)-t_\inj(F, \G(k,W))|> \varepsilon) < 2\exp\left(-\frac{\varepsilon^2k}{2|V(F)|^2}\right),
	\end{equation*}
	and
	\begin{equation*} 
	\PPP(|t(F,W)-t(F, \G(k,W))|> \varepsilon) < 2\exp\left(-\frac{\varepsilon^2k}{8|V(F)|^2}\right).
	\end{equation*}
\end{lemma}
This implies that for any $F$ that the parameter $f_F$ is $\mathcal{O}(\log(\frac{1}{\varepsilon})\varepsilon^{-2})$-testable. In the case of $(\KK,r)$-graphs for arbitrary $r$ the same as \Cref{ch3:subgraphcount} holds, this can be shown by a straightforward application of the Azuma-Hoeffding inequality, \Cref{ch3:azuma}, as in the original proofs.

\begin{example}\label{ch3:exgse}
	For $r=2$, $q,n \in \N$, $J \in \R^{q\times q}$, $h \in \R^q$, and $G \in \Pi^2_n$ we consider the energy
	\begin{align}\label{ch3:eq112}
	\EEE_\phi(G,J,h)=  \frac{1}{n^2}\sum_{1\leq i,j \leq q} J_{ij} e_G(\phi^{-1}(i),\phi^{-1}(j)) + \frac{1}{n}\sum_{1\leq i\leq q} h_i |\phi^{-1}(i)|,
	\end{align}
	of a partition $\phi \colon V(G) \to [q]$, and
	\begin{align}
	\hat \EEE(G,J,h)= \max_{\phi\colon V(G)\to[q]} \EEE_\phi(G,J,h),
	\end{align}\sindex[symbols]{e@$\EEE_\phi(G,J,h)$,$\hat \EEE(G,J,h)$}
	that is the \emph{ground state energy}\sindex[notions]{ground state energy} of the graph $G$ (cf.~\cite{BCL2}) with respect to $J$ and $h$, where $e_G(S,T)$ denotes the number of edges going form $S$ to $T$ in $G$. These graph functions originate from statistical physics, for the rigorous mathematical treatment of the topic see e.g. Sinai's book \cite{Si}. The energy expression whose maximum is sought is also referred to as a Hamiltonian.  
	In the literature this notion is also often to be found with negative sign or different normalization, more on this below.

	This  graph parameter can be expressed 
	in the terminology applied for MAX-$2$CSP.
	Let the corresponding $2$CSP formula to the pair ($G$,$J$) be $F$ with domain $K=[q]$.
	The formula $F$ is comprised of  the constraints $(g_0;x_{(i,j)})$ for every edge $(i,j)$ of $G$, where $g_0$ is the constraint type whose evaluation table is $J$, and additionally it contains $n$ copies of $(g_1;x_i)$ for every vertex $i$ of $G$, where $g_1$ is the constraint type in one variable with evaluation vector $h$. Then the optimal value of the objective function of the MAX-$2$CSP problem of the instance $F$ is equal to $\hat \EEE(G, J, h)$. Note that this correspondence is consistent with the sampling procedure, that is, to the pair ($\G(k,G)$,$J$) corresponds the $2$CSP formula $\G(k,F)$.
	Therefore $\hat \EEE(.,J,h)$ has sample complexity $\mathcal{O}(\frac{1}{\varepsilon^4})$(see \cite{AVKK2},\cite{MS}).
	
	These energies are directly connected to the number $\hom(G,H)$\sindex[symbols]{h@$\hom(G,H)$}  of admissible vertex colorings of $G$ by the colors $V(H)$ for a certain small weighted graph $H$.
	This was pointed out in \cite{BCL2}, (2.16), namely
	\begin{equation} \label{ch3:col}
	\frac{1}{|V(G)|^2}\ln\hom(G,H)=\hat \EEE(G,J)+\mathcal{O}\left(\frac{1}{|V(G)|}\right), 
	\end{equation}
	where the edge weights of $H$ are $\beta_{ij}(H)=\exp(J_{ij})$. The former line of thought of transforming ground state energies into MAX-$2$CSPs  is also valid in the case of $r$-graphs and $r$CSPs for arbitrary $r$.
	
	The results on the sample complexity of MAX-$r$CSP for $q=2$ can be extended beyond the case of simple 
	hypergraphs, higher dimensional Hamiltonians are also expressible as $r$CSP formulas. The generalization for arbitrary $q$ and to $r$-graphons will follow in the next section. Additionally we note, that an analogous statement to (\ref{ch3:col}) on testability of coloring numbers does not follow immediately for $r \geq 3$.
	
	On the other hand, with the notion of the ground state energy available, we may rewrite the MAX-$2$CSP  in a compact form as an energy problem. We will execute this task right away for limit objects.
	First, we introduce the ground state energy of a $2$-kernel with respect to an interaction matrix $J$. The collection $\phi=(\phi_1, \dots, \phi_q)$ is a fractional $q$-partition\sindex[notions]{fractional $q$-partition} of $[0,1]$ with the components being measurable non-negative functions on $[0,1]$, if for every $x \in [0,1]$ it holds that $\sum_{i=1}^q \phi_i(x)=1$.\end{example}
\begin{definition}\label{ch3:defgse}
	Let  $q \geq 1$, $J \in \R^{q \times q}$. Then the ground state energy of the  $2$-kernel $W$ with respect to $J$ is
	\begin{equation*} 
	\EEE (W,J)=\max_\phi \sum_{z \in [q]^2} J_z\int_{[0,1]^2} \phi_{z_1}(x)\phi_{z_2}(y) W(x,y) \du x \du y, 
	\end{equation*}	
	where $\phi$ runs over all fractional $q$-partitions of $[0,1]$.
\end{definition}\sindex[symbols]{e@$\EEE (W,J)$}
Let $K=[q]$, $L=\{0,1, \dots, d\}^{[q]^2}$ and $(F_n)_{n=1}^\infty$ be a convergent sequence of $2$CSP formulas. Consider the corresponding sequence of graphs $\eval(F_n)=(\tilde F_n^z)_{z \in [q]^2}$ for each $n$, and let $W=(W^z)_{z \in [q]^2}$ be the respective limit. Let $f$ be the $(L,2)$-graph parameter so that $f(\eval(F))$ is equal to the density of the MAX-$2$CSP value for the instance $F$. Then it is not hard to see that $f$ can be extended to the limit space the following way
$$
f(W)=\max_\phi  \sum_{i,j=1}^q \int_{[0,1]^2} \phi_i(x) \phi_j(y) W^{(i,j)}(x,y) \du x \du y,
$$
where $\phi$ runs over all fractional $q$-partitions of $[0,1]$. The formula is a special case of the \emph{layered ground state energy} with the interaction matrices defined by $J^{i,j}(k,l)=\I_i(k) \I_j(l)$ that is defined below.


\begin{example}
	The efficiency of testing a graph parameter can be investigated in terms of some additional continuity condition in the $\delta_\square$ metric. Direct consequence of results from \cite{BCL} will be presented in the next lemma.\end{example}
\begin{lemma} \label{ch3:cont} Let $f$ be a simple graph parameter that is $\alpha$-H\"older-continuous\sindex[notions]{H\"older-continuity} in the $\delta_\square$ metric in the following sense: There exists a $C >0$ such that for every $\varepsilon > 0$ there exists $n_0(\varepsilon)$ so that if for the simple graphs $G_1$, $G_2$ it holds that $|V(G_1)|,|V(G_2)| \geq n_0(\varepsilon)$ and $\delta_\square(G_1,G_2)\leq \varepsilon$, then $|f(G_1)-f(G_2)| \leq C \delta^\alpha_\square(G_1,G_2)$. Then $f$ is $\max\{2^{\mathcal{O}\left(\frac{1}{\varepsilon^{2/\alpha}}\right)}, n_0(\varepsilon)\}$-testable.
\end{lemma}
\begin{proof}
	To see this, let us fix $\varepsilon >0$. Then for an arbitrary simple graph $G$ with $|V(G)| \geq  n_0(\varepsilon)$ and $k \geq n_0(\varepsilon)$ we have
	\begin{equation} \label{ch3:cuttest}
	|f(G)-f(\G(k,G))| \leq C \left[\delta_\square(G, \G(k,G))\right]^\alpha < C \left(\frac{10}{\sqrt{\log_2k}}\right)^\alpha,
	\end{equation}
	with probability at least $1-\exp(-\frac{k^2}{2 \log_2 k})$. The last probability bound  in (\ref{ch3:cuttest}) is the statement of Theorem 2.9 of \cite{BCL}. We may rewrite (\ref{ch3:cuttest}) by setting $\varepsilon=C \left(\frac{10}{\sqrt{\log_2k}}\right)^\alpha$, the substitution implies that $f$ is $2^{\mathcal{O}\left(\varepsilon^{-2/\alpha}\right)}$-testable, whenever $n_0(\varepsilon)\leq 2^{\mathcal{O}(\varepsilon^{-2/\alpha})}$. 
\end{proof}
This latter approach is hard to generalize in a meaningful way to $r$-graphs for $r\geq 3$ because of the absence of a suitable metric, see the discussion above. The converse direction, namely formulating a qualitative statement about the continuity of $f$ with respect to $\delta_\square$ obtained from the information about the sample complexity is also a worthwhile problem.

\section{Testability of the ground state energy} \label{ch4:sec:gse}

Assume that $\KK$ is a compact Polish space, and $r$ is a positive integer. First we provide the basic definition of the energy of a $(\KK,r)$-graphon $W \colon [0,1]^{\hhh([r])} \to \KK$ with respect to some $q \geq 1$, an $r$-array $J \in C(\KK)^{q \times  \dots \times q}$, and a fractional partition $\phi=(\phi_1, \dots , \phi_q)$. With slight abuse of notation, the graphons in the upcoming parts of the section assume both the $\KK$-valued and the probability measure valued form, it will be clear from the context which one of them is meant.

Recall \Cref{ch3:defenergycont} of the energies of naive $r$-kernels, the version for true $(\KK,r)$-graphons is
\begin{align}\label{ch3:eq333}
\EEE_\phi(W, J)=\sum_{z_1,\dots,z_r=1}^{q}  \int_{[0,1]^{\hhh([r])}}  J_{z_1,\dots, z_r} (W(x_{\hhh([r])})) \prod_{j=1}^{r} \phi_{z_j}(x_{\{j\}}) \du \lambda(x_{\hhh([r])}).
\end{align}

The value of the above integral can be determined by first integrating over the coordinates corresponding to subsets of $[r]$ with at least two elements, and then over the remaining ones. The interior partial integral is then not dependent on $\phi$, so it can be calculated in advance in the case when we want to optimize over all choices of fractional partitions. Therefore focusing attention on the naive kernel version does not lead to any loss of generality in terms of testing, see below.

When dealing with a so-called integer partition\sindex[notions]{integer $q$-partition} $\phi=(\I_{T_1}, \dots, \I_{T_q})$, one is able  to rewrite the former expression (\ref{ch3:eq333}) as
$$
\EEE_\phi(W, J)=\sum_{z_1,\dots, z_r=1}^{q}  \int_{p_{\hhh([r],1)}^{-1}(T_{z_1}\times \dots \times T_{z_r})}  J_{z_1,\dots, z_r} (W(x_{\hhh([r])}) ) \du \lambda(x_{\hhh([r]}),
$$
where $p_D$ stands for the projection of $[0,1]^{\hhh([r])}$ to the coordinates contained in the set $D$. 

The energy of a $(\KK,r)$-graph $G$ on $k$ vertices with respect to the $J \in C(\KK)^{q \times  \dots \times q}$ for the fractional $q$-partition\sindex[notions]{fractional $q$-partition} $x_n=(x_{n,1}, \dots, x_{n,q})$ for $n=1, \dots, k$ (i. e., $x_{n, m} \in [0,1]$ and $\sum_m x_{n,m}=1$) is defined as
\begin{equation} \label{ch4:deffinen}
\EEE_{\mathrm x}(G,J)=\frac{1}{k^r}\sum_{z_1,\dots,z_r=1}^{q}  \sum_{n_1,\dots,n_r=1}^{k} J_{z_1,\dots, z_r} (G(n_1, \dots, n_r)) \prod_{j=1}^r x_{n_j,z_j}.
\end{equation}

In the case when $\KK=\{0,1\}$ and $J_{z_1,\dots, z_r}(x)=a_{z_1,\dots, z_r} \I_{1}(x)$ is a constant multiple of the indicator function of $1$ we retrieve the original GSE notion in \Cref{ch3:exgse} and \Cref{ch3:defgse}.

\begin{remark} 
	Ground state energies and subgraph densities are Lipschitz continuous graph parameters in the sense of \Cref{ch3:cont} (\cite{BCL},\cite{BCL2}), but that result implies much weaker upper bounds on the sample complexity, than the best ones known to date. This is due to the fact, that $\delta_\square(G, \G(k,G))$ decreases with magnitude $1/\sqrt{\log k}$ in $k$, which is the result of the difficulty of finding a near optimal overlay between two graphons through a measure preserving permutation of $[0,1]$ in order to calculate their $\delta_\square$ distance. On the other hand, if the sample size $k(\varepsilon)$ is exponentially large in $1/\varepsilon$, then the distance $\delta_\square(G, \G(k,G))$ is small with high probability, therefore 
	all H\"older-continuous graph parameters at $G$ can be estimated simultaneously with high success probability by looking at the values at $\G(k,G))$.
\end{remark}

Next we introduce the layered version of the ground state energy\sindex[notions]{ground state energy!layered}. This is a generalized optimization problem where we wish to obtain the optimal value corresponding to fractional partitions of the sums of energies over a finite layer set. 

\begin{definition}\label{ch4:def.lay.gse}
	Let $\eE$ be a finite layer set, $\KK$ be a compact set, and $W=(W^\ee)_{\ee \in \eE}$ be a tuple of $(\KK, r)$-graphons. Let $q$ be a fixed positive integer and let $J=(J^\ee)_{\ee \in \eE}$ with $J^\ee \in C(\KK)^{q\times \dots \times q}$ for every  $\ee \in \eE$.  For a  $\phi=(\phi_1, \dots, \phi_q)$ fractional $q$-partition of $[0,1]$ let 
	$$
	\EEE_\phi (W,J)=\sum_{\ee \in \eE} \EEE_\phi(W^\ee, J^\ee)
	$$
	and let   
	$$
	\EEE (W,J)=\max_{\phi} \EEE_\phi (W,J),
	$$
	denote the layered ground state energy, where the maximum runs over all fractional $q$-partitions of $[0,1]$. 
	
	We define for $G=(G^\ee)_{\ee \in \eE}$ the energy $\EEE_{\mathrm x}(G,J)$ analogously as the energy sum over $\eE$, see (\ref{ch4:deffinen}) above, and $\hat \EEE(G,J)=\max_{\mathrm x} \EEE_{\mathrm x} (G,J)$ where the maximum runs over integer $q$-partitions\sindex[notions]{integer $q$-partition} ($x_{n,m} \in\{0,1\}$ ), respectively $\EEE(G,J)=\max_{\mathrm x} \EEE_{\mathrm x} (G,J)$, where the maximum is taken over all fractional $q$-partitions ${\mathrm x}$.
\end{definition}

Now we will rewrite the unweighted boolean limit MAX-$r$CSP (recall \Cref{ch2:defcsp2}) as a layered ground state energy problem. Let $\eE=\{0,1\}^r$,  $\KK=\{0,1, \dots, 2^r\}$, $W=(W^z)_{z \in \{0,1\}^r}$ with $W^z$ being $(\KK,r)$-graphons, and let  
\begin{equation}
\alpha(W)= \max\limits_{\phi}  \sum_{\substack{z \in \{0,1\}^r}}  \int_{\substack{[0,1]^{\hhh([r])}}} \prod_{j = 1}^r \phi(x_{\{j\}})^{z_j} (1-\phi(x_{\{j\}}) )^{1-z_j} W^z (x) \du \lambda(x), \nonumber
\end{equation}
where the maximum is taken over all measurable functions $\phi\colon[0,1] \to [0,1]$.
If $\eval(F)=(F^z)_{z \in \{0,1\}^r}$ is a $(\KK^{\eE},r)$-graph corresponding to a boolean $r$CSP formula $F$ with $k$ variables, then the finite integer version of $\alpha$ is given by 
\begin{align*}
\hat \alpha(\eval(F))=\max_{\mathrm{x}} \frac{1}{k^r}\sum_{\substack{z \in \{0,1\}^r}}  \sum_{n_1,\dots,n_r=1}^{k} F^z(n_1, \dots, n_r) \prod_{j=1}^r x_{n_j,z_j},
\end{align*} 
where the maximum runs over integer $2$-partitions of $[k]$. It is clear that $\hat \alpha(\eval(F))$ is equal to the density of the optimum of the MAX-$r$CSP problem of $F$.

We return to the general setting and summarize the  
involved parameters in the layered ground state energy problem. These are the dimension $r$, the layer set $\eE$, the number of states $q$, the color set $\KK$, 
the finite or limit case. Our main theorem on the paper will be a generalization of the following theorem on sample complexity of $r$CSPs with respect to these factors.

The main result of \cite{AVKK2} was the following.
\begin{theorem} \label{ch4:cspcompl}
	\cite{AVKK2} 
	Let $F$ be an unweighted boolean  $r$CSP formula. Then for any $\varepsilon >0$ and $\delta >0$ we have that for $k \in {\mathcal O}(\varepsilon^{-4} \log(\frac{1}{\varepsilon}))$ it holds that
	$$
	\PPP\left(|\hat\alpha(\eval(F))-\hat\alpha(\G(k,\eval(F)))|>\varepsilon\right)<\delta.
	$$
\end{theorem}

The upper bound on $k$ in the above result was subsequently improved by \citet{MS} to $k \in \mathcal{O}(\varepsilon^{-4})$.
We will see in what follows that also the infinitary version of the 
above statement is true.
It will be stated in terms of layered ground state energies of edge colored hypergraphs, and will settle the issue regarding the efficiency of testability of the mentioned parameters in the greatest generality with respect to the previously highlighted aspects. However, what the exact order of the magnitude of the sample complexity of the MAX-$r$CSP and the GSE problem is remains an open question.

In order to simplify the analysis we introduce the \emph{canonical form} \sindex[notions]{ground state energy!canonical form}of the problem, that denote layered ground state energies of $[q]^r$-tuples of $([-d,d],r)$-graphons 
with the special interaction $r$-arrays\sindex[notions]{interaction $r$-array} $\hat J^z$ for each $z \in [q]^r$, that have the identity function $f(x)=x$ as the $(z_1, \dots, z_r)$ entry and the constant $0$ function for the other entries. In most of what follows we will drop the dependence on $J$ in the energy function when it is clear that we mean the aforementioned canonical $\hat J$, and will employ the notation $\EEE_x(G)$, $\EEE(G)$, $\hat \EEE(G)$, $\EEE_\phi (W)$, and $\EEE(W)$\sindex[symbols]{e@$\EEE_x(G)$, $\EEE(G)$, $\hat \EEE(G)$, $\EEE_\phi (W)$, $\EEE(W)$} (dependence on $q$ is hidden in the notation), where $G$ and $W$ are $[q]^r$-tuples of $([-d, d],r)$-graphs and graphons, respectively. We are ready to state the main result of the paper.

\begin{theorem} \label{ch4:main}
	Let $r\geq 1$, $q \geq 1$, and $\varepsilon >0$. Then for any $[q]^r$-tuple of $([-\|W\|_\infty,\|W\|_\infty],r)$-graphons $W=(W^z)_{z \in [q]^r}$ and $k \geq \Theta^4  \log(\Theta) q^r$ with $\Theta=\frac{2^{r+7}q^r r}{\varepsilon}$  we have 
	\begin{align}\label{ch4:eq1}
	\PPP(|\EEE(W)-\hat \EEE(\G (k,W))|>\varepsilon \|W\|_\infty)<\varepsilon.
	\end{align}
\end{theorem}

A direct consequence of \Cref{ch4:main} is the corresponding result for layered ground state energies.

\begin{corollary} \label{ch4:maincor2}
	Let $\eE$ be a finite layer set, $\KK$ a compact Polish color set, $q \geq 1$, $r$-arrays $J=(J^\ee)_{\ee \in \eE}$ with $J^\ee \in C(\KK)^{q \times \dots \times q}$, and $\varepsilon >0$. Then we have that for any $\eE$-tuple of $(\KK,r)$-graphon $W=(W^\ee)_{\ee \in \eE}$ and $k \geq \Theta^4  \log(\Theta) q^r$ with $\Theta=\frac{2^{r+7}q^r r}{\varepsilon}$ that
	$$
	\PPP(|\EEE(W,J)-\hat \EEE(\G (k,W),J)|>\varepsilon |\eE| \, \|J\|_\infty \,\|W\|_\infty)<\varepsilon.
	$$
\end{corollary}

\begin{proof}
	We make no specific restrictions on the color set $\KK$ and on the set $\eE$ of layers except for finiteness of the second, therefore it will be convenient to rewrite the layered energies $\EEE_\phi(W,J)$ into a more universal form as a sum of proper Hamiltonians in order to suppress the role of $\KK$ and $\eE$. Let
	\begin{align}
	\EEE_\phi(W,J)&=\sum_{\ee \in \eE} \sum_{z_1, \dots , z_r \in [q]}  \int_{[0,1]^{\hhh([r])}} \prod_{j \in [q]} \phi_{z_j}(x_{\{j\}})  J^\ee_{z_1, \dots , z_r}(W^\ee(x) ) \du \lambda(x_{\hhh([r])})  \nonumber \\
	&=\sum_{z_1, \dots , z_r \in [q]} \int_{[0,1]^{\hhh([r])}} \prod_{j \in [r]} \phi_{z_j}(x_{\{j\}}) \left[ \sum_{\ee \in \eE}  J^\ee_{z_1, \dots , z_r}(W^\ee(x) )  \right] \du \lambda(x_{\hhh([r])}).\nonumber
	\end{align}
	Motivated by this reformulation we introduce for every $(W,J)$ pair a special auxiliary instance of the ground state problem that is defined for a $[q]^r$-tuple of $([-d,d], r)$-graphons, where $d= |\eE| \, \|J\|_\infty \, \|W\|_\infty$. For any $z \in [q]^r$, let  $\hat W^z(x)= \sum_{\ee \in \eE}  J^\ee_{z_1, \dots , z_r}(W^\ee(x))$ for each $x \in [0,1]^{\hhh([r])}$, and let the interaction matrices $\hat J^z$ be of the canonical form. We obtain for any fractional partition $\phi$ of  $[0,1]$ into $q$ parts that $\EEE_\phi(W,J)=\EEE_\phi(\hat W,\hat J)$, and also $\EEE_{\mathrm{x}}(\G (k,W),J)=\EEE_\mathrm{x}(\G (k,\hat W),\hat J)$ for any fractional partition $\mathrm{x}$, where the two random $r$-graphs are obtained via the same sample. Therefore, without loss of generality, we are able to reduce the statement of the corollary to the statement of \Cref{ch4:main} dealing with ground state energies of canonical form. 
\end{proof}



We start with the proof of \Cref{ch4:main} by providing the necessary background. We will proceed loosely along the lines of the proof of \Cref{ch4:cspcompl} from \cite{AVKK2} with most of the required lemmas being refinements of the respective ones in the proof of that theorem. We will formulate and verify these auxiliary lemmas one after another, afterwards we will compile them to prove the main statement. The arguments made in \cite{AVKK2} carry through adapted to our continuous setting with some modifications, and we will also draw on tools from \cite{BCL} and \cite{BCL2}. The first lemma tells us that in the real-valued case the energy of the sample and that of the averaged sample do not differ by a large amount.
\begin{lemma}\label{ch4:weightedlemma}
	Let $W$ be a $([-d,d],r)$-graphon, $q \geq 1$, $ J\in \R^{q \times \dots \times q}$. Then for every $k \geq 1$ there is a coupling of $\G(k,W)$ and $\hh(k,W)$ such that
	\begin{equation*}
	\PPP\left( |\hat \EEE(\G(k,W),J) - \hat \EEE(\hh(k,W), J) |> \varepsilon \|J\|_\infty\|W\|_\infty \right)  \leq 2\exp\left(-k\left(\frac{\varepsilon^2 k}{2} - \log q\right) \right)
	\end{equation*}
\end{lemma}
\begin{proof}
	Let us fix a integer $q$-partition $\mathrm{x}$ of $[k]$, and furthermore let the two random $r$-graphs be generated by the same sample $(U_{S})_{S \in \hhh([k], r)}$. Then
	\begin{align*}
	\hat \EEE_{\mathrm x}(\G(k,W),J) &= \frac{1}{k^r} \sum_{z_1,\dots,z_r=1}^{q}  \sum_{n_1,\dots,n_r=1}^{k} J_{z_1,\dots, z_r} W((U_{S})_{S \in \hhh(\{n_1, \dots, n_r\}, r)}) \prod_{j=1}^r x_{n_j,z_j}, 
	\end{align*}
	and
	\begin{align*}
	&\hat \EEE_{\mathrm x}(\hh(k,W),J)  \\ 
	&= \frac{1}{k^r}\sum_{z_1,\dots,z_r=1}^{q}  \sum_{n_1,\dots,n_r=1}^{k}   J_{z_1,\dots, z_r} \E[W((U_{S})_{S \in \hhh(\{n_1, \dots, n_r\}, r)})\mid  (U_{S})_{S \in \hhh(\{n_1, \dots, n_r\}, 1)}]  \prod_{j=1}^r x_{n_j,z_j}.&
	\end{align*}
	Let us enumerate the elements of ${k \choose 2}$ as $e_1, e_2, \dots, e_{k \choose 2}$, and define the martingale 
	\begin{align*}
	Y_0=\E[\hat \EEE_{\mathrm x}(\G(k,W),J) \mid \{ \,U_j\mid  j\in [k]\, \}], 
	\end{align*} and 
	\begin{align*}Y_t=\E\left[\hat \EEE_{\mathrm x}(\hh(k,W),J)\mid  \{ \,U_j\mid  j\in [k]\, \}\cup \left(\cup_{j=1}^t \{ \,U_{S}\mid  e_j \subset S\, \}\right)\right]
	\end{align*}
	for each $1\leq t\leq {k \choose 2}$, so that $Y_0=\hat \EEE_{\mathrm x}(\hh(k,W),J)$ and $Y_{k \choose 2}=\hat \EEE_{\mathrm x}(\G(k,W),J).$ For each $t \in {k \choose 2}$ we  can upper bound the difference, $|Y_{t-1}-Y_t|\leq \frac{1}{k^2} \|J\|_\infty \|W\|_\infty$. By the Azuma-Hoeffding inequality, \Cref{ch3:azuma}, it follows that
	\begin{equation}
	\PPP(|Y_t-Y_0| \geq \rho) \leq 2\exp\left(-\frac{\rho^2k^4}{2{k \choose 2}  \|J\|_\infty^2 \|W\|_\infty^2}\right)\leq2\exp\left(-\frac{\rho^2k^2}{2  \|J\|_\infty^2 \|W\|_\infty^2}\right),
	\end{equation} 
	for any $\rho>0$.
	
	There are $q^k$ distinct integer $q$-partitions of $[k]$, hence
	\begin{align}
	\PPP\left( |\hat \EEE(\G(k,W),J) - \hat \EEE(\hh(k,W), J) |> \varepsilon \|J\|_\infty\|W\|_\infty \right)  \leq 2\exp\left(-k\left(\frac{\varepsilon^2 k}{2} - \log q\right) \right).  
	\end{align}
\end{proof}

In the following lemmas every $r$-graph or graphon is meant to be as bounded real-valued and directed.

We would like to point out in the beginning that in the finite case we are able to shift from the integer optimization problem to the relaxed one with having a reasonably good upper bound on the difference of the optimal values of the two.
\begin{lemma} \label{ch4:intcont}
	Let $G$ be a real-valued $r$-graph on $[k]$ and $ J\in \R^{q \times \dots \times q}$. Then
	$$
	|\EEE(G,J) - \hat \EEE(G,J)| \leq r^2 \frac{1}{2k} \|G\|_\infty \|J\|_\infty.
	$$
\end{lemma}
\begin{proof}
	Trivially we have $\EEE(G,J) \geq \hat \EEE(G,J)$. We define $G'$ by setting all entries of $G$ to $0$ which have at least two coordinates which are the same (for $r=2$ these are the diagonal entries). Thus, we get that 
	$$
	|\EEE(G,J) - \EEE(G',J)| \leq {r \choose 2} \frac{1}{k} \|G\|_\infty \|J\|_\infty.
	$$
	Now assume that we are given a fractional partition $\overline{\mathrm{x}}$ so that  $\EEE_{\overline{\mathrm{x}}}(G',J)$ attains the maximum $\EEE(G',J)$. We fix all the entries  $\overline x_{n,1}, \dots \overline x_{n,q}$ of $\overline{\mathrm{x}}$ with $n=2, \dots, k$ and regard $\EEE_{\overline{\mathrm{x}}}(G',J)$ as a function of $x_{1,1}, \dots , x_{1,q}$. This  function will be linear in the variables $x_{1,1}, \dots , x_{1,q}$, and with the additional condition $\sum_{j=1}^r x_{1,j}=1$ we obtain a linear program. By standard arguments this program possesses an integer valued optimal solution, so we are allowed to replace $\overline x_{1,1}, \dots , \overline x_{1,q}$ by integers without letting $\EEE_{\overline{\mathrm{x}}}(G',J)$ decrease. We repeat this procedure for each $n \in [k]$, obtaining an integer optimum for  $\EEE_{\overline{\mathrm{x}}}(G',J)$, which implies that $\EEE(G',J) = \hat\EEE(G',J)$. Hence, the claim follows.
\end{proof}

Next lemma is the continuous generalization of Theorem 4 from \cite{AVKK2}, and is closely related to the Weak Regularity Lemma, \Cref{ch3:weakreglemma}, of \cite{FK}, and its continuous version \Cref{ch3:wreg}. The result is a centerpiece of the cut decomposition method.\sindex[notions]{cut decomposition method}
\begin{lemma} \label{ch4:cutapprox}
	Let $\varepsilon>0$ arbitrary. For any bounded measurable function $W\colon [0,1]^r \to \R$ there exist an $s \leq \frac{1}{\varepsilon^2}$, measurable sets $S_i^j \subset [0,1]$ with $i=1, \dots, s$, $j=1, \dots, r$, and real numbers $d_1, \dots, d_s$ so that with  $B=\sum_{i=1}^s d_i \I_{S_i^1 \times \dots \times S_i^r}$ it holds that
	\begin{enumerate}[(i)]
		\item $\|W\|_2 \geq \|W-B\|_2$, 
		\item $\|W-B\|_\square < \varepsilon \|W\|_2$, and
		\item $\sum_{i=1}^s |d_i| \leq \frac{1}{\varepsilon}\|W\|_2$.
	\end{enumerate}
\end{lemma}
\begin{proof}
	We construct stepwise the required rectangles and the respective coefficients implicitly. Let $W^0=W$, and suppose that after the $t$'th step of the construction we have already obtained every set $S_i^j \subset [0,1]$ with $i=1, \dots, t$, $j=1, \dots, r$, and the real numbers $d_1, \dots, d_t$. Set $W^t=W-\sum_{i=1}^t d_i \I_{S_i^1 \times \dots \times S_i^r}$.
	We proceed to the $(t+1)$'st step, where two possible situations can occur. The first case is when
	$$\|W^t\|_\square \geq \varepsilon \|W\|_2.$$
	This implies by definition that there exist measurable subsets $S_{t+1}^1,\dots, S_{t+1}^r$ of $[0,1]$ such that $|\int_{S_{t+1}^1 \times \dots \times S_{t+1}^r} W^t(x) \du \lambda(x)| \geq \varepsilon \|W\|_2$. We define $d_{t+1}$ to be the average of $W^t$ on the product set $S_{t+1}^1 \times \dots \times S_{t+1}^r$, and proceed to the $(t+2)$'nd step. In the case of 
	$$\|W^t\|_\square < \varepsilon \|W\|_2$$ 
	we are ready with the construction and set $s=t$.
	
	We analyze the first case to obtain an upper bound on the total number of steps required by the construction. So suppose that the first case above occurs. Then
	\begin{align}
	\|W^t\|^2_2-\|W^{t+1}\|^2_2 &= \int_{\substack{S_{t+1}^1 \times \dots \times S_{t+1}^r}} (W^t)^2(x) \du \lambda(x)  -   \int_{\substack{S_{t+1}^1 \times \dots \times S_{t+1}^r}} (W^t(x)-d_{t+1})^2 \du \lambda(x)  \nonumber \\
	&= d_{t+1}^2 \lambda(S_{t+1}^1)  \dots \lambda(S_{t+1}^r) \geq \varepsilon^2\|W\|_2^2. \label{ch4:step}
	\end{align}
	This means that the square of the $2$-norm of $W^t$ decreases in $t$ in every step when  the first case occurs in the construction by at least $\varepsilon^2\|W\|_2^2$, therefore  it can happen  only  at most $\frac{1}{\varepsilon^2}$ times, with other words $s \leq \frac{1}{\varepsilon^2}$.
	It is also clear that the $2$-norm decreases in each step, so we are left to verify the upper bound on the sum of the absolute values of the coefficients $d_i$. From (\ref{ch4:step}) we get, that
	$$
	\|W\|^2_2=\sum_{t=1}^s \|W^{t-1}\|^2_2-\|W^{t}\|^2_2\geq \sum_{t=1}^s d^2_t \lambda(S_{t}^1)  \dots \lambda(S_{t}^r).
	$$
	We also know for every $t \leq  s$ that $|d_t| \lambda(S_{t}^1)  \dots \lambda(S_{t}^r) \geq  \varepsilon\|W\|_2$. Hence,
	$$
	\sum_{t=1}^{s} |d_t| \varepsilon\|W\|_2 \leq \sum_{t=1}^s d^2_t \lambda(S_{t}^1)  \dots \lambda(S_{t}^r) \leq \|W\|_2^2,
	$$
	and therefore $\sum_{t=1}^{s} |d_t| \leq \frac{1}{\varepsilon}\|W\|_2.$
\end{proof}
Next we state that the cut approximation provided by \Cref{ch4:cutapprox} is invariant under sampling. This is a crucial point of the whole argument, and is the $r$-dimensional generalization of Lemma 4.6 from \cite{BCL}.

\begin{lemma} \label{ch4:cutpres}
	For any $\varepsilon>0$ and bounded measurable function $W\colon [0,1]^r \to \R$  we have that
	$$
	\PPP\left(\left| \|\hh(k,W)\|_\square -\|W\|_\square\right| > \varepsilon \|W\|_\infty \right) < 2\exp\left(-\frac{\varepsilon^2 k}{32r^2}\right) 
	$$
	for every $k\geq\left(\frac{16r^2}{\varepsilon}\right)^4$.
\end{lemma}
\begin{proof}
	Fix an arbitrary $0 <\varepsilon < 1$, $r\geq 2$, and further let $W$ be a real-valued naive $r$-kernel. Set the sample size to $k\geq\left(\frac{16r^2}{\varepsilon}\right)^4$. Let us consider the array representation of $\hh(k,W)$ and denote the $r$-array $A_{\hh(k,W)}$ by $G$ that has zeros on the diagonal.
	We will need the following lemma from \cite{AVKK2}.
	\begin{lemma} \label{ch4:aux}
		$G$ is a real $r$-array on some finite product set $V_1 \times \dots \times V_r$, where $V_i$ are copies of $V$ of cardinality $k$. Let $S_1 \subset V_1, \dots, S_r \subset V_r$ be fixed subsets and $Q_1$ a uniform random subset of $V_2 \times \dots \times V_r$ of cardinality $p$. Then
		$$
		G(S_1, \dots, S_r) \leq  E_{Q_1} G(P(Q_1 \cap S_2 \times \dots \times S_r),S_2, \dots,S_r) + \frac{k^r}{\sqrt{p}} \|G\|_2,
		$$
		where $P(Q_1)=P_G(Q_1)=\{ \,x_1 \in V_1 \mid \sum_{(y_2,\dots y_r)\in Q_1}G(x_1,y_2,\dots,y_r) >0\, \}$ and the $2$-norm denotes $\|G\|_2=\left( \frac{\sum_{x_i \in V_i} G^2(x_1, \dots, x_r) }{|V_1| \dots |V_r|} \right)^{1/2}$.
	\end{lemma}
	If we apply \Cref{ch4:aux} repeatedly $r$ times to the $r$-arrays $G$ and $-G$, then we arrive at an upper bound on $G(S_1,\dots,S_r)$ ($(-G)(S_1, \dots, S_r)$ respectively) for any collection of the $S_1, \dots, S_r$ which does not depend on the particular choice of these sets any more, so we get that
	\begin{align}
	k^r \|G\|_\square &\leq E_{Q_1, \dots, Q_r} \max_{Q'_i \subset Q_i} \max\{G(P_G(Q'_1), \dots, P_G(Q'_r));  (-G)(P_{-G}(Q'_1), \dots, P_{-G}(Q'_r))\} \nonumber \\ & \quad + \frac{rk^r}{\sqrt{p}} \|G\|_\infty,\label{ch4:unifbd}
	\end{align}
	since $\|G\|_2 \leq \|G\|_\infty$.
	
	Let us recall that $G$ stands for the random $\hh(k,W)$. We are interested in the expectation $\E$ of the left hand side of (\ref{ch4:unifbd}) over the sample that defines $G$. Now we proceed via the method of conditional expectation. We establish an  upper bound on the expectation of right hand side of (\ref{ch4:unifbd}) over the sample $U_1, \dots, U_k$ for each choice of the tuple of sets $Q_1,\dots, Q_r$. This bound does not depend on the actual choice of the $Q_i$'s, so if we take the average (over the $Q_i$'s), that upper bound still remains valid.
	
	In order to do this, let us fix $Q_1,\dots, Q_r$, set $Q$ to be the set of elements of $V(G)$ which are contained in at least one of the $Q_i$'s, and fix also the sample points of $U_Q=\{ \, U_i \mid i \in Q \,\}$. Take the expectation $\E_{U_{Q^c}}$ only over the remaining $U_i$ sample points. 
	
	To this end, by Fubini we have the estimate
	\begin{align}
	k^r \E_{U_{[k]}} \|G\|_\square &\leq E_{Q_1, \dots Q_r} \E_{U_Q}  [ \E_{U_{Q^c}} \max_{Q'_i \subset Q_i} \max\{G(P_G(Q'_1)\cap Q^c, \dots, P_G(Q'_r)\cap Q^c); \nonumber \\ & \quad (-G)(P_{-G}(Q'_1)\cap Q^c, \dots, P_{-G}(Q'_r)\cap Q^c)\} ] + \frac{rk^r}{\sqrt{p}} \|G\|_\infty + p r^3 k^{r-1} \|G\|_\infty, \label{ch4:unifbd2}
	\end{align}
	where $U_S=\{ \,U_i \mid i \in S\, \}$.

	Our goal is to uniformly upper bound the expression in the brackets in (\ref{ch4:unifbd2}) so that in the dependence on the particular $Q_1, \dots Q_r$ and the sample points from $U_Q$ vanishes. To achieve this, we consider additionally a tuple of subsets $Q'_i \subset Q_i$, and introduce the random variable $Y(Q'_1, \dots ,Q'_r)= G(P_G(Q'_1)\cap Q^c, \dots, P_G(Q'_r) \cap Q^c)$, where the randomness comes from $U_{Q^c}$ exclusively. Let $$T_i=\{ \,x_i \in [0,1] \mid \sum_{(y_1, \dots, y_{i-1},y_{i+1},\dots y_r) \in Q_i'} W(U_{y_1}, \dots,U_{y_{i-1}},x_i,U_{y_{i+1}},\dots U_{y_r}) >0 \, \}$$ for $i \in [r]$. Note that $t_i \in P_G(Q_i')$ is equivalent to $U_{t_i} \in T_i$. Then 
	\begin{align}
	\E_{U_{Q^c}} Y(Q'_1, \dots ,Q'_r) &\leq  
	\sum_{\substack{t_1, \dots, t_r \in Q^c\\ t_i \neq t_j}} \E_{U_{Q^c}} G(t_1,\dots, t_r)\I_{P_G(Q'_1)}(t_1)\dots\I_{P_G(Q'_r)}(t_r) +r^2 k^{r-1} \|W\|_\infty   \nonumber \\
	& \leq  k^r \int_{T_1 \times \dots \times T_r} W(x) \du \lambda(x) + r^2 k^{r-1} \|W\|_\infty \leq   k^r \|W\|_\square + r^2 k^{r-1} \|W\|_\infty. \nonumber
	\end{align}
	By the Azuma-Hoeffding inequality we also have high concentration of the random variable  $Y(Q'_1, \dots ,Q'_r)$ around its mean, that is
	\begin{equation} \label{ch4:unifbd3}
	\PPP(Y(Q'_1, \dots ,Q'_r) \geq \E_{U_{Q^c}} Y(Q'_1, \dots ,Q'_r) + \rho k^r\|W\|_\infty) < \exp\left(-\frac{\rho^2k}{8r^2 }\right),
	\end{equation}
	since modification of one sampled element changes the value of $Y(Q'_1, \dots ,Q'_r)$ by at most $2rk^{r-1}\|W\|_\infty.$
	Analogous upper bounds on the expectation and the tail probability hold for each of the expressions $(-G)(P_{-G}(Q'_1), \dots, P_{_G}(Q'_r))$.
	
	With regard to the maximum expression in (\ref{ch4:unifbd2}) over the $Q'_i$ sets we have to this end either that the concentration event from (\ref{ch4:unifbd3}) holds for each possible choice of the $Q'_i$ subsets for both expressions in the brackets in (\ref{ch4:unifbd2}),  this has probability at least $1- 2^{pr+1}\exp(-\frac{\rho^2k}{8r^2})$, or it fails for some choice. In the first case we can employ the upper bound  $k^r \|W\|_\square + (r^2 k^{r-1}+ \rho k^r) \|W\|_\infty $, and in the event of failure we still have the trivial upper bound of $k^r\|W\|_\infty$. Eventually we presented an upper bound on the expectation that does not depend on the choice of $Q_1, \dots, Q_r$, and the sample points from $U_Q$. Hence by taking expectation and assembling the terms, we have
	$$
	\E_{U_{[k]}} \|G\|_\square \leq \|W\|_\square + \|W\|_\infty \left(\frac{r}{\sqrt{p}} + \frac{pr^3}{k} +  \rho + \frac{r^2}{k} +2^{pr+1}\exp\left(-\frac{\rho^2k}{8r^2 }\right) \right).
	$$
	Let $p=\sqrt{k}$ and $\rho=\frac{4r^2}{\sqrt[4]{k}}$. Then 
	\begin{align*}
	\E_{U_{[k]}} \|G\|_\square &\leq \|W\|_\square + \|W\|_\infty \left(\frac{r}{\sqrt[4]{k}} + \frac{r^3}{\sqrt{k}} + \frac{4r^2}{\sqrt[4]{k}}  + \frac{r^2}{k} + \exp\left(2\sqrt{k}r-2r^2\sqrt{k}\right) \right)\\
	&\leq  \|W\|_\square + \|W\|_\infty \left(\frac{\varepsilon}{16r} + \frac{\varepsilon^2}{2^8r} + \frac{\varepsilon}{4}  + \frac{\varepsilon^4}{2^{16}r^6} + \frac{\varepsilon^2}{2^8r^6}\right)\leq \|W\|_\square + \varepsilon/2\|W\|_\infty.
	\end{align*}

	The direction concerning the lower bound, $\E \|G\|_\square \geq \|W\|_\square - \varepsilon/2$ follows from a standard sampling argument, the idea is that we can project each set $S \subset [0,1]$ to a set $\hat S \subset [k]$ through the sample, which will fulfill the desired conditions, we leave the details to the reader. Concentration follows by the Azuma-Hoeffding inequality. We conclude that
	\begin{align}
	\PPP\left(\left| \|G\|_\square - \|W\|_\square \right| > \varepsilon\|W\|_\infty \right) &\leq 
	\PPP\left(\left| \E \|G\|_\square -\frac{1}{k^r}\|G\|_\square\right|> \varepsilon/2\|W\|_\infty\right) \nonumber \\ &\leq 2 \exp\left(-\frac{\varepsilon^2k}{32 r^2 }\right). \nonumber
	\end{align}
\end{proof}

%
%
%

Next we state a result on the relationship of a continuous linear program\sindex[notions]{continuous linear program} (LP) and its randomly sampled finite subprogram. We will rely on the next concentration result that is a generalization of the Azuma-Hoeffding inequality\sindex[notions]{Azuma-Hoeffding inequality!inhomogeneous}, \Cref{ch3:azuma}, and suits well the situation when the martingale jump sizes have inhomogeneous distribution. It can be found together with a proof in the survey \cite{Sas} as Corollary 3.

\begin{lemma}[Generalized Azuma-Hoeffding inequality]\label{ch4:azumagen}
	Let $k \geq 1$ and $(X_n)_{n=0}^k$ be a martingale sequence with respect to the natural filtration $(\mathcal F_n)_{n=1}^k$. If $|X_{n}-X_{n+1}|\leq d$ almost surely and $\E[(X_n-X_{n+1})^2\mid \mathcal F_n] \leq \sigma^2$ for each $n \in [k]$, then for every $n \leq k$ and $\delta>0$ it holds that 
	\begin{align}
	\PPP(X_n-X_0 >\delta n) \leq  \exp\left( -n \frac{\sigma^2}{d^2}\left((1+\frac{\delta d}{\sigma^2})\ln(1+\frac{\delta d}{\sigma^2})-\frac{\delta d}{\sigma^2} \right) \right).
	\end{align}
\end{lemma}

Measurability for all of the following functions is assumed.

\begin{lemma} \label{ch4:lpsample3}
	Let $c_m\colon[0,1] \to \R$, $U_{i, m}\colon[0,1] \to \R$ for $i=1, \dots , s$, $m=1, \dots, q$, $u \in \R^{s \times q} $, $\alpha \in \R$. Let $d$ and $\sigma$ be positive reals such that $\|c\|_\infty\leq d$ and $\|c\|_2\leq \sigma$ and set $\gamma=\frac{\sigma^2}{d^2}.$ If the optimum of the linear program
	
	\begin{align}
	&\textnormal{maximize} && \int_0^1 \sum_{m=1}^q f_m(t) c_m(t) \du t\nonumber \\
	&\textnormal{subject to } && \int_0^1 f_m(t)U_{i, m}(t) \du t \leq u_{i, m} && \textrm{for  $i \in [s]$ and $m \in [q]$ } \nonumber \\
	&&& 0 \leq f_m(t) \leq 1 \quad &&\textrm{for $t \in [0,1]$ and $m \in[q]$} \nonumber \\
	&&&\sum_{m=1}^q f_m(t)  =1 \quad &&\textrm{ for $t \in [0,1]$} \nonumber
	\end{align}
	is less than $\alpha$, then for any $\varepsilon,\delta>0$ and $k \in \N$  and a uniform random sample $\{X_1, \dots, X_k\}$ of $[0,1]^k$ the optimum of the sampled linear program
	\begin{align}
	&\textnormal{maximize}&& \sum_{1\leq n \leq k} \sum_{m=1}^q\frac{1}{k} x_{n,m} c_m(X_n) \nonumber \\
	&\textnormal{subject to}&& \sum_{1\leq n \leq k} \frac{1}{k} x_{n,m} U_{i, m}(X_n) \leq u_{i, m} - \delta \|U\|_\infty \quad && \textrm{for  $i \in [s]$ and $m \in [q]$ } \nonumber \\
	&&& 0 \leq x_{n,m} \leq 1 \quad &&\textrm{ for $n \in [k]$ and $m \in [q]$} \nonumber \\
	&&& \sum_{m=1}^q x_{n,m} = 1 && \textrm{ for $n \in [k]$} \nonumber
	\end{align}
	
	is less than $\alpha + \varepsilon$ with probability at least $$1-\left[\exp\left( -\frac{\delta^2k}{2}\right)+\exp \left(-k\gamma\left((1+\frac{\varepsilon}{\gamma d})\ln(1+\frac{\varepsilon}{\gamma d})-\frac{\varepsilon}{\gamma d} \right)\right) \right].$$

\end{lemma}
\begin{proof}
	We require a continuous version of Farkas' Lemma\sindex[notions]{Farkas' Lemma}.
	\begin{claim} \label{ch4:farkas}
		Let $(Af)_{i,m}=\int_0^1 A_{i,m}(t)f_m(t) \du t$ for the bounded measurable functions $A_{i,m}$ on $[0,1]$ for $i \in [s]$ and $m \in [q]$ , and let $v \in \R^{sq}$. There is no fractional $q$-partition solution $f=(f_1, \dots, f_q)$ to $Af \leq v$ if and only if, there exists a non-zero $0 \leq y \in \R^{sq}$ with $\|y\|_1=1$ such that there is no fractional $q$-partition solution $f$ to $y^T(Af) \leq y^Tv$.
	\end{claim}
	For clarity we remark that in the current claim and the following one $Af$ and $v$ are indexed by a pair of parameters, but are regarded as $1$-dimensional vectors in the multiplication operation. 
	\begin{proof}
		One direction is trivial: if there is a solution $f$ to $Af \leq v$, then it is also a solution to $y^T(Af) \leq y^Tv$ for any $y \geq 0$.
		
		We turn to show the opposite direction. Let \begin{align*}C=\{ \,Af\mid  f \textrm{ is a fractional $q$-partition of $[0,1]$} \, \}. \end{align*} The set $C$ is a nonempty convex closed subset of $\R^{sq}$ containing $0$. Let $B=\{ \,x\mid  x_{i,m}\leq v_{i,m}\, \}\subset \R^{sq}$, this set is also a nonempty convex closed set. The absence of a solution to $Af \leq v$ is equivalent to saying that $C \cap B$ is empty. It follows from the Separation Theorem for convex closed sets that there is a $0 \neq y' \in \R^{sq}$ such that $y'^Tc < y'^Tb$ for every $c\in C$ and $b \in B$. Additionally every coordinate $y'_{i,m}$ has to be non-positive. To see this suppose that $y'_{i_0,m_0} >0$, we pick a $c\in C$ and $b\in B$, and send $b_{i_0,m_0}$  to minus infinity leaving every other coordinate of the two points fixed ($b$ will still be an element of $B$), for $b_{i_0}$ small enough the inequality $y'^Tc < y'^Tb$ will be harmed eventually. We conclude that for any $f$ we have $y'^T(Af) < y'^Tv$, hence for $y=\frac{-y'}{\|y'\|_1}$  the inequality $y^T(Af) \leq y^Tv$ has no solution.
	\end{proof}
	From this lemma the finitary version follows without any difficulties.
	\begin{claim} \label{ch4:cons}
		Let  $B$ be a real $sq \times k$ matrix, and let $v \in \R^{sq}$ . There is no fractional $q$-partition $x \in  \R^{kq}$ so that  $Bx \leq v$ if and only if, there is a non-zero $0 \leq y \in \R^{sq}$ with $\|y\|_1=1$ such that there is no fractional $q$-partition $x \in  \R^{kq}$ so that $y^T Bx \leq y^Tv$.
	\end{claim}
	\begin{proof}
		Let $A_{i,m}(t)=\sum_{n=1}^k \frac{B_{(i,m),n}}{k} \I_{[\frac{n-1}{k},\frac{n}{k})}(t)$ for $i=1, \dots, s$. The nonexistence of a fractional $q$-partition $x \in  \R^{kq}$ so that  $Bx \leq v$ is equivalent to nonexistence of a fractional $q$-partition $f$ so that  $Af \leq v$.  For any nonzero $0 \leq y$, the nonexistence of a fractional $q$-partition $x \in  \R^{kq}$ so that $y^T Bx \leq y^Tv$ is equivalent to the nonexistence of a fractional $q$-partition $f$ so that $y^T(Af) \leq y^Tv$. Applying Claim \ref{ch4:farkas} verifies the current claim.
	\end{proof}
	
	The assumption of the lemma is by Claim \ref{ch4:farkas} equivalent to the statement that there exists a nonzero $0 \leq y \in \R^{sq}$ and $0 \leq \beta$ with $\sum_{i=1}^s \sum_{m=1}^q y_{i, m}+\beta=1$  such that
	\begin{align}
	\int_0^1 \sum_{i=1}^s \sum_{m=1}^q y_{i, m} U_{i, m}(t)f_m(t) \du t 
	- \int_0^1 \beta \sum_{m=1}^q c_m(t)f_m(t)  
	\leq \sum_{i=1}^s \sum_{m=1}^q y_{i, m}  u_{i, m} - \beta \alpha \nonumber
	\end{align}
	has no solution $f$ among fractional $q$-partitions.
	This is equivalent to the condition 
	$$
	\int_0^1 h(t) \du t > A,
	$$
	where $h(t) = \min\limits_{m}\left[\sum_{i=1}^s y_{i, m} U_{i, m}(t)-\beta c_m(t)\right]$, and 
	$A=\sum_{i=1}^s \sum_{m=1}^q y_{i, m}  u_{i, m} - \beta \alpha$.
	Let $T_m=\{ \,t \mid h(t) = \sum_{i=1}^s y_{i, m} U_{i, m}(t)-\beta c_m(t)\, \}$ for $m \in [q]$ and define the functions $h_1(t)=\sum_{m=1}^q\I_{T_m}(t) \left[\sum_{i=1}^s y_{i, m} U_{i, m}(t)\right]$ and $h_2(t)=\sum_{m=1}^q\I_{T_m}(t) \beta c_m(t)$. Clearly, $h(t)=h_1(t)-h_2(t)$. Set also $A_1=\sum_{i=1}^s \sum_{m=1}^q y_{i, m}  u_{i, m}$ and $A_2=\beta\alpha.$ Fix an arbitrary $\delta >0$ and $k \geq 1.$ By the Azuma-Hoeffding inequality it follows that with probability at least $1-\exp(-\frac{k\delta^2}{2})$ we have that
	$$
	\frac{1}{k} \sum_{n=1}^k h_1(X_n) > A_1-\delta \|h_1\|_\infty.
	$$
	Note that $\|h_1\|_\infty=\|\sum_{i=1}^s \sum_{m=1}^q\I_{T_m} U_{i, m}y_{i, m}\|_\infty \leq \|U\|_\infty \sum_{i=1}^s \sum_{m=1}^q|y_{i, m}| \leq \|U\|_\infty.$ 
	Moreover, by \Cref{ch4:azumagen} the event
	\begin{align}
	\frac{1}{k} \sum_{n=1}^k h_2(X_n) < A_2 + \varepsilon
	\end{align}
	has probability at least $1-\exp \left(-k\gamma\left((1+\frac{\varepsilon}{\gamma d})\ln(1+\frac{\varepsilon}{\gamma d})-\frac{\varepsilon}{\gamma d} \right)\right).$
	Thus, 
	$$
	\frac{1}{k} \sum_{n=1}^k h(X_n) >  \sum_{i=1}^s \sum_{m=1}^q y_{i, m}  (u_{i, m}-\delta \|U\|_\infty) - \beta (\alpha+\varepsilon)
	$$
	with probability at least $$1-\left[\exp\left( -\frac{\delta^2k}{2}\right)+\exp \left(-k\gamma\left((1+\frac{\varepsilon}{\gamma d})\ln(1+\frac{\varepsilon}{\gamma d})-\frac{\varepsilon}{\gamma d} \right)\right) \right].$$

	We conclude the proof by noting that the last event is equivalent to the event in the statement of our lemma by Claim \ref{ch4:cons}.
\end{proof}

We start the principal part of the proof of the main theorem in this paper.

\begin{proof}[Proof of \Cref{ch4:main}]


	
	It is enough to prove \Cref{ch4:main} for tuples of naive $([-d,d],r)$-digraphons. We first employ \Cref{ch4:weightedlemma} to replace the energy $\hat \EEE(\G(k,W))$ by the energy of the averaged sample $\hat \EEE(\hh(k,W))$ without altering the ground state energy of the sample substantially with high probability. 
	Subsequently, we apply \Cref{ch4:intcont} to change from  the integer version of the energy $\hat \EEE(\hh(k,W))$ to the relaxed one  $\EEE(\hh(k,W))$. That is
	\begin{align*}
	|\hat \EEE(\G(k,W))-\EEE(\hh(k,W))|\leq \varepsilon^2 \|W\|_\infty
	\end{align*}  
	with probability at least $1-\varepsilon^2$.

	We begin with the main argument by showing that the ground state energy of the sample can not be substantially smaller than that of the original, formally \begin{align}\label{ch4:eq001}
	\EEE(\hh(k,W))\geq \EEE(W)-\frac{r^2}{k^{1/4}}\|W\|_\infty
	\end{align} with high probability. In what follows $\E$ denotes the expectation with respect to the uniform independent random sample $(U_S)_{S \in \hhh([k],r)}$ from $[0,1]$. To see the correctness of the inequality, we consider a fixed fractional partition $\phi$ of $[0,1]$, and define the random fractional partition of $[k]$ as $y_{n,m}=\phi_m(U_n)$ for every $n \in [k]$ and $m \in [q]$. Then we have that 
	
	\begin{align}
	\E \EEE(\hh(k,W)) &\geq \E  \EEE_y(\hh(k,W)) \nonumber \\
	&=\E \frac{1}{k^r}\sum_{z \in [q]^r} \sum_{n_1,\dots,n_r=1}^{k} W^z(U_{\hhh(\{n_1, \dots, n_r\},r)}) \prod_{\substack{j=1}}^r y_{n_j,z_j} \nonumber \\
	&\geq \frac{k!}{k^r(k-r)!}\sum_{z \in [q]^r}  \int_{[0,1]^{\hhh([r])}} W^z(t_{\hhh([r])})  \prod_{j=1}^r \phi_{z_j}(t_j) \du \lambda(t) - \frac{r^2}{k}\|W\|_\infty \nonumber \\
	&\geq \EEE_\phi (W) - \frac{r^2}{k}\|W\|_\infty . \nonumber
	\end{align}
	
	This argument proves the claim in expectation, concentration will be provided by standard martingale arguments. For convenience, we define a martingale by $Y_0=\E \EEE(\hh(k,W))$ and $Y_j=E\left[\EEE(\hh(k,W)) \mid  \{ \, U_S \mid S \in \hhh([j],r-1) \, \} \right]$ for $1\leq j\leq k$. The difference $|Y_j-Y_{j+1}| \leq \frac{2r}{k} \|W\|_\infty$ is bounded from above for any $j$, thus by the inequality of Azuma and Hoeffding, \Cref{ch3:azuma}, it follows that 
	\begin{align}
	\PPP & \left( \EEE(\hh(k,W)) < \EEE (W) - \frac{2r^2}{k^{1/4}}\|W\|_\infty \right)  \nonumber \\ 
	& \quad\leq \PPP\left( \EEE(\hh(k,W)) < \E \EEE(\hh(k,W)) - \frac{r^2}{k^{1/4}}\|W\|_\infty \right)  \nonumber \\
	& \quad =\PPP \left( Y_k < Y_0-  \frac{r^2}{k^{1/4}}\|W\|_\infty \right) \leq \exp \left(-\frac{r^2\sqrt{k}}{8} \right). \label{ch4:azuma} 
	\end{align}
	So the lower bound (\ref{ch4:eq001}) on $\EEE(\hh(k,W))$ is established. Note that by the condition regarding $k$ we can establish (rather crudely) the upper bound $\exp(-\frac{r^2\sqrt{k}}{8}) \leq \varepsilon 2^{-7}.$
	
	Now we turn to prove that $\EEE(\hh(k,W))<\EEE(W)+\varepsilon$ holds also with high probability for $k \geq \left( \frac{2^{r+7}q^r r}{\varepsilon}\right)^{4} \log(\frac{2^{r+7}q^r r}{\varepsilon})q^r$. Our two main tools will be \Cref{ch4:cutapprox}, that is a variant the Cut Decomposition Lemma from \cite{AVKK2} (closely related to the Weak Regularity Lemma by Frieze and Kannan \cite{FK}), and linear programming duality, in the form of \Cref{ch4:lpsample3}. Recall the definition of the cut norm, for $W\colon [0,1]^r \to \R$, it is given as
	$$
	\|W\|_\square=\max\limits_{S^1, \dots ,S^r \subset [0,1]} \left|\int_{S^1 \times \dots \times S^r} W(x) \du \lambda(x)\right|, 
	$$
	and for an $r$-array $G$ by the expression
	$$
	\|G\|_\square=\frac{1}{k^r}\max\limits_{S^1, \dots ,S^r \subset V(G)} \left| G(S^1, \dots ,S^r) \right|.
	$$

	Before starting the second part of the technical proof, we present an informal outline. Our task is to certify that there is no assignment of the variables on the sampled energy problem, which produces an overly large value relative to the ground state energy of the continuous problem. For this reason we build up a cover of subsets over the set of fractional partitions of the variables of the finite problem, also build a cover of subsets over the fractional partitions  of the original continuous energy problem, and establish an association scheme between the elements of the two in such a way, that with high probability we can state that the optimum on one particular set of the cover of the sampled energy problem does not exceed the optimal value of the original problem on the associated set of the other cover. To be able to do this, first we have to define these two covers, this is done with the aid of the cut decomposition, see \Cref{ch4:cutapprox}. We will replace the original continuous problem by an auxiliary one, where the number of variables will be bounded uniformly  in terms of our error margin $\varepsilon$. \Cref{ch4:cutpres} makes it possible for us to replace the sampled energy problem by an auxiliary problem with the same complexity as for the continuous problem. This second replacement will have a straightforward relationship to the approximation of the original problem. We will produce the cover sets of the two problems by localizing the auxiliary problems, association happens through the aforementioned straightforward connection. Finally, we will linearize the local problems, and use the linear programming duality principle from \Cref{ch4:lpsample3} to verify that the local optimal value on the sample does not exceed the local optimal value on the original problem by an infeasible amount, with high probability.

	Recall that for a $\phi=(\phi_1, \dots, \phi_q)$ a fractional $q$-partition of $[0,1]$ the energy is given by the formula
	\begin{align} \label{ch4:energydef1}
	\EEE_\phi(W)=\sum_{z \in [q]^r} \int_{[0,1]^r} \prod_{j \in [r]} \phi_{z_j}(t_j)W^z(t) \du \lambda(t),
	\end{align}
	and for an $\mathrm{x}=(x_{1,1},x_{1,2},\dots,x_{1,q}, x_{2,1},\dots,x_{k,q})$ a fractional $q$-partition of $[k]$ by
	\begin{align} \label{ch4:energydef2}
	\EEE_\mathrm{x}(\hh(k,W))=\sum_{z \in [q]^r} \frac{1}{k^r} \sum_{n_1,\dots, n_r=1}^{k} \prod_{j \in [r]} x_{t_j, z_j} W^z(U_{n_1}, \dots, U_{n_r}).
	\end{align}
	
	We are going to establish a term-wise connection with respect to the parameter $z$ in the previous formulas. Therefore we consider the function 
	\begin{equation}\label{ch4:energyterm}
	\EEE^z_\phi(W^z)=\int_{[0,1]^r} \prod_{j \in [r]} \phi_{z_j}(t_j)W^z(t) \du \lambda(t),\end{equation}
	it follows that $\EEE_\phi(W)=\sum_{z \in [q]^r}\EEE^z_\phi(W^z)$. Analogously we consider \begin{align*}\EEE^z_\mathrm{x}(\hh(k,W^z))=\frac{1}{k^r} \sum_{n_1,\dots, n_r=1}^{k} \prod_{j \in [r]} x_{t_j, z_j} W^z(U_{n_1}, \dots, U_{n_r}), \end{align*} so $\EEE_\mathrm{x}(\hh(k,W))=\sum_{z \in [q]^r}\EEE^z_\mathrm{x}(\hh(k,W^z))$ with the sampled graphs on the right generated by the same sample points. Note that the formulas (\ref{ch4:energydef1})-(\ref{ch4:energyterm}) make prefect sense even when the parameters $\phi$ and $\mathrm{x}$ are only vectors of bounded functions and reals respectively without forming partition.
	
	\Cref{ch4:cutapprox} delivers for any $z \in [q]^r$ an integer $s_z \leq \frac{2^6 q^{2r}}{\varepsilon^2}$, measurable sets $S_{z,i,j} \subset [0,1]$ with $i=1, \dots, s_z$, $j=1, \dots, r$, and the real numbers $d_{z,1}, \dots, d_{z,s_z}$ such that the conditions of the lemma are satisfied, namely \begin{align*}\|W^z-\sum_{i=1}^{s_z} d_{z,i}  \I_{S_{z,i,1} \times \dots \times S_{z,i,r}}\|_\square \leq \frac{\varepsilon}{8 q^r} \|W^z\|_2,\end{align*}
	and $\sum_{i=1}^{s_z} |d_{z,i}| \leq \frac{8 q^r}{\varepsilon} \|W^z\|_2$. The cut function allows a sufficiently good approximation for $\EEE_\phi(W^z)$, for any $\phi$. Let $D^z=\sum_{i=1}^{s_z} d_{z,i} \I_{S_{z,i,1} \times \dots \times S_{z,i,r}}$. Then
	
	\begin{align}
	|\EEE^z_\phi(W^z)-\EEE^z_\phi(D^z)|  &=\left|\int_{[0,1]^r} \prod_{j \in [r]} \phi_{z_j}(t_j)\left[W^z(t)-D^z(t)\right] \du \lambda(t) \right|   \nonumber \\ 
	&\leq \|W^z-D^z\|_\square \leq \frac{\varepsilon}{8q^r} \|W^z\|_\infty. \nonumber
	\end{align}
	
	We apply the cut approximation to $W^z$ for every $z \in [q]^r$ to obtain the $[q]^r$-tuple of naive $r$-kernels $D=(D^z)_{z \in [q]^r}$.  We define the "push-forward" of this approximation for the sample $\hh(k,W)$. To do this we need to define the subsets $[k] \supset \hat S_{z,i,j} =\{ \,m \mid U_m \in S_{z,i,j}\, \}$.  Let $\hat D^z=\sum_{i=1}^{s_z} d_{z,i} \I_{\hat S_{z,i,1} \times \dots \times \hat S_{z,i,r}}$. First we condition on the event from \Cref{ch4:cutpres}, call this event $E_1$, that is \begin{align*}E_1=\bigcap\limits_{z \in [q]^r} \left\{\left| \|\hh(k,W^z)-\hat D^z\|_\square - \|W^z-D^z\|_\square \right| <  
	\frac{\varepsilon}{8q^r} \|W\|_\infty \right\}.\end{align*} On $E_1$ it follows that for any $\mathrm{x}$ that is a fractional $q$-partition
	\begin{align}
	|\EEE^z_\mathrm{x}(\hh(k,W^z))-\EEE^z_\mathrm{x}(\hat D^z)| &\leq \|\hh(k,W^z)-\hat D^z\|_\square \nonumber \\ 
	&\leq \|W^z- D^z\|_\square + \frac{\varepsilon}{8q^r} \|W\|_\infty.\nonumber 
	\end{align}
	This implies that
	$$
	|\EEE_\phi(W)-\EEE_\phi(D)| \leq \frac{\varepsilon}{8} \|W\|_\infty \quad \textrm{and} \quad |\EEE_x(\hh(k,W))-\EEE_x(\hh(k,D))| \leq \frac{\varepsilon}{4} \|W\|_\infty. 
	$$
	The probability that $E_1$ fails is at most $2 q^r \exp\left( -\frac{\varepsilon^2k}{2^{11} r^2 q^{2r}}  \right)$ whenever $k \geq \left( \frac{2^{7}q^r r^2}{\varepsilon}\right)^{4}$ due to \Cref{ch4:cutpres}, in the current theorem we have the condition $k\geq \left( \frac{2^{r+7}q^r r}{\varepsilon}\right)^{4} \log(\frac{2^{r+7}q^r r}{\varepsilon})q^r$, which implies the aforementioned one. The failure probability of $E_1$ is then strictly less than $\frac{\varepsilon}{2^7}.$

	Let $\mathcal S=\{ \,S_{z,i,j} \mid z \in [q]^r, 1\leq i\leq s_z, 1\leq j \leq r\, \}$ denote their set, and let $\mathcal S'$ stand for the corresponding set on the sample. Note that $s'=|\mathcal S| \leq 2^6 r q^{3r} \frac{1}{\varepsilon^2}$ in general, but in some cases the $W^z$ functions are constant multiples of each other, so the cut approximation can be chosen in a way that $S_{z,i,j}$ does not depend on $z \in [q]^r$, and in this case we have the slightly refined upper bound $2^6 r q^{2r} \frac{1}{\varepsilon^2}$ for $s'$, consequences of this in the special case are discussed in the remark after the proof. Let $\eta > 0$ be arbitrary, and define the sets

	$$
	I(b, \eta)=\left\{ \,\phi \mid \forall z \in [q]^r, 1\leq i\leq s_z, 1\leq j \leq r \colon \left|\int_{S_{z,i,j}} \phi_{z_j}(t) \du t -b_{z,i,j}\right| \leq 2 \eta\, \right\},
	$$
	and
	$$
	I'(b, \eta)=\left\{ \,\mathrm{x} \mid \forall z \in [q]^r, 1\leq i\leq s_z, 1\leq j \leq r\colon  \left|\frac{1}{k}\sum_{U_n \in S_{z,i,j}} x_{n,z_j} -b_{z,i,j}\right| \leq  \eta\, \right\}
	$$
	For a collection of non-negative reals $\{b_{z,i,j}\}$.
	At this point in the definitions of the above sets we do not require $\phi$ and $\mathrm{x}$ to be fractional $q$-partitions, but to be vectors of bounded functions and vectors respectively.
	We will use the grid points $\mathcal A=\{ \,(b_{z,i,j})_{z,i,j} \mid \forall z,i,j\colon b_{z,i,j} \in [0,1] \cap \eta\Z \, \}$.

	On every nonempty set $I(b,\eta)$ we can produce a linear approximation of $\EEE_\phi(D)$ (linearity is meant in the functions $\phi_m$)  which carries through to a linear approximation of $\EEE_\mathrm{x}(\hh(k,D))$ via sampling. The precise description of this is given in the next auxiliary result.
	
	\begin{lemma}[Local linearization]  \label{ch4:linearize}
		If $\eta \leq \frac{\varepsilon}{16  q^r 2^{r}}$, then for every $b \in \mathcal A$ there exist $l_0 \in \R$ and functions $l_1,l_2, \dots, l_q\colon[0,1] \to \R$ such that for every $\phi \in I(b,\eta)$ it holds that $$\left|\EEE_\phi(D)-l_0-\int_0^1 \sum_{m=1}^q l_m(t)\phi_m(t) \du t\right| < \frac{\varepsilon}{2^{r+3}}\|W\|_\infty,$$ and for every  $\mathrm{x} \in I'(b,\eta)$ we have $$\left|\EEE_\mathrm{x}(\hh(k,D))-l_0-\sum_{n=1}^{k}\sum_{m=1}^q\frac{1}{k} x_{n,m} l_m(U_i)\right| < \frac{\varepsilon}{2^{r+5}}\|W\|_\infty.$$ Additionally we have that $l_1, l_2, \dots, l_q$  are bounded from above by $\frac{8q^{2r}}{\varepsilon}\|W\|_\infty$ and $\int_0^1 \sum_{m=1}^q l^2_m(t)\du t \leq  2^{2r+9} r^2 q^{3r} \|W\|^2_\infty.$
	\end{lemma}
	
	\begin{proof}
		Recall the decomposition of the energies as sums over $z \in [q]^r$ into terms
		\begin{align}\EEE^z_\phi(D^z)&=\sum_{i=1}^{s_z} d_{z,i} \int_{[0,1]^r} \prod_{j=1}^r \phi_{z_j}(t_j) \I_{S_{z,i,1} \times \dots \times S_{z,i,r}}(t) \du t \nonumber \\  &=\sum_{i=1}^{s_z} d_{z,i} \int_{[0,1]^r} \prod_{m=1}^q \prod_{\substack{j=1 \\ z_j=m}}^r \phi_m(t_j) \I_{S_{z,i,1} \times \dots \times S_{z,i,r}}(t) \du t,\nonumber
		\end{align}
		and
		$$
		\EEE^z_\mathrm{x}(\hat D^z)=\sum_{i=1}^{s_z} d_{z,i} \frac{1}{k^r} \prod_{m=1}^q \prod_{\substack{j=1 \\ z_j=m}}^r \sum_{n\colon U_n \in S_{z,i,j}} x_{n,m}.
		$$
		We linearize and compare the functions $\EEE^z_\phi(D^z)$ and $\EEE^z_{\mathrm{x}}(\hat D^z)$ term-wise. In the end we will sum up the errors and deviations occurred at each term.
		Let $b\in \mathcal A$ and $\eta>0$ as in the statement of the lemma with $I(b,\eta)$ being nonempty. Let us fix an arbitrary $\phi \in  I(b,\eta)$, $z \in [q]^r$, and $1\leq i \leq s_z$. Then
		\begin{align*}
		\prod_{\substack{j=1}}^r \left[\int_0^1\phi_{z_j}(t_j) \I_{S_{z,i,j}}(t_j)\du t_j\right] 
		&=B^i(z) + \sum_{j=1}^r \left[\int_0^1\phi_{z_j}(t_j) \I_{S_{z,i,j}}(t_j)\du t_j-b_{z,i,j}\right]B^{i,j}(z)+  \Delta  \\
		&= (1-r)B^i(z)  +  \sum_{m=1}^q \int_0^1\phi_m(t) \left[\sum_{j=1, z_j=m}^r\I_{S_{z,i,j}}(t)B^{i,j}(z)\right] \du t  + \Delta ,
		\end{align*}
		where $B^i(z)$ stands for $\prod_{j=1}^{r} b_{z,i,j}$, $B^{i,j}(z)=\prod_{l\neq j} b_{z,i,l}$,  and $|\Delta| \leq 4\eta^2 2^r$.
		Analogously for an arbitrary fixed element $\mathrm{x} \in  I'(b,\eta)$ and a term of $\EEE^z_\mathrm{x}(\hat D^z)$ we have 
		\begin{align}
		& \prod_{\substack{j=1}}^r \left[\frac{1}{k}\sum_{n\colon U_n \in S_{z,i,j}} x_{n,z_j}-b_{z,i,j}+b_{z,i,j}\right] \nonumber \\
		& \qquad =(1-r)B^i(z)  + \sum_{m=1}^{q}\sum_{n=1}^k\frac{1}{k} x_{n,m} \left[ \sum_{j=1, z_j=m}^r \I_{S_{z,i,j}}(U_n) B^{i,j}(z) \right]  + \Delta', \nonumber
		\end{align}
		where $|\Delta'| \leq \eta^2 2^r$.

		If we multiply these former expressions by the respective coefficient $d_{z,i}$ and sum up over $i$ and $z$, then we obtain the final linear approximation consisting of the constant $l_0$ and the functions $l_1, \dots, l_q.$ We would like to add that these objects do not depend on $\eta$ if $I(b,\eta)$ is nonempty, only the accuracy of the approximation does. As overall error in approximating the energies we get in the first case of $\EEE_\phi(D)$ at most $32 \eta^2 2^r \frac{q^{2r}}{\varepsilon} \|W\|_\infty \leq \frac{\varepsilon}{2^{r+3}} \|W\|_\infty$, and in the second case of $\EEE_{\mathrm{x}}(\hh(k,D))$ at most $\frac{\varepsilon}{2^{r+5}} \|W\|_\infty$.
		
		Now we turn to prove the upper bound on $|l_m(t)|$. Looking at the above formulas we could write out $l_m(t)$ explicitly, for our upper bound it is enough to note that $$\sum_{j=1, z_j=m}^r\I_{S_{z,i,j}}(t)B^{i,j}(z)$$ is at most $r$. So it follows that for any $t \in [0,1]$ it holds that
		$$|l_m(t)|\leq \frac{8q^{2r}}{\varepsilon} r\|W\|_\infty.$$
		
		It remains to verify the assertion regarding $\int_{0}^1 \sum_{m=1}^q l^2_m(t) \du t.$ 
		Note that $I(b,\eta) \subset I(b, 2\eta)$, so we can apply the same linear approximation to elements $\psi$ of $I(b, 2\eta)$ as above with a deviation of at most $\frac{\varepsilon}{2^{r+1}} \|W\|_\infty$ from $\EEE_\psi(D)$. Let $\phi$ be an arbitrary element of $I(b,\eta)$, and let $T \subset [0,1]$ denote the set of measure $\eta$ corresponding to the largest $\sum_{m=1}^q |l_m(t)|$ values. Define 
		\[
		\hat \phi_m(t) =
		\begin{cases} 
		\hfill \phi_m(t)+ \mathrm{sgn}(l_m(t))    \hfill & \text{ if $t \in T$} \\
		\hfill \phi_m(t) \hfill & \text{otherwise.} \\
		\end{cases}
		\]
		Then $\hat \phi \in I(b,2\eta)$, since $\|\phi_m-\hat \phi_m\|_1\leq \eta$ for each $m \in [q]$, but $\hat \phi$ is not necessarily a fractional partition. Therefore we have 
		\begin{align*}
		\int_T \sum_{m=1}^q |l_m(t)| \du t &= \int_0^1 \sum_{m=1}^q (\hat \phi_m(t)-\phi_m(t)) l_m(t) \du t \\ &\leq \left|\int_0^1 \sum_{m=1}^q \hat \phi_m(t)l_m(t) \du t - \EEE_{\hat \phi} (D) \right|+ |\EEE_{\hat \phi} (D)-\EEE_{\phi} (D)| \\ & \qquad \qquad + \left|\int_0^1 \sum_{m=1}^q  \phi_m(t)l_m(t) \du t - \EEE_{\phi} (D)\right| \\
		& \leq \frac{5}{2^{r+3}} \varepsilon \|W\|_\infty  + |\EEE_{\hat \phi} (D)-\EEE_{\phi} (D)|.
		\end{align*} 
		We have to estimate the last term of the above expression.
		\begin{align*}
		|\EEE_{\hat \phi} (D)-\EEE_{\phi} (D)| &\leq \sum_{z \in [q]^r} \left| \int_{[0,1]^r} \left( \prod_{j=1}^r \phi_{z_j}(t_j) - \prod_{j=1}^r \hat\phi_{z_j}(t_j)\right) D^z(t) \du t \right| \\
		&\leq 2\|W\|_\infty \sum_{z \in [q]^r}  \int_{[0,1]^r} \sum_{j=1}^r \left|\prod_{i <j}^r\phi_{z_i}(t_i) \prod_{i >j}^r\hat\phi_{z_i}(t_i) (\hat \phi_{z_j}(t_j)-\phi_{z_j}(t_j))\right| \du t \\
		& \leq  2\|W\|_\infty 2^r q^{r-1} r \sum_{m=1}^q \|\phi_m-\hat \phi_m\|_1 \leq  2\|W\|_\infty 2^r q^{r} r \eta. 
		\end{align*}
		We conclude that 
		\begin{align*}
		\int_T \sum_{m=1}^q |l_m(t)| \du t \leq \left(\frac{5}{2^{r+3}} + \frac{r}{2^{3}} \right)  \varepsilon \|W\|_\infty. 
		\end{align*} 
		This further implies that for each $t \notin T$ we have $\sum_{m=1}^q |l_m(t)| \leq \left(\frac{5}{2^{r+3}} + \frac{r}{2^{3}} \right)  \frac{\varepsilon}{\eta} \|W\|_\infty \leq \left(10 + 2^{r+1}r \right)q^r  \|W\|_\infty.$ These former bounds indicate
		\begin{align*}
		\int_0^1 \sum_{m=1}^q l^2_m(t) \du t &=  \int_{[0,1] \setminus T} \sum_{m=1}^q l^2_m(t) \du t + \int_T \sum_{m=1}^q l^2_m(t) \du t \\ &\leq  2^{2r+8} r^2 q^{2r} \|W\|^2_\infty  + \|l\|_\infty \int_T \sum_{m=1}^q |l_m(t)| \du t  \\ &\leq  2^{2r+8} r^2 q^{2r} \|W\|^2_\infty  + (2^{r+4} r q^{r})  (8 q^{2r} r) \|W\|^2_\infty\\
		&\leq 2^{2r+9} r^2 q^{3r} \|W\|^2_\infty. 
		\end{align*} 

	\end{proof}

	We return to the proof of the main theorem, and set $\eta=\frac{\varepsilon}{16  q^r 2^{r}}.$ For each $b \in \mathcal A$ we apply \Cref{ch4:linearize}, so that we have for any $\phi \in I(b,\eta)$ and $\mathrm{x} \in I'(b,\eta)$ that
	\begin{align}
	\left|\EEE_\phi(W)-l_0-\sum_{m=1}^q \int_0^1 \phi_m(t) l_m(t) \du t \right| &= \frac{\varepsilon}{2^{r+3}}\|W\|_\infty, \nonumber \\
	\left|\EEE_\mathrm{x}(\hh(k,W))-l_0-\sum_{n=1}^{k}\frac{1}{k} x_{n,m} l_m(U_n) \right| &= \frac{\varepsilon}{2^{r+5}}\|W\|_\infty, \nonumber
	\end{align}
	since $\eta$ is small enough.
	Note that $l_0,l_1, \dots,$ and $l_q$ inherently depend on $b$.
	We introduce the event $E_2(b)$, which stands for the occurrence of the following implication:
	
	\vskip 1em
	If the linear program
	\begin{align}
	&\textnormal{maximize} && l_0+\sum_{n=1}^{k}\sum_{m=1}^q\frac{1}{k} x_{n,m} l_m(U_n)  \nonumber \\
	&\textnormal{subject to } && \mathrm{x} \in  I'(b, \eta) \nonumber \\
	&&& 0\leq x_{n,m} \leq 1 \quad &&\textrm{for $n=1, \dots,k$ and $m=1, \dots, q$} \nonumber \\
	&&&\sum_{m=1}^q x_{n,m}=1 &&\textrm{ for $m=1, \dots,q$} \nonumber
	\end{align}
	has optimal value $\alpha$, then the continuous linear program
	\begin{align}
	&\textnormal{maximize}&& l_0+\int_0^1 \sum_{m=1}^q l_m(t)\phi_m(t) \du t \nonumber \\
	&\textnormal{subject to}&& \phi \in I(b,\eta) \nonumber \\
	&&& 0\leq \phi_m(t) \leq 1 \quad &&\textrm{ for $t \in [0,1]$ and $m=1, \dots, q$} \nonumber \\
	&&&\sum_{m=1}^q \phi_m(t)=1 \quad && \textrm{ for $t \in [0,1]$} \nonumber
	\end{align}
	has optimal value at least $\alpha- (\varepsilon/2)\|W\|_\infty$. 
	\vskip 1em
	%
	
	We apply \Cref{ch4:lpsample3} with $\delta=\eta$, $\sigma^2= 2^{2r+9} r^2 q^{3r} \|W\|^2_\infty$, $d=\frac{8q^{2r}}{\varepsilon} r\|W\|_\infty$, and $\gamma=\frac{\sigma^2}{d^2}$, and attain that the probability that $E_2(b)$ fails is at most 
	\begin{align*}
	\exp &\left(-\frac{k\eta^2}{2} \right) + \exp \left(-k\gamma\left((1+\frac{\varepsilon\|W\|_\infty}{\gamma d})\ln(1+\frac{\varepsilon\|W\|_\infty}{\gamma d})-\frac{\varepsilon\|W\|_\infty}{\gamma d} \right)\right)\\
	&\leq \exp \left(-\frac{k\varepsilon^2}{2^8 q^{2r} 2^{2r}}\right)+  \exp \left(-k\varepsilon^2 2^{2r+3}q^{-r}\left( \frac{1}{2^{4r+15}q^{r}r^2}  \right)\right)\\
	&=\exp \left(-\frac{k\varepsilon^2}{2^{2r+8} q^{2r} }\right)+  \exp \left(-\frac{k\varepsilon^2}{2^{2r+12}q^{2r}r^2}  \right) \leq 2\exp \left(-\frac{k\varepsilon^2}{2^{2r+12}q^{2r}r^2}  \right) ,
	\end{align*}
	where we used that $(1+x)\ln(1+x) -x \geq (1+x)(x-x^2/2)-x= x^2/2-x^3/2 \geq x^2/4$ for $0 \leq x\leq \frac{1}{4}.$
	Denote by $E_2$ the event that for each $b \in \mathcal A$ the event $E_2(b)$ occurs. Then we have \begin{align*}\PPP(E_2)&\geq 1- 2 \left(\frac{2^{r+3} q^{r}}{\varepsilon}\right)^{2^6 r q^{3r} \frac{1}{\varepsilon^2}}\exp \left(-\frac{k\varepsilon^2}{2^{2r+12}q^{2r}r^2}  \right)\\
	&\geq 1- 2\exp\left( \log\left(\frac{2^{r+3} q^{r} }{\varepsilon}\right) 2^{6} r q^{3r} \varepsilon^{-2} -  \log(\frac{2^{r+7}q^r r}{\varepsilon}) 2^{2r+16} r^{2} q^{3r}  \varepsilon^{-2}  \right) \\
	&\geq 1- 2\exp\left( -  \log(\frac{2^{r+7}q^r r}{\varepsilon}) 2^{2r+15} r^{2} q^{3r}  \varepsilon^{-2}  \right) \\
	&\geq 1-\varepsilon/4.
	\end{align*}
	Therefore for $k\geq \left( \frac{2^{r+7}q^r r}{\varepsilon}\right)^{4} \log(\frac{2^{r+7}q^r r}{\varepsilon})q^r$ we have that $\PPP(E_1 \cap E_2)\geq 1- \varepsilon/2$. We only need to check that conditioned on $E_1$ and $E_2$ our requirements are fulfilled.
	For this, consider an arbitrary fractional $q$-partition of $[k]$ denoted by $\mathrm{x}$. For some $b \in \mathcal A$ we have that $\mathrm{x} \in I'(b,\eta)$. If we sum up the error gaps that were allowed for the Cut Decomposition and at the local linearization stage, then the argument we presented above yields that there exists a $\phi \in I(b, \eta)$ such that conditioned on the event $E_1 \cap E_2$ it holds 
	$$
	\EEE_\phi(W) \geq \EEE_\mathrm{x}(\hh(k,W))-\varepsilon\|W\|_\infty.
	$$
	This is what we wanted to show.

\end{proof}

We can improve on the tail probability bound in \Cref{ch4:main} significantly by a constant factor strengthening of the lower threshold condition imposed on the sample size.   
\begin{corollary}\label{ch4:maincor1}
	Let $r\geq 1$, $q \geq 1$, and $\varepsilon >0$. Then for any $[q]^r$-tuple of  $([-d,d],r)$-graphons $W=(W^z)_{z \in [q]^r}$ and $k \geq \Theta^4  \log(\Theta) q^r$ with $\Theta=\frac{2^{r+10}q^r r}{\varepsilon}$  we have that
	\begin{align}\label{ch4:eq2}
	\PPP(|\EEE(W)-\hat \EEE(\G (k,W))|>\varepsilon \|W\|_\infty)< 2 \exp\left( -\frac{\varepsilon^2k}{8r^2}\right).
	\end{align}
\end{corollary}
\begin{proof}
	For $k \geq \Theta^4  \log(\Theta) q^r$ we appeal to \Cref{ch4:main}, hence
	\begin{align*}
	|\EEE(W)-\E \hat \EEE(\G (k,W))| &\leq \PPP(|\EEE(W)-\hat \EEE(\G (k,W))|>\varepsilon/8 \|W\|_\infty) 2 \|W\|_\infty+ \varepsilon/8 \|W\|_\infty \\ &< \varepsilon/2 \|W\|_\infty.
	\end{align*}
	Using a similar martingale construction to the one in the first part of the proof of \Cref{ch4:main} the Azuma-Hoeffding inequality can be applied, thus
	\begin{align*}
	\PPP(|\EEE(W)-\hat \EEE(\G (k,W))|>\varepsilon \|W\|_\infty) &\leq \PPP(|\E \hat \EEE(\G (k,W))-\hat \EEE(\G (k,W))|>\varepsilon/2 \|W\|_\infty) \\ 
	&\leq 2 \exp\left( -\frac{\varepsilon^2k}{8r^2}\right).
	\end{align*}
	
\end{proof}

\begin{remark} \label{ch4:mainrem}
	A simple investigation of the above proof also exposes that in the case when the $W^z$'s are constant multiples of each other then we can employ the same cut decomposition to all of them with the right scaling, which implies that the upper bound on $|\mathcal S|$ can be strengthened to $2^6 r q^{2r} \frac{1}{\varepsilon^2}$, gaining a factor of $q^r$. Therefore in this case the statement of  \Cref{ch4:maincor1} is valid with the improved lower bound condition $\left( \frac{2^{r+10}q^r r}{\varepsilon}\right)^{4} \log(\frac{2^{r+10}q^r r}{\varepsilon})$ on $k$.
\end{remark}

\begin{remark} \label{ch4:rem2}
	Suppose that $f$ is the following simple graph parameter. Let $q \geq 1$, $m_0 \geq 1$, and $g$ be a polynomial of $l$ variables and degree $d$ with values between $0$ and $1$ on the unit cube, where $l$ is the number of unlabeled node-$q$-colored graphs on $m_0$ vertices, whose set we denote by $\mathcal M_{q,m_0}$. Note that $ l \leq 2^{m_0^2/2}q^{m_0}$. Let then
	\begin{align}\label{ch4:polymax}
	f(G)=\max_{\mathcal T} g( (t(F,(G,\mathcal T)))_{F \in \mathcal M_{q,m_0}} ), 
	\end{align} 
	where the maximum goes over all node-$q$-colorings of $G$, and $(G,\mathcal T)$ denotes the node-$q$-colored graph by imposing $\mathcal T$ on the node set of $G$. Using the identity $t(F_1,G)t(F_2,G)=t(F_1 \cup F_2,G)$, where $F_1 \cup F_2$ is the disjoint union of the (perhaps colored) graphs $F_1$ and $F_2$, we can replace in (\ref{ch4:polymax}) $g$ by $g'$ that is linear, and its variables are indexed by $\mathcal M_{q,d m_0}.$ Then it becomes clear that $f$ can be regarded as a ground state energy of $d m_0$-dimensional arrays by associating to every $G$ an tuple $(A^z)_{z \in [q]^r}$  with $r=dm_0$, where the entries $A^z(i_1, \dots, i_r)$ are the coefficients of $g'$ corresponding to the element of  $\mathcal M_{q,d m_0}$ given by the pair $z$ and $G|_{(i_1, \dots, i_r)}$. We conclude that $f$ is efficiently testable by \Cref{ch4:main}.
\end{remark}

\section{Testability of variants of the ground state energy}\label{sec:appl}

In the current section we derive further testability results using the techniques employed in the proofs of the previous section, and apply \Cref{ch4:main} to some specific quadratic programming problems.

\subsection{Microcanonical version}
Next we will state and prove the microcanonical version of \Cref{ch4:main}, that is the continuous generalization of the main result of \cite{VKK} for an arbitrary number $q$ of the states. To be able to do this, we require the microcanonical analog of \Cref{ch4:intcont}, that will be a generalization of Theorem 5.5 from \cite{BCL2} for arbitrary $r$-graphs (except for the fact that we are not dealing with node weights), and its proof will also follow the lines of the aforementioned theorem. Before stating the lemma, we outline some notation and state yet another auxiliary lemma.

\begin{definition}\label{ch4:microdef}
	Let for $\na=(a_1, \dots, a_q) \in \Pd_q$\sindex[symbols]{p@$\Pd_q$} (that is, $a_i \geq 0$ for each $i \in [q]$ and $\sum_i a_i=1$) denote
	$$
	\Omega_\na=\left\{ \,\phi \textrm{ fractional $q$-partition of $[0,1]$} \mid \int_0^1 \phi_i(t)\du t=a_i \textrm{ for } i \in [q] \, \right\},
	$$\sindex[symbols]{o@$\Omega_\na$,$\omega_\na$,$\hat \omega_\na$}
	$$
	\omega_\na=\left\{ \, \mathrm{x} \textrm{ fractional $q$-partition of $V(G)$} \mid \frac{1}{|V(G)|}\sum_{u \in V(G)} x_{u,i}=a_i \textrm{ for  $i \in [q]$}  \, \right\},
	$$
	and
	$$
	\hat \omega_\na=\left\{ \,\mathrm{x} \textrm{ integer $q$-partition of $V(G)$} \mid \left|\frac{\sum_{u \in V(G)} x_{u,i}}{|V(G)|}-a_i \right|\leq \frac{1}{|V(G)|} \textrm{ for  $i \in [q]$} \, \right\}.
	$$
	The elements of the above  sets are referred to as integer $\na$-partitions and fractional $\na$-partitions, respectively.
	
	We call the following expressions microcanonical ground state energies\sindex[notions]{ground state energy!microcanonical} with respect to $\na$ for $(\KK,r)$-graphs and graphons and $C(\KK)$-valued $r$-arrays $J$, in the finite case we add the term fractional and integer respectively to the name. Denote
	$$ \EEE_\na(W,J)=\max_{\phi \in \Omega_\na} \EEE_\phi(W,J), \quad \EEE_\na(G,J)=\max_{\mathrm{x} \in \omega_\na} \EEE_\mathrm{x}(G,J),  \quad \hat \EEE_\na(G,J)=\max_{\mathrm{x} \in \hat \omega_\na} \EEE_\mathrm{x}(G,J).$$\sindex[symbols]{e@$\EEE_\na(W,J), \EEE_\na(G,J), \hat \EEE_\na(G,J)$}
	
	The layered versions for a finite layer set $\eE$, and the canonical versions $\EEE_\na(W)$, $\EEE_\na(G)$, and $\hat \EEE_\na(G)$\sindex[symbols]{e@$\EEE_\na(W)$, $\EEE_\na(G)$, $\hat \EEE_\na(G)$} are defined analogously.
	
\end{definition}

The requirements for an $\mathrm{x}$ to be an integer fractional $\na$-partition (that is $\phi \in \Omega_\na$)\sindex[notions]{fractional $\na$-partition} are rather strict and we are not able to guarantee with high probability that if we sample from an fractional $\na$-partition of $[0,1]$, that we will receive an fractional $\na$-partition on the sample, in fact this will not happen with probability $1$. To tackle this problem we need to establish an upper bound on the difference of two microcanonical ground state energies with the same parameters. This was done in the two dimensional case in \cite{BCL2}, we slightly generalize that approach.

\begin{lemma} \label{ch4:encont}
	Let $r\geq 1$, and $q \geq 1$. Then for any $[q]^r$-tuple of naive $r$-kernels $W=(W^z)_{z \in [q]^r}$, and probability distributions $\na, \nb \in \Pd_q$  we have
	$$
	|\EEE_\na (W) -\EEE_\nb (W)| \leq r  \|W\|_\infty \|\na-\nb\|_1.
	$$
	The analogous statement is true for a $[q]^r$-tuple of $([-d,d],r)$-digraphs $G=(G^z)_{z \in [q]^r}$,
	$$
	|\EEE_\na (G) -\EEE_\nb (G)| \leq  r  \|G\|_\infty \|\na-\nb\|_1.
	$$
\end{lemma}
\begin{proof}
	We will find for each fractional $\na$-partition $\phi$ a fractional $\nb$-partition $\phi'$ and vice versa, so that the corresponding energies are as close to each other as in the statement. So let $\phi=(\phi_1,\dots,\phi_q)$ be an arbitrary fractional $\na$-partition, we define $\phi'_i$ so that the following holds: if $a_i \geq b_i$ then $\phi'_i(t) \leq \phi_i(t)$ for every $t \in [0,1]$, otherwise $\phi'_i(t) \geq \phi_i(t)$ for every $t \in [0,1]$. It is easy to see that such a $\phi'=(\phi'_1, \dots,\phi'_q)$ exists. Next we estimate the energy deviation. 
	\begin{align}
	|\EEE_\phi(W)-\EEE_{\phi'}(W)| &\leq \sum_{z \in [q]^r} \left|\int_{[0,1]^r} \phi_{z_1}(x_1)\dots\phi_{z_r}(x_r)-\phi'_{z_1}(x_1)\dots\phi'_{z_r}(x_r)\du \lambda(x) \right| \|W\|_\infty \nonumber \\
	&\leq \sum_{z \in [q]^r} \sum_{m=1}^r \left|\int_{[0,1]^r} (\phi_{z_m}(x_m)-\phi'_{z_m}(x_m))\prod_{j< m}\phi_{z_j}(x_j) \prod_{j>m}\phi'_{z_j}(x_j)\du \lambda(x)\right| \|W\|_\infty \nonumber \\
	&=\sum_{z \in [q]^r} \sum_{m=1}^r \int_{[0,1]} \left| \phi_{z_m}(x_m)-\phi'_{z_m}(x_m) \right|\du x_m\prod_{j< m}a_{z_j} \prod_{j>m}b_{z_j}\|W\|_\infty \nonumber \\
	&= \sum_{m=1}^r \sum_{j=1}^q\int_{[0,1]} \left| \phi_{j}(t)-\phi'_{j}(t) \right| \du t \left(\sum_{j=1}^q a_{j}\right)^{m-1} \left(\sum_{j=1}^q b_{j}\right)^{r-m-1} \|W\|_\infty\nonumber \\
	&= r \|\na-\nb\|_1 \|W\|_\infty. \nonumber
	\end{align}
	The same way we can find for any fractional $\nb$-partition $\phi$ an fractional $\na$-partition $\phi'$ so that their respective energies differ at most by $r \|\na-\nb\|_1\|W\|_\infty$.
	This implies the first statement of the lemma. The finite case is proven in a completely analogous fashion.
	
\end{proof}

We are ready to show that the difference of the fractional and the integer ground state energies is $o(|V(G)|)$ whenever all parameters are fixed, this result is a generalization with respect to the dimension in the non-weighted case of Theorem 5.5 of \cite{BCL2}, the proof proceeds similar to the one concerning the graph case that was dealt with in \cite{BCL2}.

\begin{lemma} \label{ch4:intcont2}
	Let $q,r,k \geq 1$,$\na  \in \Pd_q$, and $G=(G^z)_{z \in [q]^r}$ be a tuple of $([-d,d],r)$-graphs  on $[k]$. Then
	$$
	|\EEE_\na (G) - \hat \EEE_\na (G)| \leq  \frac{1}{k}  \|G\|_\infty 5^r q^{r+1}.
	$$
\end{lemma}
\begin{proof}
	The inequality $\EEE_\na (G) \leq \hat \EEE_\na (G) + \frac{1}{k} \|G\|_\infty 5^r q^{r+1}$ follows from \Cref{ch4:encont}. Indeed, for this bound a somewhat stronger statement it possible,
	$$\hat \EEE_\na (G) \leq \max_{\nb \colon |b_i-a_i|\leq 1/k} \EEE_\nb (G) \leq  \EEE_\na (G) + r\frac{q}{k} \|G\|_\infty .
	$$
	Now we will show that $\hat \EEE_\na (G) \geq  \EEE_\na (G) - \frac{1}{k}\|G\|_\infty 5^r q^{r+1}$. We consider an arbitrary fractional $\na$-partition $\mathrm{x}$. A node $i$ from $[n]$ is called bad in a fractional partition $x$, if at least two elements of $\{x_{i,1},\dots, x_{i,q}\}$ are positive. We will reduce the number of fractional entries of the bad nodes of $\mathrm{x}$ step by step until we have at most $q$ of them, and keep track of the cost of each conversion, at the end we round the corresponding fractional entries of the remaining bad nodes in some certain way. 
	
	We will describe a step of the reduction of fractional entries. For now assume that we have at least $q+1$ bad nodes and select an arbitrary set $S$ of cardinality $q+1$ of them. To each element of $S$ corresponds a $q$-tuple of entries and each of these $q$-tuples has at least two non-$\{0,1\}$ elements. 
	
	We reduce the number of fractional entries corresponding to $S$ while not disrupting any entries corresponding to nodes that lie outside of $S$. To do this we fix for each $i \in [q]$ the sums $\sum_{v \in S} x_{v,i}$ and for each $v \in S$ the sums $\sum_{i=1}^q x_{v,i}$ (these latter are naturally fixed to be $1$), in total $2q+1$ linear equalities. We have at least $2q+2$ fractional entries corresponding to $S$, therefore there exists a subspace of solutions of dimension at least $1$ for the $2q+1$ linear equalities. That is, there is a family of fractional partitions parametrized by $-t_1 \leq t \leq t_2$ for some $t_1,t_2>0$ that obey our $2q+1$ fixed equalities and have the following form. Let $x^t_{i,j}=x_{i,j}+t\beta_{i,j}$, where $\beta_{i,j}=0$ if $i \notin S$ or $x_{i,j} \in\{0,1\}$, and $\beta_{i,j}\neq 0$ else, together these entries define $\mathrm{x}^t$. The boundaries $-t_1$ and $t_2$ are non-zero and finite, because eventually an entry corresponding to $S$ would exceed $1$ or would be less than $0$ with $t$ going to plus, respectively minus infinity. Therefore at these boundary points we still have an fractional $\na$-partition that satisfies our selected equalities, but the number of fractional entries decreases by at least one. We will formalize how the energy behaves when applying this procedure.
	$$
	\EEE_{\mathrm{x}^t}(G)=\EEE_\mathrm{x}(G)+c_1t+ \dots+c_r t^r,
	$$
	where for $l \in [r]$ we have 
	\begin{align}
	c_l=\frac{1}{k^r}\sum_{z \in [q]^r} \sum_{\substack{u_1, \dots, u_l\in S \\ u_{l+1}, \dots, u_r \in  V \setminus S\\ \pi }} \beta_{u_1,z_{\pi(1)}} \dots \beta_{u_l,z_{\pi(l)}} x_{u_{l+1},z_{\pi(l+1)}} \dots x_{u_{r},z_{\pi(r)}} G^z(u_{\pi(1)}, \dots,u_{\pi(r)}), \nonumber
	\end{align}
	where the second sum runs over permutations $\pi$ of $[k]$ that preserves the ordering of the elements of $\{1,\dots,l\}$ and $\{l+1,\dots, r\}$ at the same time.
	We deform the entries corresponding to $S$ through $t$ in the direction so that $c_1t \geq 0$ until we have eliminated at least one fractional entry, that is we set $t=-t_1$, if $c<0$, and $t=t_2$ otherwise. Note, that as $\mathrm{x}^t$ is a fractional partition, therefore $0\leq x_{i,j}+t\beta_{i,j}\leq 1$, which implies that for $t\beta_{i,j}\leq 0$ we have $|t\beta_{i,j}| \leq x_{i,j}$. On the other hand, $\sum_{j} t\beta_{i,j}=0$ for any $t$ and $i$. Therefore $\sum_{j}|t\beta_{i,j}|=2\sum_{j}|t\beta_{i,j}| \I_{\{t\beta_{i,j}\leq 0\}}\leq 2 \sum_{j}x_{i,j}=2$ for any $i \in [k]$. This simple fact enables us to upper bound the absolute value of the terms $c_lt^l$.
	\begin{align}
	|c_l t^l| &\leq \frac{(k-q-1)^{r-l}}{k^r} \|G\|_\infty \sum_{z \in [q]^r}\sum_{\substack{u_1, \dots, u_l\in S \\ \pi }} |t\beta_{u_1,z_{\pi(1)}}| \dots |t\beta_{u_l,z_{\pi(l)}}| \nonumber \\
	&=\frac{(k-q-1)^{r-l}}{k^r} \|G\|_\infty {r \choose l} q^{r-l}\sum_{z \in [q]^l}\sum_{\substack{u_1, \dots, u_l\in S}} |t\beta_{u_1,z_1}| \dots |t\beta_{u_l,z_l}| \nonumber \\
	&\leq \frac{1}{k^l} \|G\|_\infty {r \choose l} q^{r-l} \left(\sum_{u \in S, j \in [q]} |t\beta_{u,j}|\right)^l\leq \frac{1}{k^l} \|G\|_\infty {r \choose l} q^{r-l} (2q+2)^l. \nonumber
	\end{align}
	It follows that in each step  of  elimination of a fractional entry of $\mathrm{x}$ we have to admit a decrease of the energy value of at most
	$$
	\sum_{l=2}^r |c_l t^l| \leq \frac{1}{k^2} \|G\|_\infty (3q+2)^r.
	$$
	There are in total $kq$ entries in $\mathrm{x}$, therefore, since in each step the number of fractional entries is reduced by at least $1$, we can upper bound the number of required steps for reducing the cardinality of bad nodes to at most $q$  by $k(q-1)$, and conclude that we admit an overall energy decrease of at most  $\frac{1}{k} \|G\|_\infty (q-1)(3q+2)^r$ to construct from $\mathrm{x}$ a fractional partition $\mathrm{x}'$ with at most $q$ nodes with fractional entries 
	In the second stage we proceed as follows. Let $B=\{u_1, \dots, u_m\}$ be the set of the remaining bad nodes of $\mathrm{x}'$, with $m\leq q$. For $u_i \in B$ we set $x''_{u_i, j}=\I_i(j)$, for the rest of the nodes we set $\mathrm{x}''=\mathrm{x}'$, obtaining an integer $\na$-partition of $[k]$. Finally, we estimate the cost of this operation. We get that
	$$
	\EEE_{\mathrm{x}''}(G)\geq \EEE_{\mathrm{x}'}(G) - \frac{1}{k^r}\|G\|_\infty  |B| k^{r-1} q^r.
	$$
	The original fractional $\na$-partition was arbitrary, therefore it follows that
	$$
	\EEE_\na (G) - \hat \EEE_\na (G) \leq \frac{1}{k}\|G\|_\infty 5^r q^{r+1}. 
	$$

\end{proof}

We are ready state the adaptation of \Cref{ch4:main} adapted to the microcanonical setting.

\begin{theorem} \label{ch4:micro}
	Let $r\geq 1$, $q \geq 1$,  $\na \in \Pd_q$, and $\varepsilon >0$. Then for any $[q]^r$-tuple of $([-d,d,r])$-graphons $W=(W^z)_{z \in [q]^r}$ and $k \geq \Theta^4  \log(\Theta) q^r$ with $\Theta=\frac{2^{r+7}q^r r}{\varepsilon}$  we have

	$$
	\PPP\left(|\EEE_\na(W)-\hat \EEE_\na(\G (k,W))|>\varepsilon \|W\|_\infty\right)<\varepsilon.
	$$
\end{theorem}
\begin{proof}
	Let $W$ be as in the statement and $k \geq \Theta^4  \log(\Theta) q^r$ with $\Theta=\frac{2^{r+7}q^r r}{\varepsilon}$.
	We start with pointing out that we are allowed to replace the quantity $\hat \EEE_\na(\G (k,W))$ by $\EEE_\na(\G (k,W))$ in the statement of the theorem by \Cref{ch4:intcont2} and only introduce an initial error at most  $ \frac{1}{k} \|G\|_\infty 5^r q^{r+1}\leq \frac{\varepsilon}{2}  \|W\|_\infty$.
	
	The lower bound on $\EEE_\na(\G (k,W))$ is the result of standard sampling argument combined with \Cref{ch4:encont}.
	Let us consider a fixed $\na$-partition $\phi$ of $[0,1]$, and define the random fractional partition of $[k]$ as $y_{n,m}=\phi_m(U_n)$ for every $n \in [k]$ and $m \in [q]$. The partition $y$ is not necessarily an fractional $\na$-partition, but it can not be very far from being one. For $m \in [q]$ it holds that
	$$
	\PPP\left( \left|\frac{\sum_{n=1}^{k}y_{n,m}}{k}-a_m \right|\geq \varepsilon \right) \leq 2 \exp(-\varepsilon^2k/2), 
	$$
	therefore for our choice of $k$ the sizes of the partition classes obey $|\frac{1}{k}\sum_{n=1}^{k}y_{n,m}-a_m|< \frac{\varepsilon}{2(q+1)}$ for every $m \in [q]$ with probability at least $1-\varepsilon/2$. 
	
	We appeal to \Cref{ch4:encont} to conclude
	\begin{align}
	\E \EEE_\na(\G(k,W)) &\geq \E  \EEE_y(\G(k,W))-(\varepsilon/2)\|W\|_\infty \nonumber \\
	&=\E \frac{1}{k^r}\sum_{z \in [q]^r} \sum_{n_1,\dots,n_r=1}^{k} W(U_{n_1}, \dots, U_{n_r}) \prod_{\substack{j=1}}^r y_{n_j,z_j} -(\varepsilon/2)\|W\|_\infty \nonumber \\
	&\geq \frac{k!}{k^r(k-r)!}\sum_{z \in [q]^r}  \int_{[0,1]^r}W(t_1, \dots, t_r)  \prod_{j=1}^r \phi_{z_j}(t_j) \du t - \left(\frac{r^2}{k} +\varepsilon/2\right)\|W\|_\infty \nonumber \\
	&\geq \EEE_\phi (W) - \left(\frac{r^2}{k} +\varepsilon/2\right)\|W\|_\infty. \nonumber
	\end{align}
	The concentration of the random variable $\EEE_\na(\G(k,W))$ can be obtained through martingale arguments identical to the technique used in the proof of the lower bound in \Cref{ch4:main}.
	
	For the upper bound on  $\EEE_\na(\G (k,W))$ we are going to use the cut decomposition and local linearization, the approach to approximate the energy of $\EEE_\phi(W)$ and $\EEE_{\mathrm x}(\G(k,W))$ for certain partitions $\phi$, respectively ${\mathrm x}$ is completely identical to the proof of \Cref{ch4:main}, therefore we borrow all the notation from there, and we do not refer to again in what follows.
	
	Now we consider a $b \in \mathcal A$ and define the event $E_3(b)$ that is occurrence the following implication.
	
	If the linear program
	\begin{align}
	&\textnormal{maximize} && l_0+\sum_{n=1}^{k}\sum_{m=1}^q\frac{1}{k} x_{n,m} l_m(U_n)  \nonumber \\
	&\textnormal{subject to } && {\mathrm x} \in  I'(b, \eta)\cap \omega_\na \nonumber \\
	&&& 0\leq x_{n,m} \leq 1 \quad &&\textrm{for $n=1, \dots,k$ and $m=1, \dots, q$} \nonumber \\
	&&&\sum_{m=1}^q x_{n,m}=1 &&\textrm{ for $n=1, \dots,k$} \nonumber
	\end{align}
	has optimal value $\alpha$, then the continuous linear program
	\begin{align}
	&\textnormal{maximize}&& l_0+\int_0^1 \sum_{m=1}^q l_m(t)\phi_m(t) \du t \nonumber \\
	&\textnormal{subject to}&& \phi \in I(b,\eta)  \cap \left(\bigcup_{\substack{\nc\colon  |a_i-c_i|\leq \eta }}\Omega_\nc\right) \nonumber \\
	&&& 0\leq \phi_m(t) \leq 1 \quad &&\textrm{ for $t \in [0,1]$ and $m=1, \dots, q$} \nonumber \\
	&&&\sum_{m=1}^q \phi_m(t)=1 \quad && \textrm{ for $t \in [0,1]$} \nonumber
	\end{align}
	has optimal value at least $\alpha- \frac{\varepsilon}{2} \|W\|_\infty$.
	
	Recall that $\eta = \frac{\varepsilon}{16  q^r 2^{r}}.$ It follows by applying \Cref{ch4:lpsample3} that $E_3(b)$ has probability at least $1-2\exp \left(-\frac{k\varepsilon^2}{2^{2r+12}q^{2r}r^2}\right)$. When conditioning on $E_1$, the event from the proof of \Cref{ch4:main},  and $E_ 3=\cap_{b \in \mathcal A} E_3(b)$ we conclude that 
	$$
	\EEE_\na(\G(k,W)) \leq \max_{\nc \colon  |a_i-c_i|\leq \eta } \EEE_\nc(W) + \varepsilon/2)\|W\|_\infty \leq \EEE_\na(W)+ (rq\eta+ \varepsilon/2)\|W\|_\infty\leq  \EEE_\na(W)+ \varepsilon\|W\|_\infty.
	$$
	Also, like in \Cref{ch4:main}, the probability of the required events to happen simultaneously is at least $1-\varepsilon/2.$ This concludes the proof.
\end{proof}

\subsection{Quadratic assignment and maximum acyclic subgraph problem}

The two optimization problems that are the subject of this subsection, the quadratic assignment problem\sindex[notions]{optimization problem!quadratic assignment problem} (QAP) and maximum acyclic subgraph problem\sindex[notions]{optimization problem!maximum acyclic subgraph problem} (AC), are known to be NP-hard, similarly to MAX-$r$CSP that was investigated above. The first polynomial time approximation schemes were designed for the QAP by Arora, Frieze and Kaplan ~\cite{AFK}. Dealing with a QAP means informally that one aims to minimize the transportation cost of his enterprise that has $n$ production locations and $n$ types of production facilities. This is to be achieved by an optimal assignment of the facilities to the locations with respect to the distances (dependent on the location) and traffic (dependent on the type of the production). In formal, terms this means that we are given two real quadratic matrices of the same size, $G$ and $J \in  \R^{n \times n}$, and the objective is to calculate
$$
\mathrm{Q}(G,J)=\frac{1}{n^2}\max_\rho \sum_{i,j=1}^n J_{i,j} G_{\rho(i),\rho(j)},
$$
where $\rho$ runs over all permutations of $[n]$. We speak of metric QAP\sindex[notions]{optimization problem!quadratic assignment problem!metric}, if the entries of $J$ are all non-negative with zeros on the diagonal, and obey the triangle inequality, and $d$-dimensional geometric QAP\sindex[notions]{optimization problem!quadratic assignment problem!$d$-dimensional geometric} if the rows and columns of $J$ can be embedded into a $d$-dimensional $L^p$ metric space so that distances of the images are equal to the entries of $J$.

The continuous analog of the problem is the following. Given the measurable functions $W, J\colon  [0,1]^2 \to \R$, we are interested in obtaining
$$
\hat{\mathrm{Q}}_\rho(W,J)=\int_{[0,1]^2} J(x,y)W(\rho(x),\rho(y)) \du x \du y, \qquad \hat{\mathrm{Q}}(W,J)=\max_\rho\hat{\mathrm{Q}}_\rho(W,J),
$$ \sindex[symbols]{q@$\hat{\mathrm{Q}}(W,J)$}
where $\rho$ in the previous formula runs over all measure preserving permutations of $[0,1]$.
In even greater generality we introduce the QAP with respect to fractional permutations of $[0,1]$. A fractional permutation $\mu$ is a probability kernel\sindex[notions]{probability kernel}, that is $\mu\colon [0,1] \times \mathcal L ([0,1]) \to [0,1]$ so that 
\begin{enumerate}[(i)]
	\item for any $A \in  \mathcal L([0,1])$ the function $\mu(.,A)$ is measurable, \item for any $x \in [0,1]$ the function $\mu(x,.)$ is a probability measure on $\mathcal L([0,1])$, and \item for any $A \in  \mathcal L([0,1])$ $\int_0^1 \du \mu(x,A)=\lambda(A)$. 	
\end{enumerate}
Here $\mathcal L([0,1])$ is the $\sigma$-algebra of the Borel sets of $[0,1]$.

Then we define  
$$
\mathrm{Q}_\mu(W,J)=\int_{[0,1]^2} \int_{[0,1]^2} J(\alpha,\beta) W(x,y) \du \mu(\alpha,x) \du \mu(\beta,y) \du \alpha \du \beta, 
$$
and
$$
\mathrm{Q}(W,J)=\max_\mu \mathrm{Q}_\mu(W,J),
$$
where the maximum runs over all fractional permutations. For each measure preserving permutation $\rho$ one can consider the fractional permutation $\mu$ with the probability measure $\mu(\alpha,.)$ is defined as the atomic measure $\delta_{\rho(\alpha)}$ concentrated on $\rho(\alpha)$, for this choice of $\mu$ we have $\mathrm{Q}_\rho(W,J)=\mathrm{Q}_\mu(W,J)$.

An $r$-dimensional generalization of the problem for $J$ and $W\colon  [0,1]^r \to \R$ is
\begin{align}
\mathrm{Q}(W,J) 
=\max_\mu \int_{[0,1]^r} \int_{[0,1]^r} J(\alpha_1,\dots,\alpha_r) W(x_1,\dots,x_r) \du \mu(\alpha_1,x_1) \dots \du \mu(\alpha_r,x_r) \du \alpha_1 \dots \du \alpha_r, \nonumber
\end{align}
where the maximum runs over all fractional permutations $\mu$ of $[0,1]$. The definition of the finitary case in $r$ dimensions is analogous.

A special QAP is the maximum acyclic subgraph problem (AC). Here we are given a weighted directed graph $G$ with vertex set of cardinality $n$, and our aim is to determine the maximum of the total value of edge weights of a subgraph of $G$ that contains no directed cycle. We can formalize this as follows. Let $G \in \R^{n \times n}$ be the input data, then the maximum acyclic subgraph density is
$$
\mathrm{AC}(G)=\frac{1}{n^2} \max_\rho \sum_{i,j=1}^n G_{i,j} \I(\rho(i) < \rho(j)),
$$\sindex[symbols]{a@$\mathrm{AC}(G)$}
where $\rho$ runs over all permutations of $[n]$.

This can be thought of as a QAP with the restriction that $J$ is the upper triangular $n \times n$ matrix with zeros on the diagonal and all nonzero entries being equal to $1$. However in general AC cannot be reformulated as metric QAP. The continuous version of the problem 
\begin{align*}
\hat{\mathrm{AC}}(W)=\sup_{\phi} \int_{[0,1^2]} \I(\phi(x) > \phi(y)) W(x,y) \du x\du y
\end{align*} for a function $W\colon [0,1]^2 \to \R$ is defined analogous to the QAP, where the supremum runs over measure preserving permutations $\phi\colon [0,1] \to [0,1]$, as well as the relaxation $\mathrm{AC}(W)$, where the supremum runs over probability kernels.

Both the QAP and the AC problems resemble the ground state energy problems that were investigated in previous parts of this paper. In fact, if the number of clusters of the distance matrix $J$ in the QAP is bounded from above by an integer that is independent from $n$, then this special QAP is a ground state energy with the number of states $q$ equal to the number of clusters of $J$. By the number of clusters we mean here the smallest number $m$ such that there exists an $m \times m$ matrix $J'$  so that $J$ is a blow-up of $J'$, that is not necessarily equitable. To establish an approximation to the solution of the QAP  we will only need the cluster condition approximately, and this will be shown in what follows.

\begin{definition}
	We call a  measurable function $J\colon [0,1]^r \to \R$ $\nu $-clustered \sindex[notions]{clustered naive graphon}for a non-increasing function $\nu\colon \R^+ \to \R^+$, if for any $\varepsilon>0$ there exists another measurable function $J'\colon [0,1]^r \to \R$ that is a step function with $\nu(\varepsilon)$ steps  and $\|J-J'\|_1 < \varepsilon \|J\|_\infty$.
\end{definition}
Note, that by the Weak Regularity Lemma (\cite{FK}), \Cref{ch3:weakreglemma}, any $J$ can be well approximated by a step function with $\nu(\varepsilon)=2^{\frac{1}{\varepsilon^2}}$ steps in the cut norm. To see why it is likely that this approximation will not be sufficient for our purposes, consider an arbitrary $J\colon [0,1]^r \to \R$. Suppose that we have an approximation in the cut norm of $J$ at hand denoted by $J'$. Define the probability kernel $\mu_0(\alpha,.)=\delta_\alpha$ and the naive $r$-kernel $W_0=J-J'$. In this case
$|\mathrm{Q}_{\mu_0}(W,J)-\mathrm{Q}_{\mu_0}(W,J')|=\|J-J'\|_2^2$. This $2$-norm is not granted to be small compared to $\varepsilon$ by any means.

In some special cases, for example if $J$ is a $d$-dimensional geometric array or the array corresponding to the AC, we are able to require bounds on the number of steps required for the $1$-norm approximation of $J$ that are sub-exponential in $\frac{1}{\varepsilon}$. By the aid of this fact we can achieve good approximation of the optimal value of the QAP via sampling. Next we state an application of \Cref{ch4:main} to the clustered QAP.
\begin{lemma} \label{ch4:qap}
	Let  $\nu: \R^+ \to \R^+$ be nondecreasing, and let $J\colon [0,1]^r \to \R$ be a $\nu$-clustered measurable function. Then  there exists an absolute constant $c>0$ so that for every $\varepsilon >0$, every naive $r$-kernel $W$, and $k \geq c \log(\frac{\nu(\varepsilon)^r}{\varepsilon})(\frac{\nu(\varepsilon)^r}{\varepsilon})^{4}$ we have
	$$
	\PPP(|\mathrm{Q}(W,J)-\mathrm{Q}(\G(k,W),\G'(k,J))| \geq \varepsilon \|W\|_\infty \|J\|_\infty) \leq  \varepsilon,
	$$
	where $\G(k,W)$ and $\G'(k,J)$ are generated by distinct independent samples.
\end{lemma}
\begin{proof}
	Without loss of generality we may assume that $\|J\|_\infty \leq 1$. First we show that under the cluster condition we can introduce a microcanonical ground state energy problem whose optimum is close to $\mathrm{Q}(W,J)$, and the same holds for the sampled problem. Let $\varepsilon>0$ be arbitrary and $J'$ be an approximating step function with $q=\nu(\varepsilon)$ steps. We may assume that $\|J'\|_\infty\leq 1$. We set $\na=(a_1, \dots , a_q)$ to be the vector of the sizes of the steps of $J'$, and construct from $J'$ a real $r$-array of size $q$ in the natural way by 
	associating to each class of the steps of $J'$ an element of $[q]$ (indexes should respect $\na$), and set the entries of the $r$-array corresponding to the value of the respective step of $J'$. We will call  the resulting $r$-array $J''$. From the definitions it follows that
	$$
	\mathrm{Q}(W,J')=\EEE_\na(W,J'')
	$$
	for every $r$-kernel $W$.
	On the other hand we have 
	\begin{align*}
	|\mathrm{Q}(W,J)-&\mathrm{Q}(W,J')| \\ &\leq \max_\mu \left|\mathrm{Q}_\mu(W,J)-\mathrm{Q}_\mu(W,J')\right|  \\
	&=\max_\mu \left|\int_{[0,1]^r} (J-J')(\alpha_1, \dots, \alpha_r) \int_{[0,1]^r} W(x) \du \mu(\alpha_1,x_1) \dots \du \mu(\alpha_r,x_r) \du \alpha_1 \dots \du \alpha_r \right| \\
	&\leq \max_\mu \int_{[0,1]^r} |(J-J')(\alpha_1, \dots, \alpha_r)| \|W\|_\infty \du \alpha_1 \dots \du \alpha_r \\ &=\|J-J'\|_1 \|W\|_\infty \leq \varepsilon \|W\|_\infty.
	\end{align*}
	
	Now we proceed to the sampled version of the optimization problem. First we gain control over  the difference between the QAPs corresponding to $J$ and $J'$. $\G(k,W)$ is induced by the sample $U_1,\dots,U_k$, $\G'(k,J)$ and $\G'(k,J')$ by the distinct independent sample $Y_1, \dots,Y_k$.
	\begin{align}
	&|\mathrm{Q}(\G(k,W),\G(k,J))-\mathrm{Q}(\G(k,W),\G(k,J'))| \nonumber \\ & \qquad \leq \max_\rho |\mathrm{Q}_\rho(\G(k,W),\G(k,J))-\mathrm{Q}_\rho(\G(k,W),\G(k,J'))| \nonumber \\
	&\qquad =\max_\rho \frac{1}{k^r}\left|\sum_{i_1, \dots,i_r=1}^k (J-J')(Y_{i_1},\dots,Y_{i_r})\right| \|W\|_\infty. \label{ch4:qapeq}
	\end{align}
	We analyze the random sum on the right hand side of (\ref{ch4:qapeq}) by first upper bounding its expectation.
	\begin{align*}
	&\frac{1}{k^r}\E_Y \left|\sum_{i_1, \dots,i_r=1}^k (J-J')(Y_{i_1},\dots,Y_{i_r}) \right| \\
	&\qquad \leq \frac{r^2}{k} \left[\|J\|_\infty+\|J'\|_\infty\right] + \E_Y |(J-J')(Y_1,\dots,Y_r)| \\ &\qquad =\frac{2r^2}{k}+ \varepsilon \leq 2\varepsilon.
	\end{align*}
	By the Azuma-Hoeffding inequality the sum is also sufficiently small in probability.
	$$
	\PPP\left(\frac{1}{k^r} \left|\sum_{i_1, \dots,i_r=1}^k (J-J')(Y_{i_1},\dots,Y_{i_r}) \right| \geq 4 \varepsilon \right) \leq 2\exp(-\varepsilon^2k/8)\leq \varepsilon.
	$$
	We obtain that
	$$
	\left|\mathrm{Q}(\G(k,W),\G'(k,J))-\mathrm{Q}(\G(k,W),\G'(k,J'))\right| \leq 4\varepsilon \|W\|_\infty
	$$
	with probability at least $1-\varepsilon$, if $k$ is such as in the statement of the lemma. Set $\nb=(b_1,\dots, b_q)$ to be the probability distribution for that $b_i=\frac{1}{k}\sum_{j=1}^k \I_{S_i}(Y_j)$, where $S_i$ is the $i$th step of $J'$ with $\lambda(S_i)=a_i$. Then we have 
	$$
	\mathrm{Q}(\G(k,W),\G(k,J'))=\hat \EEE_\nb(\G(k,W),J'').
	$$
	It follows again from the Azuma-Hoeffding inequality that we have $\PPP(|a_i-b_i|>\varepsilon/q) \leq 2 \exp(-\frac{\varepsilon^2k}{2q^2})$ for each $i \in [q]$, thus we have 
	$\|\na-\nb\|_1<\varepsilon$ with probability at least $1-\varepsilon$. We can conclude that with probability at least $1-2\varepsilon$ we have
	\begin{align*}
	&|\mathrm{Q}(W,J)-\mathrm{Q}(\G(k,W),\G(k,J))|   \leq |\mathrm{Q}(W,J)-\mathrm{Q}(W,J')|  \\ &\qquad \qquad  + |\EEE_\na(W,J'') -\hat \EEE_\nb(\G(k,W),J'')|  +|\mathrm{Q}(\G(k,W),\G(k,J))-\mathrm{Q}(\G(k,W),\G(k,J'))| \\
	& \qquad \qquad \qquad \qquad \qquad \qquad \qquad \leq (5+2r) \varepsilon \|W\|_\infty + |\EEE_\na(W,J'') -\hat \EEE_\na(\G(k,W),J'')|.
	\end{align*}
	By the application of \Cref{ch4:micro} the claim of the lemma is verified.
	
\end{proof}
Next we present the application of \Cref{ch4:qap} for two special cases of QAP. 
\begin{corollary} \label{ch4:qap2}
	The optimal values of the $d$-dimensional geometric QAP and the maximum acyclic subgraph problem are efficiently testable. That is, let $d\geq 1$, for every $\varepsilon>0$ there exists an integer $k_0=k_0(\varepsilon)$ such that $k_0$ is a polynomial in $1/\varepsilon$, and for every $k \geq  k_0$ and any $d$-dimensional geometric QAP given by the pair $(G,J)$ we have
	\begin{align}
	\PPP(|\mathrm{Q}(G,J)-\mathrm{Q}(\G(k,G),\G'(k,J))| \geq \varepsilon \|G\|_\infty \|J\|_\infty) \leq  \varepsilon,\end{align}
	where $\G(k,W)$ and $\G'(k,J)$ are generated by distinct independent samples. 
	The formulation regarding the testability of the maximum acyclic subgraph problem is analogous.
\end{corollary}
Note that testability here is meant in the sense of the statement of \Cref{ch4:qap}, since the size of $J$ is not fixed and depends on $G$.
\begin{proof}
	In the light of \Cref{ch4:qap} it suffices to show that for both cases any feasible $J$ is  $\nu$-clustered, where $\nu(\varepsilon)$ is polynomial in $1/\varepsilon$. For both settings we have $r=2$.
	
	We start with the continuous version of the $d$-dimensional geometric QAP given by the measurable function $J\colon [0,1]^2 \to \R^+$, and an instance is given by the pair $(W,J)$, where $W$ is a $2$-kernel. Note, that $d$ refers to the dimension corresponding to the embedding of the indices of $J$ into an $L^p$ metric space, not the actual dimension of $J$. We are free to assume that $0 \leq J \leq 1$, simply by rescaling. By definition, there exists a measurable embedding $\rho\colon  [0,1] \to [0,1]^d$, so that $J(i,j)=\|\rho(i)-\rho(j)\|_p$ for every $(i,j) \in [0,1]^2$. Fix $\varepsilon >0$ and consider the partition $\P'=(T_1,\dots, T_\beta)=([0,\frac{1}{\beta}),[\frac{1}{\beta}, \frac{2}{\beta}), \dots,[\frac{\beta-1}{\beta},1])$ of the unit interval into $\beta=\lceil \frac{2 \sqrt[p]{d}}{\varepsilon}\rceil$ classes. Define the partition $\P=(P_1, \dots, P_q)$ of $[0,1]$ consisting of the classes $\rho^{-1}(T_{i_1} \times \dots \times T_{i_d})$ for each $(i_1, \dots, i_d) \in [\beta]^r$, where $|\P|=q=\beta^d=\frac{2^d d^{d/p}}{\varepsilon^d}$. We construct the approximating step function $J'$ of $J$ by averaging $J$ on the steps determined by the partition classes of $\P$. It remains to show that this indeed is a sufficient approximation in the $L^1$-norm.
	$$
	\|J-J'\|_1=\int_{[0,1]^2} |J(x)-J'(x)|\du x= \sum_{i, j=1}^q \int_{P_{i} \times P_{j}} |J(x)-J'(x)|\du x \leq \sum_{i, j}^q \frac{1}{q^2} \varepsilon = \varepsilon.
	$$
	By \Cref{ch4:qap} and \Cref{ch4:micro} it follows that the continuous $d$-dimensional metric QAP is $\mathcal{O}(\log(\frac{1}{\varepsilon})\frac{1}{\varepsilon^{4rd+4}})$-testable, and so is the discrete version of it.
	
	Next we show that the AC is also efficiently testable given by the upper triangular matrix $J$ whose entries above the diagonal are $1$. Note that here we have $r=2$. Fix $\varepsilon>0$ and consider the partition $\P=(P_1, \dots, P_q)$ with $q=\frac{1}{\varepsilon}$, and set $J'$ to $0$ on every step $P_i \times P_j$ whenever $i \geq j$, and to $1$ otherwise. This function is indeed approximating $J$ in the $L^1$-norm.
	$$
	\|J-J'\|_1=\int_{[0,1]^2} |J(x)-J'(x)|\du x= \sum_{i=1}^q \int_{P_{i} \times P_{i}} |J(x)-J'(x)|\du x \leq  \varepsilon.
	$$
	Again, by \Cref{ch4:qap} and \Cref{ch4:micro} it follows that the AC is $\mathcal{O}(\log(\frac{1}{\varepsilon})\frac{1}{\varepsilon^{12}})$-testable.
	
\end{proof}

\section{Further Research}\label{sec:fr}
Our framework based on exchangeability principles allows us to extend the notion of a limit to the case of unbounded hypergraphs and efficient testability of ground state energies in this setting.
The notion of exchangeability is crucial here.
The notion of efficient testability in an unbounded case could be of independent interest, perhaps the results on ground state energy carry through for the setting when the $r$-graphons (induced by $r$-graphs) are in an $L^p$ space for some $p \geq 1$. \\

Another problem is to characterize  more precisely the class of problems which are efficiently parameter testable as opposed to the hard ones. Improving the bounds in $1/\varepsilon$ for the efficiently testable problems is also a worthwhile question.

\section*{Acknowledgement}

We thank Jennifer Chayes, Christian Borgs and Tim Austin for a number of interesting and stimulating discussions and the relevant new ideas in the early stages of this research.

\bibliographystyle{notplainnat}
\bibliography{thesis_refs}

\end{document}